\documentclass[acmsmall,10pt,screen]{acmart}
\usepackage{microtype}
\pdfoutput=1

\acmConference[PL'18]{ACM SIGPLAN Conference on Programming Languages}{January 01--03, 2018}{New York, NY, USA}
\acmYear{2018}
\acmISBN{} 
\acmDOI{} 
\startPage{1}

\setcopyright{none}

\usepackage{xy}
\xyoption{all}
\usepackage{pgf}
\usepackage{bm}
\usepackage[inline]{enumitem}
\usepackage{mathpartir}
    
    \makeatletter 
    \def\arcr{\@arraycr}
    \makeatother
\usepackage{mathtools}
\usepackage{subcaption}
\usepackage{wrapfig}
\usepackage{thmtools}

\hypersetup{final}

\usepackage{tikz}
  \usetikzlibrary{calc,decorations.pathreplacing,
    decorations.pathmorphing,
    shapes, fit}
  \usetikzlibrary{cd}
  \usetikzlibrary{positioning}
  \tikzset{
    role/.style={
      draw,
      rectangle,
      thick,
      minimum width=1cm,
      minimum height=.4cm}
  }
\usepackage{pgfplots}
\usepackage{listings}
\usepackage{amsmath}
\usepackage{booktabs}

\usepackage{definitions}

\usepackage{etoolbox} 
\apptocmd{\sloppy}{\hbadness 10000\relax}{}{}

\usepackage{anyfontsize} 
\usepackage{silence}

\SetSymbolFont{stmry}{bold}{U}{stmry}{m}{n}
\begin{document}
\title{\CAMP: Cost-Aware Multiparty Session Protocols}

\author{David Castro-Perez}
\email{d.castro-perez@imperial.ac.uk}
\orcid{0000-0002-6939-4189}
\author{Nobuko Yoshida}
\email{n.yoshida@imperial.ac.uk}
\orcid{0000−0002−3925−8557}
\affiliation{%
  \institution{Imperial College London}
  \department{Computing}
  \streetaddress{180 Queen's Gate}
  \city{London}
  \state{}
  \postcode{2AZ SW11}
\country{United Kingdom}
}

\begin{abstract}
This paper presents \CAMP, a new static performance analysis framework for
message-passing concurrent and distributed systems, based on the theory of
\emph{multiparty session types} (\MPST). Understanding the run-time
performance of concurrent and distributed systems is of great importance for
the identification of bottlenecks and optimisation opportunities. In the
message-passing setting, these bottlenecks are generally \emph{communication
overheads} and \emph{synchronisation times}. Despite its importance,
reasoning about these \emph{intensional} properties of software, such as
performance, has received little attention, compared to verifying
\emph{extensional} properties, such as correctness. Behavioural protocol
specifications based on sessions types capture not only extensional, but also
intensional properties of concurrent and distributed systems. \CAMP{}
augments \MPST{} with annotations of \emph{communication latency} and
\emph{local computation cost}, defined as estimated execution times, that we
use to extract cost equations from protocol descriptions. \CAMP{} is also
extendable to analyse \emph{asynchronous communication optimisation} built on
a recent advance of session type theories. We apply our tool to different
existing benchmarks and use cases in the literature with a wide range of
communication protocols, implemented in C, MPI-C, Scala, Go, and
OCaml. Our benchmarks show that, in most of the cases, we predict an
upper-bound on the real execution costs with < 15\% error.
\end{abstract}

%




\maketitle

\section{Introduction}
\label{sec:intro}
Understanding the amount of resources, e.g.\ time or memory that are required
by a computation, is of great importance. Correct but slow-performing
software can cause a number of problems, ranging from the unnecessary use of
resources, to exploitable security vulnerabilities. Worse still, performance
issues are very difficult to detect in runtime because of their non fail-stop
nature; and although the root causes of performance bugs can be very diverse,
\emph{uncoordinated functions} and \emph{synchronisation issues} are
prevalent, i.e.\ inefficient composition of efficient functions, and
unnecessary synchronisation that increases thread competition
~\cite{DBLP:conf/pldi/JinSSSL12}. These inefficient compositions have more
impact in a distributed setting, where the \emph{communication overhead} and
\emph{synchronisation cost} may become the bottleneck of the whole system.

The development of new static performance analysis tools will reduce the
impact of bad performing software, by allowing the identification of their
bottlenecks and optimisation. Further, for concurrency and distribution, such
a tool must take into account communication and synchronisation overheads.
This paper presents a new static performance analysis framework, \CAMP{}
(\emph{Cost-Aware Multiparty Protocols}), that can identify potential
performance bottlenecks in concurrent and distributed systems. Specifically,
\CAMP{} addresses the following two main challenges: \emph{non-determinism}
and \emph{practicality}. Firstly, the non-deterministic nature of concurrent
and distributed systems makes it hard to reason statically about the
performance of alternative interleavings of actions in a program trace; and
secondly, making the performance analysis practically useful for already
existing implementations is not trivial.

\CAMP{} solves the non-determinism issue by building on top of
\emph{multiparty session types} (\MPST) \cite{Honda2008Multiparty,CDPY2015}.
\MPST{} is a well-established theory that describes not only
\emph{extensional}, but also \emph{intensional} information about
communicating systems. Specifically, \MPST{} captures the communication
structure, or \emph{protocol} among distributed peers. Protocols appear not
only in distributed networks but also in parallel multicore programming as
\emph{patterns} or \emph{topologies}
\cite{Rauber10,Taubenfeld06,Lea97,Goetz06}.
\MPST{} uses \emph{global types} for describing such
protocols from a global point of view, and can be used to ensure
\emph{deadlock-freedom} and \emph{session fidelity}: every send has a
matching receive, and every component of the concurrent/distributed system
complies with its part in the global protocol. Built on the \MPST{} theory,
\CAMP{} enables the protocol-based performance analysis, giving a precise
abstraction as (correct) communication structures of programs. By tying the
analysis to a particular protocol specification that is statically enforced
on the concurrent/distributed system, \CAMP{} solves the issue of
non-determinism.

On the practical side, since all we require is a global type, \CAMP{} can be
readily applied to existing implementations, as long as they are proven to
comply with a known global type. We show this by taking existing benchmarks,
either implemented using \MPST{}-based tools, or following a known protocol.
Different extensions of the core \MPST{} have been already used to implement
a wide range of applications written in different programming languages
through several transports and architectures e.g.
\cite{CHJNY:2019,HY2017,NCY2015,bettytoolbook,CY:2020, INY19}, and our
methodology is easily adaptable to these variants.
In addition, not only \CAMP{} is immediately usable
for analysis of representative parallel
patters \cite{dwarfs-cacm,Krommydas16,Rauber10}, but also
it is applicable to Savina benchmarks \cite{IS14}
or multicore algorithms
which incur more complex patterns and synchronisations (\S\ref{sec:benchstructure}).

The key notion in \CAMP{} is that of execution \emph{cost}: the amount of
time that it takes a protocol, participant or function to run from beginning
to end. To statically compute execution costs for concurrency and
distribution, \CAMP{} extends global types with \emph{sizes} for values of
messages (encoded in the payload types) and \emph{local computation cost}
information. This size and cost information can be obtained via profiling, or
further static analysis, such as using \emph{sized-types}
\cite{DBLP:conf/popl/HughesPS96}. Our cost models take these extended
protocols, and compute a set of equations which describe the total cost of
each participant. These measurements provide us with
fine-grained information to obtain communication overhead and synchronisation
cost among participants of a protocol. For recursive protocols, \CAMP{}
produces a set of \emph{recurrence equations} that describe the total cost
after each iteration of the protocol. For non-terminating protocols (e.g.\
streaming computation split in multiple stages), \CAMP{} computes the
\emph{latency}, or the average cost per iteration.

\CAMP{} enables to quantify the performance gain of \emph{asynchronous
communication optimisation}. 
We evaluate this using non-optimised and optimised benchmarks.
The optimisation
analysis by \CAMP{} is grounded on \emph{asynchronous session subtyping},
which is one of the most advanced session types theories in the literature, and has been
actively studied over a decade using various different formalisms, e.g.,
\emph{first and higher-order mobile processes}
\cite{mostrous_yoshida_honda_esop09,MostrousY15,DBLP:conf/tlca/MostrousY09,CDSY2017,cdy14,GPPSY2021},
\emph{denotational semantics} \cite{Dezani16,DemangeonY15} and \emph{automata
  theories} \cite{BravettiCZ17,BravettiCZ18,LY2017,BCLYZ2019}.

\textsl{\textbf{Contributions.}} \ We present a compile-time performance
analysis framework, \CAMP{}, for concurrent/distributed systems that infers
upper execution cost bounds of multiparty session protocols. The cost models
in \CAMP{} are parametric, and can combine both static and dynamic (e.g.,
profiling) information to produce accurate results. We prove that the cost
analysis by \CAMP{} is sound with respect to the operational semantics of a
given global type instrumented with sizes and execution costs; and
extensible to analyse communication optimisations. Our main
contributions are:
\begin{enumerate}[label=\alph*), labelindent=0em, leftmargin=*]
\item we define the semantics of \CAMP{}, integrating \emph{global} and
\emph{local type semantics} with \emph{local computation costs}, that can be
used to explore the costs of particular traces (\S\ref{sec:mpst});
\item we instrument \emph{global types} with size and local cost information,
and use it to statically estimate an upper bound of the execution cost of a
protocol, that we prove sound with respect to the operational semantics
(\S\ref{sec:cost});
\item we define multiple metrics on the \emph{cost recurrences} associated with
recursive protocols, that can be used to effectively analyse the
performance behaviour of potentially infinite executions (\S\ref{sec:thro});
\item we extend \CAMP{} to handle \emph{asynchronous message optimisations},
enabling us to statically quantify the potential performance gains when
performing such optimisations/reordering (\S\ref{sec:extn});
\item we implement a DSL for specifying global types, from which we can
extract cost equations (\S\ref{sec:impl}), and we compare our cost model
predictions with real benchmarks used in \MPST{} implementations in
different languages: C-MPI \cite{NCY2015}, C+pthreads \cite{CY:2020}, Go
\cite{CHJNY:2019}, \OCaml{} \cite{INY19} and F$\star$ \cite{OOPSLA20FStar}.
Additionally, we apply \CAMP{} to a subset of the Savina benchmarking suite
(Scala) \cite{IS14}. These benchmarks include examples of common, and complex
topologies, such as ring, butterfly and a double-buffering protocol
(\S\ref{sec:eval}).
\end{enumerate}
\S\ref{sec:relw} discusses related work and \S\ref{sec:concl}
concludes the paper.
\textbf{\textsl{Appendix}}
includes additional definitions and
full proofs.
The anonymised git repository \url{https://github.com/camp-cost/camp}
provides a working
prototype implementation, described in \S\ref{sec:impl} and the data used in
\S\ref{sec:eval}, with
instructions for replicating our experiments. We will also submit
it as an  {\bf \textsl{artifact}}.

\section{Overview}
\label{sec:overview}
\begin{wrapfigure}{l}{.55\columnwidth}
  \centering
  \vspace{-5mm}
	\begin{tikzpicture}[x=1.5cm, y=-1.25cm, align=center, font=\footnotesize, draw, minimum width=.375cm, minimum height=.375cm, inner sep=0pt, outer sep=2pt]
		\node [draw, inner sep=1mm] (G) at (0,0) {$G$};
		\node [ellipse, draw, inner sep=.5mm] (cG) at (-1.5,0) {$\cost$};
		\node [draw, inner sep=1mm] (L1) at (-.75,1) {$L_1$};
		\node [draw, inner sep=1mm] (L2) at (-.25,1) {$L_2$};
		\node [] (LL) at (.25,1) {\strut...};
		\node [draw, inner sep=1mm] (Ln) at (.75,1) {$L_n$};
		\node [draw, inner sep=1mm] (P1) at (-.75,2) {$P_1$};
		\node [draw, inner sep=1mm] (P2) at (-.25,2) {$P_2$};
		\node [] (PP) at (.25,2) {\strut...};
		\node [draw, inner sep=1mm] (Pn) at (.75,2) {$P_n$};
        \draw [->] (G.west) to (cG.east);
        \draw [dashed, -] (cG.south) to ($(cG.south) + (0,2.3)$);
        \draw [dashed, -] ($(cG.south) + (0,2.3)$) -| (Pn.south);
        \draw [dashed, -] ($(cG.south) + (0,2.3)$) -| (P2.south);
        \draw [dashed, -] ($(cG.south) + (0,2.3)$) -| (P1.south);
		\draw [->] (G.south) to (L1);
		\draw [->] (G.south) to (L2);
		\draw [->, dashed] (G.south) to (LL);
		\draw [->] (G.south) to (Ln);
		\draw [->] (L1) to (P1);
		\draw [->] (L2) to (P2);
		\draw [->, dashed] (LL) to (PP);
		\draw [->] (Ln) to (Pn);
		\draw [draw=none, xshift=0pt, yshift=0pt] (-2.3,1) to node [anchor=west, text width=1cm, font=\scriptsize, align=center] {cost \emph{approximates} the execution time of each $P_i$} (-2.3,1.75);
		\draw [draw=none, xshift=0pt, yshift=0pt] (-1.75,-.00) to node [anchor=west, text width=3cm, font=\scriptsize, align=center] {\emph{estimate} cost of each role in $G$} (-1.75,-.75);
		\draw [draw=none, xshift=0pt, yshift=0pt] (1,.25) to node [anchor=west, text width=2cm, font=\scriptsize, align=center] {\emph{project} global type $G$ onto each role} (1,.75);
		\draw [draw=none, xshift=0pt, yshift=0pt] (1,1.25) to node [anchor=west, text width=2cm, font=\scriptsize, align=center] {\emph{type-check} each process $P_i$ against local type $L_i$} (1,1.75);
	\end{tikzpicture}
	\caption{\CAMP{} framework}
\vspace{-8mm}
\label{fig:mpst}
\end{wrapfigure}

\myparagraph{\MPST{} Basics.}
We first explain how \MPST{} satisfies extensional properties.
Fig.\ \ref{fig:mpst} depicts the standard top-down methodology
of \MPST{} enhanced with cost-analysis, which we illustrate by a
simple \emph{scatter-gather} example between two Masters
($\roleM_1$, $\roleM_2$) and two Workers ($\roleW_1$,
$\roleW_2$).

\vskip.2cm
{\small
  \begin{center}
    $
    \begin{array}{l @{\;} l}
      G =
      &
      \roleM_{1} \gMsg \roleW_1 \gTy{\tau_1}.
      \roleM_1 \gMsg \roleW_2 \gTy{\tau_1}.
      \\
      &
      \roleW_1 \gMsg \roleM_2 \gTy{\tau_2}.
      \roleW_2 \gMsg \roleM_2 \gTy{\tau_2}
    \end{array}
    $
  \end{center}
  }
\vskip.2cm
\noindent $G$ is a \emph{global type}: a specification of the protocol
between participants from a global perspective. $G$ says
master $\roleM_1$ first sends a message with type $\tau_1$ to
worker $\roleW_1$ then to worker $\roleW_2$,
and finally master $\roleM_2$ collects
a message with type $\tau_2$ from each worker.
For each participant $\Rp$, the global type is projected to a \emph{local type}, which
describes localised send and receive actions from $\Rp$ viewpoint:
\[
\begin{array}{c}
  L_1 = \roleW_1\tSend\tau_1. \roleW_2\tSend\tau_1.\lEnd
  \qquad
  L_2 = \roleM_1\tRecv\tau_1. \roleM_2\tSend\tau_2.\lEnd
  \qquad L_3 = \roleW_1\tRecv\tau_2. \roleW_2\tRecv\tau_2.\lEnd
\end{array}
\]
$L_1$ says $\roleM_1$ should send ($\tSend$) a
$\tau_1$ message to $\roleW_1$, then to $\roleW_2$, while
$L_2$ says $\roleW_1$ receives ($\tRecv$) a $\tau_1$ message
from  $\roleM_1$, followed by sending a $\tau_2$ message
to $\roleM_2$ (worker $\roleW_2$ has the same type $L_2$).
Local types are used to statically check local programs $P_i$
implementing these types, i.e. the communication structures
of each program complies with their local type. A well-typed system of
programs is guaranteed free from deadlock and type
errors, following the protocol given by $G$ (session fidelity).

\myparagraph{Cost-Aware \MPST{}}
Now we consider the cost to run $G$:
{\small
\[
\roleM_1 \gMsg \roleW_1 \gTy{\tau_1 \hasCost \ccc_1}.
  \roleM_1 \gMsg \roleW_2 \gTy{\tau_1 \hasCost \ccc_1}.
  \roleW_1 \gMsg \roleM_2 \gTy{\tau_2}.
  \roleW_2 \gMsg \roleM_2 \gTy{\tau_2}
\]
}
Here $\ccc_1$ represents the \emph{local computation cost} at the
receiver side. In this example, we are assuming
the computational cost at $\roleW_1$ and $\roleW_2$ is
$\ccc_1$,  while such cost at $\roleM_2$ is $0$.
Another factor we should take into account
is the \emph{communication cost}, which is parameterised
by types, i.e. the time required for sending ($\csend(\tau)$) and
receiving ($\crecv(\tau)$)
a value of type $\tau$.

We assume our transport is
\emph{asynchronous}, i.e.~sending is non-blocking
and the order of messages are preserved, like TCP communications,
hence there is no communication ordering between $\roleW_1$
and $\roleW_2$. Both workers can process the values
independently at two different locations or CPUs.
Then, the total execution cost at
$\roleM_1$, $\roleW_1$ and $\roleW_2$ are:
\[
  \begin{array}{@{}c@{}}
    \roleM_1\mapsto 2\times \csend(\tau_1)
    \quad
    \roleW_1\mapsto \csend(\tau_1) + \crecv(\tau_1) + \ccc_1 + \csend(\tau_2)
    \quad
    \roleW_2\mapsto 2 \times \csend(\tau_1) + \crecv(\tau_1) + \ccc_1 + \csend(\tau_2)
  \end{array}
\]
To consider the cost for $\roleM_2$, we should take care of the dependencies
in the protocol. Each $\roleW_i$ can operate in parallel, and they exhibit
almost the same cost. The only difference is that worker $\roleW_1$ can
perform its computation as soon as $\roleM_1$ sends one message, but worker
$\roleW_2$ can only proceed after $\roleM_1$ sends the second message. This
difference means that $\roleM_2$ can start gathering one of the messages,
while the other worker finishes its actions, which will be delayed by the time
it takes $\roleM_1$ to send one message. Hence the cost of $\roleW_2$
is:
\[
       \roleM_2\mapsto
        \max(\crecv(\tau_2), \csend(\tau_1))
        + \csend(\tau_1) + \crecv(\tau_1) + \ccc_1 + \csend(\tau_2) + \crecv(\tau_2)
\]
In \S\ref{sec:cost}, we shall prove the cost calculated based on local types
and global types semantics coincide.

In many scenarios, we do not know how many iterations recursive protocols are
going to run, or this number of iterations is large. In such cases, computing
the cost of the protocol is not useful or meaningful. In such scenarios, we
calculate the average cost per iteration of a protocol (\emph{latency}) from a
set of recurrences. From this latency, we calculate other useful metrics, such
as the latency divided by the number of messages exchanged per iteration by
participant (\emph{latency relative to a particular participant}). The latency
relative to a participant is used to estimate how much work can a participant do
per iteration of the protocol.

\section{Cost-Aware Multiparty Session Protocols}
\label{sec:mpst}

\begin{wrapfigure}{r}{.5\columnwidth}
  \begin{center}
%
%
  \begin{tikzpicture}[every node/.style={font=\small,
                                         minimum height=0.42cm,
                                         minimum width=1.2cm}]

   \node [matrix, very thin,column sep=0.4cm,row sep=0.1cm] (matrix) at (0,0) {
    &
    \node(0,0) (stage10) {}; & &
    \node(0,0) (stage20) {}; & &
    \node(0,0) (stage30) {}; & \\
    &
    \node(0,0) (stage11) {}; & &
    \node(0,0) (stage21) {}; & &
    \node(0,0) (stage31) {}; & \\
    &
    \node(0,0) (stage12) {}; & &
    \node(0,0) (stage22) {}; & &
    \node(0,0) (stage32) {}; & \\
    &
    \node(0,0) (stage13) {}; & &
    \node(0,0) (stage23) {}; & &
    \node(0,0) (stage33) {}; & \\
    &
    \node(0,0) (stage14) {}; & &
    \node(0,0) (stage24) {}; & &
    \node(0,0) (stage34) {}; & \\
    &
    \node(0,0) (stage15) {}; & &
    \node(0,0) (stage25) {}; & &
    \node(0,0) (stage35) {}; & \\
    &
    \node(0,0) (stage16) {}; & &
    \node(0,0) (stage26) {}; & &
    \node(0,0) (stage36) {}; & \\
    &
    \node(0,0) (stage17) {}; & &
    \node(0,0) (stage27) {}; & &
    \node(0,0) (stage37) {}; & \\
    &
    \node(0,0) (stage18) {}; & &
    \node(0,0) (stage28) {}; & &
    \node(0,0) (stage38) {}; & \\
    &
    \node(0,0) (stage19) {}; & &
    \node(0,0) (stage29) {}; & &
    \node(0,0) (stage39) {}; & \\
    &
    \node(0,0) (stage110) {}; & &
    \node(0,0) (stage210) {}; & &
    \node(0,0) (stage310) {}; & \\
    &
    \node(0,0) (stage111) {}; & &
    \node(0,0) (stage211) {}; & &
    \node(0,0) (stage311) {}; & \\
    &
    \node(0,0) (stage112) {}; & &
    \node(0,0) (stage212) {}; & &
    \node(0,0) (stage312) {}; & \\
  };

  \fill
    (stage10) node[fill=white] {\Large $\Rp$}
    (stage20) node[fill=white] {\Large $\Rq$}
    (stage30) node[fill=white] {\Large $\Rr$};


  \filldraw[fill=blue!10]
    (stage11.north west) rectangle (stage11.south east)

    (stage25.north west) rectangle (stage25.south east)

    (stage39.north west) rectangle (stage39.south east)
    ;

  \filldraw[fill=green!10]
    (stage110.north west) rectangle (stage110.south east)

    (stage22.north west) rectangle (stage22.south east)

    (stage36.north west) rectangle (stage36.south east)
    ;

  \filldraw[fill=purple!10]
    (stage111.north west) rectangle (stage112.south east)

    (stage23.north west) rectangle (stage24.south east)

    (stage37.north west) rectangle (stage38.south east)
    ;

    \draw [dashed]
      (stage12.north west) rectangle (stage19.south east)
      (stage21.north west) rectangle (stage21.south east)
      (stage31.north west) rectangle (stage35.south east)
      ;

  \fill
    (stage16) node { \emph{wait} }
    (stage21) node { \emph{wait} }
    (stage33) node { \emph{wait} }
    ;

  \fill
    (stage11) node { send }
    (stage110) node { receive }
    (stage111) node [yshift=-.25cm] { compute }

    (stage22) node { receive }
    (stage23) node [yshift=-.25cm] { compute }
    (stage25) node { send }

    (stage36) node { receive }
    (stage37) node [yshift=-.25cm] { compute }
    (stage39) node { send }
    ;

  \draw [-Latex] (stage11.south east) -- (stage22.north west);
  \draw [-Latex] (stage25.south east) -- (stage36.north west);
  \draw [-Latex] (stage39.south west) -- (stage110.north east);

  \node at ($(stage11.north west) + (-.2,0)$) (co1) {} ;
  \node at ($(stage11.south west) + (-.2,0)$) (co2) {} ;
  \draw[thick, blue, -] ($(co1) + (-.1,0)$) --
                  ($(co1) + ( .1,0)$) ;
  \draw[thick, blue, -] ($(co2) + (-.1,0)$) --
                  ($(co2) + (.1,0)$) ;
  \draw[thick, blue, -, dotted] (co1.center) -- (co2.center) ;
  \node [xshift=-.2cm] at ($(co1)!.5!(co2)$) { $\csend$ };

  \node at ($(stage22.north east) + (.2,0)$) (ciq1) {} ;
  \node at ($(stage22.south east) + (.2,0)$) (ciq2) {} ;
  \draw[thick, blue, -]                     ($(ciq1) + (-.1,0)$) --
                                            ($(ciq1) + ( .1,0)$) ;
  \draw[thick, blue, -]                     ($(ciq2) + (-.1,0)$) --
                                            ($(ciq2) + (.1,0)$) ;
  \draw[thick, blue, -, dotted]              (ciq1.center) --
                                              (ciq2.center) ;
  \node [xshift=.2cm] at ($(ciq1)!.5!(ciq2)$) { $\crecv$ };

  \node at ($(stage23.north east) + (.2,0)$) (ccq1) {} ;
  \node at ($(stage24.south east) + (.2,0)$) (ccq2) {} ;
  \draw[thick, blue, -]                     ($(ccq1) + (-.1,0)$) --
                                            ($(ccq1) + ( .1,0)$) ;
  \draw[thick, blue, -]                     ($(ccq2) + (-.1,0)$) --
                                            ($(ccq2) + (.1,0)$) ;
  \draw[thick, blue, -, dotted]              (ccq1.center) --
                                              (ccq2.center) ;
  \node [xshift=.2cm] at ($(ccq1)!.5!(ccq2)$) { $\ccc_\Rq$ };

  \node at ($(stage25.north east) + (.2,0)$) (coq1) {} ;
  \node at ($(stage25.south east) + (.2,0)$) (coq2) {} ;
  \draw[thick, blue, -]                     ($(coq1) + (-.1,0)$) --
                                            ($(coq1) + ( .1,0)$) ;
  \draw[thick, blue, -]                     ($(coq2) + (-.1,0)$) --
                                            ($(coq2) + (.1,0)$) ;
  \draw[thick, blue, -, dotted]              (coq1.center) --
                                              (coq2.center) ;
  \node [xshift=.2cm] at ($(coq1)!.5!(coq2)$) { $\csend$ };

  \node at ($(stage36.north east) + (.2,0)$) (cir1) {} ;
  \node at ($(stage36.south east) + (.2,0)$) (cir2) {} ;
  \draw[thick, blue, -]                     ($(cir1) + (-.1,0)$) --
                                            ($(cir1) + ( .1,0)$) ;
  \draw[thick, blue, -]                     ($(cir2) + (-.1,0)$) --
                                            ($(cir2) + (.1,0)$) ;
  \draw[thick, blue, -, dotted]              (cir1.center) --
                                              (cir2.center) ;
  \node [xshift=.2cm] at ($(cir1)!.5!(cir2)$) { $\crecv$ };

  \node at ($(stage37.north east) + (.2,0)$) (ccr1) {} ;
  \node at ($(stage38.south east) + (.2,0)$) (ccr2) {} ;
  \draw[thick, blue, -]                     ($(ccr1) + (-.1,0)$) --
                                            ($(ccr1) + ( .1,0)$) ;
  \draw[thick, blue, -]                     ($(ccr2) + (-.1,0)$) --
                                            ($(ccr2) + (.1,0)$) ;
  \draw[thick, blue, -, dotted]              (ccr1.center) --
                                              (ccr2.center) ;
  \node [xshift=.2cm] at ($(ccr1)!.5!(ccr2)$) { $\ccc_\Rr$ };

  \node at ($(stage39.north east) + (.2,0)$) (cor1) {} ;
  \node at ($(stage39.south east) + (.2,0)$) (cor2) {} ;
  \draw[thick, blue, -]                     ($(cor1) + (-.1,0)$) --
                                            ($(cor1) + ( .1,0)$) ;
  \draw[thick, blue, -]                     ($(cor2) + (-.1,0)$) --
                                            ($(cor2) + (.1,0)$) ;
  \draw[thick, blue, -, dotted]              (cor1.center) --
                                              (cor2.center) ;
  \node [xshift=.2cm] at ($(cor1)!.5!(cor2)$) { $\csend$ };

  \node at ($(stage110.north west) + (-.2,0)$) (ci1) {} ;
  \node at ($(stage110.south west) + (-.2,0)$) (ci2) {} ;
  \draw[thick, blue, -] ($(ci1) + (-.1,0)$) --
                  ($(ci1) + ( .1,0)$) ;
  \draw[thick, blue, -] ($(ci2) + (-.1,0)$) --
                  ($(ci2) + (.1,0)$) ;
  \draw[thick, blue, -, dotted] (ci1.center) -- (ci2.center) ;
  \node [xshift=-.2cm] at ($(ci1)!.5!(ci2)$) { $\crecv$ };

  \node at ($(stage111.north west) + (-.2,0)$) (cc1) {} ;
  \node at ($(stage112.south west) + (-.2,0)$) (cc2) {} ;
  \draw[thick, blue, -] ($(cc1) + (-.1,0)$) --
                  ($(cc1) + ( .1,0)$) ;
  \draw[thick, blue, -] ($(cc2) + (-.1,0)$) --
                  ($(cc2) + (.1,0)$) ;
  \draw[thick, blue, -, dotted] (cc1.center) -- (cc2.center) ;
  \node [xshift=-.2cm] at ($(cc1)!.5!(cc2)$) { $\ccc_\Rp$ };

  \end{tikzpicture}
  \end{center}
  \vskip-0.2cm
  \caption{A ring protocol, and an execution trace.}
  \label{fig:ring-proto}
  \vspace{-8mm}
\end{wrapfigure}
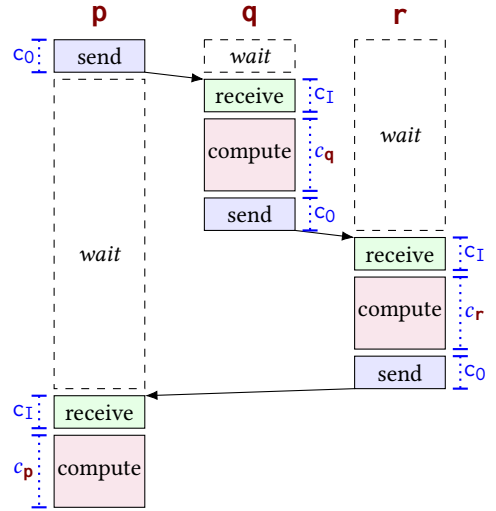

This section introduces cost-aware multiparty session protocols (\CAMP{}) which is
an extension of multiparty session types (\MPST{})
\cite{CDPY2015,DY13,Honda2008Multiparty} where the payload types ($\tau$) are
types that have been extended with \emph{size annotations}, adapted from the
literature on \emph{sized types}~\cite{DBLP:conf/popl/HughesPS96,
DBLP:journals/pacmpl/AvanziniL17}, and interactions have been extended with
\emph{cost annotations} ($\ccc$) which represent the local execution time at the
receiver:

\vskip.2cm
{\small
$
\begin{array}{@{}l@{}}
  \tau \coloneqq \mathsf{int}^i \mid \cdots \mid D^i \; \vec{\tau}
  \\
  \ccc \Coloneqq i \mid k \mid \ccc + \ccc \mid \max (\ccc, \ccc) \mid k \times \ccc
\end{array}
$
}
\vskip.2cm

\noindent
Our types are base types (integer, boolean, \ldots) annotated with a size $i$,
type constructors $D$ applied to a sequence of sized types $\vec{\tau}$,
annotated with a size $i$. Cost expressions $\ccc$
are either sizes $i$, constants
$k$, the addition of two costs, the maximum of two costs, or a constant
multiplied by a cost. A size $i$ an arithmetic expression that may contain
constants ($k$) or size variables ($n$, $m$, \ldots).
Definitions of global and local types are based on the most
commonly used \MPST{} in the literature
\cite{CDPY2015}. The syntax of global ($G$)
and local ($L$) types in \MPST{} is given below:
\[
   \begin{array}{@{}l@{\;}l@{}}
      G  & \Coloneqq \role{p} \gMsg \role{q} \gTy{\tau \hasCost \ccc} . G
           \mid      \role{p} \gMsg \role{q} : \{l_i  .  G_i\}_{i\in I}
           \mid \gFix X .  G
           \mid X
           \mid \gEnd\\[1mm]

      L  & \Coloneqq \role{p} \tSend \tau . L
           \mid      \role{p} \tSelect \{l_i . L_i\}_{i \in I}
           \mid      \role{p} \tRecv \tau \hasCost \ccc . L
           \mid      \role{p} \tBranch \{l_i . L_i\}_{i \in I}
           \mid      \lFix X . L
           \mid X
           \mid \lEnd
 \end{array}
\]
We start with a set of \emph{roles}, $\role{p}$, $\role{q}$, \ldots, and a set
of \emph{labels}, $l_1$, $l_2$, \ldots.  These are considered as natural
numbers: roles are \emph{participant} identifiers, e.g.\ thread or process ids;
and labels are tags that differentiate branches in the data/control flow.

Global type $\role{p} \gMsg \role{q} \gTy{\tau \hasCost \ccc} . G$ denotes
\emph{data} interactions from role $\role{p}$ to role $\role{q}$ with value of
type $\tau$ and local computation cost $\ccc$;
\emph{Branching} is represented by $\role{p} \gMsg \role{q} : \{ l_i .
G_i\}_{i\in I}$ with actions $l_i$ from $\role{p}$ to $\role{q}$. $\gEnd$
represents a \emph{termination} of the protocol. $\gFix X .  G$ represents a
\emph{recursion}, which is \emph{equivalent} to $[\gFix X. \; G / X]
G$.  We assume recursive types are guarded \cite{pierce2002types}.

Each role in $G$ represents a different participant in a parallel
process. \emph{Local types} represent the communication actions
performed by each role.  The \emph{send} type $\role{p} \lSend{\tau} . L$
expresses sending of a value of type $\tau$ to role $\role{p}$ followed by
interactions specified by $L$.  The \emph{receive} type $\role{p}\tRecv \tau \hasCost \ccc . L$ receives a value of type
$\tau$ from role $\role{p}$ with local computation
cost $\ccc$.  The \emph{selection} type
represents the transmission to role $\role{p}$ of label $l_i$ chosen
in the set of labels ($i\in I$) followed by $L_i$. The {\em branching}
type is its dual.  The rest are the same as $G$.
$\getRoles(G)/\getRoles(L)$ denote the set of roles that occur in
$G$/$L$.

\begin{myremark}\label{rem:syntax}\rm
  Global types which combine
  label and data messages are also used in the literature. They can be encoded as global types in this paper
  by using singleton labels (see \cite[p.178]{DY13}). E.g.
$\role{p} \gMsg \role{q} : \{l_i(\tau_i)  .  G_i\}_{i\in I}$
is encoded as
$\role{p} \gMsg \role{q} : \{l_i  .  G_i'\}_{i\in I}$
and $G_i' = \role{p} \gMsg \role{q} \gTy{\tau_i} . G_i$.
It is possible to account for the differences in cost by setting the cost of
sending/receiving labels appropriately, e.g.\ removing the cost of sending
labels, and slightly increasing the size of the data messages, to account for
the fact that they must be sent alongside a label.
\end{myremark}

\myparagraph{End Point Projection}
The local type $L$ of a participant $\role{p}$ in a global type $G$
can be obtained by the \emph{end point projection} (EPP) of $G$ onto
$\role{p}$, denoted by $G$ as $G \project \role{p}$. The local type
gives a local view of a global protocol onto each participant.
Our definition of EPP follows the standard
projection rules in \cite{DY13,Demangeon2012Nested}.
The projection uses the
\emph{full merging} operator \cite{DY13,Demangeon2012Nested}, which
allows more well-formed global types than the original projection
rules \cite{Honda2008Multiparty}.

\begin{definition}[Projection and Merging]
  \label{def:session-project-merging}
The \emph{end point projection} (EPP) of $G$ onto
$\role{p}$, denoted by $G$ as $G \project \role{p}$,
is the partial function defined below, together with
the merging of local types $L_i$:\\[1mm]
{\bf\emph{Projection:}}
\[
\small
      \begin{array}{@{}l@{}l@{}}
        \begin{array}[t]{@{}l@{}}
          (\role{q} \gMsg \role{r} \gTy{\tau \hasCost \ccc } . \; G) \project \role{p}
          \\ \quad =
          \left\{
          \begin{array}{@{}l@{\;\;}l@{}}
            \role{r} \lSend{\tau} . \; (G \project \role{p})
            & \text{if} \; \role{p} = \role{q} \neq \role{r}
            \\
            \role{q} \lRecv{\tau} \hasCost \ccc . \; (G \project \role{p})
            & \text{if} \; \role{p} = \role{r} \neq \role{q}
            \\
            G \project \role{p}
            & \text{otherwise}
          \end{array} \right.
        \end{array}
&
        \quad
        \begin{array}[t]{@{}l@{}}
          (\role{q} \gMsg \role{r} : \{l_i .  G_i\}_{i\in I}) \project \role{p}
          \\ \quad =
          \left\{
          \begin{array}{@{}l l@{}}
            \role{r} \lChoice \{l_i .  G_i \project \role{p} \}
            & \text{if} \; \role{p} = \role{q} \neq \role{r}
            \\
            \role{q} \lBranch \{l_i .  G_i \project \role{p} \}
            & \text{if} \; \role{p} = \role{r} \neq \role{q}
            \\
            \sqcap_{i\in I} (G_i \project \role{p})
            & \text{otherwise}
          \end{array}\right.\\
        \end{array}\\
        \\[-2mm]
        (\gFix X .  G) \project \role{p}\\
        \quad\quad= \left\{
        \begin{array}{ll}
          \gFix X . ( G \project \role{p})
          & \text{if} \; G \project \role{p} \neq X', \; \forall X'
          \\
          \gEnd
          & \text{otherwise}
        \end{array}\right.
&\quad\quad
        \begin{array}{@{}l@{}}
          (X) \project \role{p} =  X
          \hspace{1cm}
          \gEnd \project \role{p} =  \gEnd
        \end{array}
      \end{array}
    \]
    {\bf\emph{Merging:}}
    \[
\small
      \begin{array}{@{}l@{}}
        \role{p} \lBranch \{ l_i .  L _i\}_{i \in I} \sqcap \role{p} \lBranch \{l_j .  L _j'\}_{j \in J}
         =
        \role{p} \lBranch \{l_k .  L _k \sqcap  L _k' \}_{k \in I \cap J} \cup
        \{ l_l .  L _l \}_{l \in I \setminus J}\cup 
        \{ l_m .  L _m \}_{m \in J \setminus I} 
        \\[.2cm]
        \lFix X .  L _1 \sqcap \lFix X .  L _2
        = \lFix X . ( L _1 \sqcap  L _2)
        \hspace{1cm}
        L  \sqcap  L =  L
      \end{array}
    \]
\end{definition}
The first line of the projection rule defines a case where the sender and
receiver are the same \cite{DYBH12}.  The global type projection onto a role is
not necessarily defined. Particularly, projecting $\role{q}\gMsg \role{r} :
\{l_i .  G_i\}$ onto $\role{p}$, with $\role{r} \neq \role{p}$ and $\role{q}\neq
\role{p}$, is only defined if the
projection of all $G_i$ onto $\role{p}$ can be \emph{merged}
(Def.\ \ref{def:session-project-merging}).
Two local types can be
merged only if they are the same, or if they branch on the same role, and
their continuations can be merged.
For example, $\role{p}$'s local type of the global type
$\role{p} \gMsg \role{q} \gTy{\tau \hasCost \ccc} . \gEnd$
is $\role{q} \tSend \tau . \lEnd$, while $\role{q}$'s is
$\role{p} \tRecv \tau \hasCost \ccc . \lEnd$.
As a more complex example, $\role{r}$'s local type of the branching
global type:
\[
\gFix X . \role{p} \gMsg \role{q}
\left
\{
\begin{array}{@{}l@{}}
             l_1  .  \role{q} \gMsg \role{r}
             : l_2  . \role{p} \gMsg \role{r}
             : l_3 .  X, \
             l_4 . \role{q} \gMsg \role{r}
             : l_5 . \role{p} \gMsg \role{r}
             : l_6 . \gEnd
           \end{array}\right
\}
\]
is
$\gFix X . \role{q}\lBranch
 \{l_2 . \role{p}\lBranch \{l_3 . X\}, l_5 . \role{p}\lBranch\{l_6 . \lEnd\}\}$.

We say that a global type is \emph{well formed}, if its projection on all its
roles is defined. We denote:
$\WF(G)
  = \forall \role{p} \in \getRoles(G)
  , \exists  L ,  G \project \role{p} =  L$.

\begin{definition}[Label Broadcasting] We define a macro to represent the
broadcasting of a label to multiple
participants in a choice. We write $\role{p} \gMsg \{ \role{q}_j \}_{j \in [1, n]}
\bm{:} \{ l_i . G_i \}_{i \in I}$ as a synonym to $\role{p} \to \role{q}_1 \{
l_i. \role{p} \to \role{q}_2 \{ l_i. \ldots . \role{p} \to \role{q}_n \{ l_i . G_i \} \ldots
\} \}_{i \in I}$.  Similarly, for local types, $\{ \role{q} \}_{j \in [1,n]}
\lChoice \{ l_i . L_i \}_{i \in I}$ expands to $\role{q} \lChoice \{ l_i
. \role{q}_2 \lChoice \{l_i . \ldots \role{q}_n \lChoice \{ l_i . L_i \} \ldots \}
\}_{i \in I}$.
\end{definition}
\noindent It is straightforward to derive: $(\role{p}\to\{\role{q}_j\}_{j \in
J} \bm{:} \{l_i . G_i\}_{i \in I}) \project \role{r} = \{ \role{q} \}_{j \in J}
\lChoice \{ l_i . G_i \project \role{r} \}_{i\in I}$, if $\role{p} = \role{r}$,
and
$(\role{p}\to\{ \role{q}_j\}_{j \in J} \bm{:} \{l_i . G_i\}_{i \in I}) \project
\role{r} = \role{q}_j \lBranch \{ l_i . G_i \project \role{r} \}_{i\in I}$, if
$\role{r} = \role{q}_j$ for some $j \in J$.

\subsection{Labelled Transition System of Global Types}
\label{sec:GTLTS}
We introduce the labelled transition system (LTS) of global types to associate
protocol execution costs with cost annotations.  Our semantics is based on the
LTS for global and local types in \Citet{DY13} that define their asynchronous
operational semantics, and prove their sound and complete correspondence.

We designate the \emph{observables} ($\act, \act', \ldots$) to be the send, receive,
branch and select actions that trigger a transition, and
an \emph{internal} transition at a role, which represents the cost $\ccc$ a
role spends performing computation at the receiver (denoted by $\role{p} \tRun
\ccc$).  The syntax of the observables is:
\[ \act \Coloneqq \role{p}\role{q}\tSend \tau \mid \role{p}\role{q}\tRecv \tau
\mid
\role{p}\role{q}\tSelect \mid \role{p}\role{q}\tBranch \mid
\role{p} \tRun \ccc
\]
The $\role{p} \tRun \ccc$ does not affect the communication structure of the
protocol, similar to the silent actions of common
process calculi. We say that the \emph{subject} of an action $\act$ is the
role in charge of performing it: $\role{p} = \subj(\role{p}\role{q}\tSend\tau) =
\subj(\role{q}\role{p}\tRecv\tau) = \subj(\role{p}\role{q}\tSelect) =
\subj(\role{q}\role{p}\tBranch ) = \subj(\role{p} \tRun \ccc)$.

Following
\cite{DY13}, we extend the grammar of $G$ to represent the intermediate steps in
the execution with the construct $\role{p} \gMsgt \role{q} \gTy{\tau \hasCost \ccc} . G$ to
represent the fact that $\role{p}$ has sent the message of type
$\tau$ but $\role{q}$ has not received it yet,
and $\role{p} \gEval (\tau \hasCost \ccc).G$ to represent that $\role{p}$ is performing a
computation of type $\tau$ and cost $\ccc$.
For the branching
we use
$\role{p} \gMsgt \role{q} \; j \; \{l_i. \;
G_i\}_{i\in I}$ to represent the fact that $\role{p}$ has sent label $l_j$ to
$\role{q}$. Then the LTS for global types is defined as below. The
main rules different from \cite{DY13} are
[GR1a,GR2a,GR2b] which consider the execution cost. When we send a
message or a label, the type becomes the received
mode $\role{p} \gMsgt \role{q}$ (e.g. [GR1a]) and then it
asynchronously receives the corresponding message (e.g. [GR2a]).
We also observe the actions under the prefix if the participants
are unrelated (e.g. [GR4a]).

\begin{definition}[LTS for Global Types]\label{def:lts-global}
  The relation $G \xrightarrow{\act} G'$ is defined as follows:

  \scalebox{.9}{%
  \begin{mathpar}\mprset{sep=0.5em}
     \rulename{GR1a} \;
       \role{p} \gMsg  \role{q} \gTy{\tau \hasCost \ccc} . G
       \xrightarrow{\role{p}\role{q}\tSend \tau}
       \role{p} \gMsgt \role{q} \gTy{\tau \hasCost \ccc} . G
     \and
     \rulename{GR1b} \;
         \role{p} \gMsg  \role{q} \{l_{i}. G_i\}_{i \in I}
         \xrightarrow{\role{p}\role{q}\tSelect l_{j} }
         \role{p} \gMsgt \role{q} \; : j \; \{l_{i}. G_i\}_{i \in I} \quad (j \in I)
     \and
     \rulename{GR2a} \;
         \role{p} \gMsgt \role{q} \gTy{\tau \hasCost \ccc} . G
         \xrightarrow{\role{p}\role{q}\tRecv \tau}
         \role{q} \gEval (\tau \hasCost \ccc) . G
     \and
     \rulename{GR2b} \;
         \role{q} \gEval (\tau \hasCost \ccc) . G
         \xrightarrow{\role{q}\tRun \ccc}
         G
     \and
     \rulename{GR2c} \;
         \role{p} \gMsgt \role{q} \; : j \; \{l_{i}. G_i\}_{i \in I}
         \xrightarrow{\role{p}\role{q}\tBranch l_j}
         G_j
     \and
     \rulename{GR4a}
     \inferrule{%
       G \xrightarrow{\act} G'
       \\
       \role{p}, \role{q} \not\in \subj(\act)
     }{%
       \role{p} \gMsg \role{q} \gTy{\tau \hasCost \ccc} . G
       \xrightarrow{\act}
       \role{p} \gMsg \role{q} \gTy{\tau \hasCost \ccc} . G'
     }
     \and
     \rulename{GR4b}
     \inferrule{%
       \forall i \in I
       \\
       G_i \xrightarrow{\act} G_i'
       \\
       \role{p}, \role{q} \not\in \subj(\act)
     }{%
       \role{p} \gMsg \role{q} \{l_{i}. \; G_i\}_{i \in I}
       \xrightarrow{\act}
       \role{p} \gMsg \role{q} \{l_{i}. \; G_i'\}_{i \in I}
     }
     \and
     \rulename{GR5a}
     \inferrule{%
       G \xrightarrow{\act} G'
       \\
       \role{q} \not\in \subj(\act)
     }{%
       \role{p} \gMsgt \role{q} \gTy{\tau \hasCost \ccc} . G
       \xrightarrow{\act}
       \role{p} \gMsgt \role{q} \gTy{\tau \hasCost \ccc} . G'
     }
\and
\hspace*{-4mm}
     \rulename{GR5b}
     \inferrule{%
       G \xrightarrow{\act} G'
       \\
       \role{p} \not\in \subj(\act)
     }{%
       \role{p} \gEval (\tau \hasCost \ccc) . G
       \xrightarrow{\act}
       \role{p} \gEval (\tau \hasCost \ccc) . G'
     }
\quad
     \rulename{GR3}
     \inferrule{%
       G[\gFix X. G/X] \xrightarrow{\act} G'
     }{%
       \gFix X. G \xrightarrow{\act} G'
     }
     \quad
     \rulename{GR5c}
     \inferrule{%
       G_j \xrightarrow{\act} G_j'
       \\
       \role{q} \not\in \subj(\act)
       \\
       \forall i \in I \setminus j, G_i = G_i'
     }{%
       \role{p} \gMsgt \role{q} \; j \; \{l_{i}. \; G_i \}_{i \in I}
       \xrightarrow{\act}
       \role{p} \gMsgt \role{q} \; j \; \{l_{i}. \; G_i' \}_{i \in I}
     }
  \end{mathpar}
  }
\end{definition}

\subsection{Labelled Transition System of Local Types}
\label{sec:LTLTS}
The labelled transition system (LTS) of local types are given for configurations
($C$) which map each participant to its local type and a set of FIFO queues
($Q$) where each represents a queue from sender $\role{p}$ to receiver
$\role{q}$.  We also extend the syntax $(\tau \hasCost
\ccc).L$ to represent the intermediate state where the receiver executes a local
computation with cost $\ccc$.
In the definition below, [LR2,LR3] formalise the observability of the local
computation cost $\ccc$ when receiving the value. Other rules are the
standard FIFO enqueue and dequeue rules.

\begin{definition}[LTS for Local Types]\label{def:lts_local}%
The relation
$\langle C, Q \rangle \xrightarrow{\ell} \langle C', Q' \rangle$
where $C = [\role{p} \mapsto L_i]_{i\in I}$ and
$Q = [\role{p}\role{q} \mapsto w]_{i\in I, j \in I}$
is defined as follows:
\begin{center}
\scalebox{.9}{%
$
\begin{array}{l r l}
\text{[LR1]}
&
\langle C[\role{p} \mapsto \role{q}\lSend{\tau}. \; L],
Q[\role{p}\role{q}\mapsto w] \rangle
&
\xrightarrow{\role{p}\role{q}\lSend \tau}
\langle C[\role{p} \mapsto L], Q[\role{p}\role{q}\mapsto a \cdot
w] \rangle
\\
\text{[LR2]}
&
\langle C[\role{p} \mapsto \role{q}\lRecv{\tau\hasCost \ccc}. \; L],
Q[\role{q}\role{p}\mapsto w\cdot \tau] \rangle
&
\xrightarrow{\role{q}\role{p}\lRecv \tau}
\langle C[\role{p} \mapsto (\tau \hasCost \ccc).L],
Q[\role{q}\role{p}\mapsto w] \rangle
\\
\text{[LR3]}
&
\langle C[\role{p} \mapsto (\tau \hasCost \ccc). L], Q \rangle
&
\xrightarrow{\role{p} \gEval \ccc}\
\langle C[\role{p} \mapsto L], Q \rangle
\\
\text{[LR4]}
&
\langle
C[\role{p} \mapsto \role{q}\lChoice\{l_{i}. L_i\}_{i\in I}],
Q[\role{p}\role{q}\mapsto w]
\rangle
&
\xrightarrow{\role{p}\role{q}\lChoice l_i}
\langle C[\role{p} \mapsto L_i], Q[\role{p}\role{q}\mapsto l_i \cdot w] \rangle
\\
\text{[LR5]}
&
\langle
C[\role{p} \mapsto \role{q}\lBranch\{l_{i}. L_i\}_{i\in I}],
Q[\role{q}\role{p}\mapsto w\cdot l_i] \rangle
&
\xrightarrow{\role{q}\role{p}\lBranch l_i}
\langle C[\role{p} \mapsto L_i], Q[\role{q}\role{p}\mapsto w] \rangle
\end{array}
$
}
\end{center}
\end{definition}
The following theorem shows the global type semantics is exactly
matched with local asynchronous interactions between participants.

\begin{restatable}{theorem}{ltsSoundComplete}[Soundness and Completeness]\label{thm:lts-sound-complete}
  Let $G$ be a global type with $\roles = \getRoles(G)$ and let
  $C = [\role{p} \mapsto G \project \role{p}]_{\role{p} \in \roles}$.
  Then $G \approx (C;[\role{p}\role{q} \mapsto \epsilon]_{\role{p},\role{q} \in \roles})$.
\end{restatable}

\begin{proof}(Sketch)
The proof of soundness and completeness is a straightforward adaptation of that
in \cite{DY13}. The distinction between send/receive and select/branch actions
is straightforward, e.g.\ actions GR1a and GR1b are special cases of rule GR1 in
\cite{DY13} (see Remark ~\ref{rem:syntax}). The addition of cost actions to
local types is does not complicate the proof, since it only happens after
communication has taken place, which is ensured by being local to each role, and
the local context.
\end{proof}

\begin{definition}[Deadlock-freedom]
\label{def:deadlock}
We call $\langle C_0, Q_0 \rangle$ \emph{deadlock-free} if
for all $C$ and $Q$ such that
$\langle C_0, Q_0 \rangle\xrightarrow{\vec{\ell}}\langle C, Q
\rangle$, either (1) $\forall \role{p} \in \dom{C}.\ C(\role{p})=\lEnd$
and $Q=\emptyset$; or (2)
$\langle C, Q \rangle\xrightarrow{\ell}\langle C', Q' \rangle$
for some $C'$ and $Q'$. We call global type $G$
\emph{deadlock-free}
if $\langle C_0, \emptyset \rangle$ such that
$C_0 = [\role{p} \mapsto G \project \role{p}]_{\role{p} \in \roles}$
with $\roles = \getRoles(G)$
is deadlock-free.
\end{definition}

Note that the definition of deadlock-freedom is not affected by cost
annotations. Hence by Remark~\ref{rem:syntax}, we can directly apply the result
in \cite{DY13} to obtain:

\begin{theorem}[Deadlock-freedom]{\rm (\cite{DY13})}\label{thm:deadlockfree}
$\WF(G)$ is deadlock-free.
\end{theorem}

\section{Cost for Multiparty Session Protocols (1): Bounded Recursion}
\label{sec:cost}

This section presents the cost analysis for protocols with bounded recursions.
We first define the cost of a trace, as the total cost accumulated by each
participant at the end of the execution of a trace. Next we introduce the cost
model using global types and show that it provides an upper bound
of the cost of any possible trace for each participant.

\subsection{Cost of Local Traces}
\label{subsec:trace}
We first explain several assumptions for giving a calculation of
cost.
Our first assumption in theory is that
participants do not share resources with other participants,
i.e.\ they can run on
independent CPUs, with no source of contention such as memory or
shared cache.
\begin{wrapfigure}{l}{.35\columnwidth}
  \begin{center}
  \begin{tikzpicture}[baseline={([xshift=-2cm, yshift=-4cm]current bounding box.north)}]
    \foreach[count=\i] \a in {\Rp,\Rr,\Rq}{
      \node[draw, circle] (nd\i) at (\i*360/3: 1.1cm) {\small $\a$ };
    }
    \node at ($(nd1)+(.8,.3)$) {$\csend$ };
    \node at ($(nd2)+(-.7,0.3)$) {$\csend$ };
    \node at ($(nd3)+(0,-.8)$) {$\csend$ };

    \node at ($(nd1)+(-.8,-.3)$) {$\crecv$ };
    \node at ($(nd2)+(.7,-.2)$) {$\crecv$ };
    \node at ($(nd3)+(0,.8)$) {$\crecv$ };

    \node at ($(nd1)+(0,.55)$) { $\ccc_\Rp$ };
    \node at ($(nd2)+(0,-.55)$) { $\ccc_\Rr$ };
    \node at ($(nd3)+(.55,0)$) { $\ccc_\Rq$ };
    \path[draw, ->, >=stealth] (nd3) edge[bend left=35] (nd2);
    \path[draw, ->, >=stealth] (nd2) edge[bend left=35] (nd1);
    \path[draw, ->, >=stealth] (nd1) edge[bend left=35] (nd3);
  \end{tikzpicture}
  \end{center}
  \vskip-1cm
  \caption{A cost-annotated ring protocol.}
  \label{fig:ring-cost}
\end{wrapfigure}
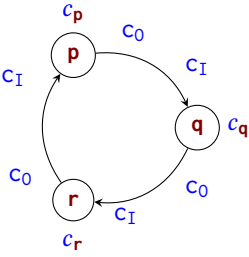

The cost of a trace is the total execution time
taken by each participant to run
the protocol from start to end.
We compute this total execution cost by
tracking the dependency from the input to the output (\emph{IO-dependency}),
and by associating each action in a
trace (Def. \ref{def:lts-global}) with an execution cost. Assuming that every
participant has access to their own set of resources (including CPU) implies
that every action in the trace (Def.\ \ref{def:lts-global}) will be
triggered as early as possible, e.g.\ send actions will not be arbitrarily delayed by other
actions. This also implies that any pair of actions will happen in parallel if
swapped freely according to the semantics.

We explain these assumptions with simple examples.
Consider the following trace:

\vskip.1cm
{\small
\begin{center}
$
\Rp \gMsg \Rr \gTy{\tau \hasCost \ccc}. \Rq \gMsg \Rr \gTy{\tau' \hasCost \ccc'}
\xrightarrow{{\Rp\Rr \tSend \tau \cdot \; \Rq\Rr \tSend \tau'}}
\Rp \gMsgt \Rr \gTy{\tau \hasCost \ccc}. \Rq \gMsgt \Rr \gTy{\tau' \hasCost \ccc'}
$
\end{center}
}
\vskip.1cm

\noindent
According to Def.\ \ref{def:lts-global}, since
$\subj(\Rp\Rr\tSend \tau) = \Rp \neq \Rq = \subj(\Rq\Rr \tSend \tau')$,
because the sender is different in each action, this would be another possible trace
for the same global type:

\vskip.1cm
{\small
\noindent
\begin{center}
$
\Rp \gMsg \Rr \gTy{\tau \hasCost \ccc}. \Rq \gMsg \Rr \gTy{\tau' \hasCost \ccc'}
\xrightarrow{{\Rq\Rr \tSend \tau' \cdot \; \Rp\Rr \tSend \tau}}
\Rp \gMsgt \Rr \gTy{\tau \hasCost \ccc}. \Rq \gMsgt \Rr \gTy{\tau' \hasCost \ccc'}
$
\end{center}
}
\vskip.1cm

\noindent
The intuition is that, since $\Rp$ and $\Rq$ are running on different CPUs, and
their actions are independent, both $\Rq\Rr \tSend \tau'$ and $\Rp\Rr \tSend
\tau$ can happen in parallel.

We assume that the cost of the message-passing operations depend on the size of
the data that is sent and that the costs of sending and receiving
messages are known, and that they are functions on the size of the
data.
Due to the presence of IO-dependencies, we record when the \emph{send}
actions have happened: a participant cannot perform a computation until received
a value, and it cannot receive a value until at least the time it took for the
sender to finish sending the data has passed.

\myparagraph{Cost Environments, Queues and Trace Cost}
To record the cost, we
use \emph{queues} that record when the data becomes available. A cost dependency
queue is, similarly to $Q$ in Def. \ref{def:lts_local}, a mapping from pairs of
participants to queues of execution times, that records when the data in the queue
becomes available. We use $\Queue$ for these cost queues.

The cost of a sequence of actions is defined as a mapping from
participants to the total execution time accumulated by each
participant, defined as \emph{cost environments}.
It is computed by adding the cost of each individual
action to the cost of the participant that performs it, taking into
account the cost dependencies recorded by the queue.


We call mappings from participants to total accumulated
costs \emph{cost environments}. If $\ccc$ is an execution time estimation and $R$
is a set of participants, then
$\Time = [\role{r} \mapsto \ccc]_{r\in R}$, with the usual
extension $\Time[\role{r} \mapsto \ccc]$ and indexing $\Time(\role{r})$ operations. We define
$\Time(\role{r}) = 0$, if $\role{r} \notin \Time$ and the following
operations:
{
\begin{center}
$
\begin{array}{@{}l@{\;}l@{}l@{\;}l}
  \Time[\ccc' | \role{r} \addcost \ccc] & = \Time[\role{r} \mapsto
\max(\Time(\role{r}), \ccc') + \ccc]
\end{array}
$
\end{center}
}
\noindent
This operation is used to record cost dependencies. Specifically, $\Time[\ccc' | \role{r} \addcost \ccc]$ means
that $\role{r}$ incurs additional cost $\ccc$, after possibly waiting for an
action by an external process with total cost $\ccc'$. If
the action that took $\ccc$ time depends on an action that took $\ccc'$ time, then the
cost of $\role{r}$ will be updated with the maximum of either the time $\ccc'$, or the
total accumulated cost by $\role{r}$. We write:
\begin{center}
{
$\Time[\role{r}' | \role{r} \addcost \ccc]$ for
$\Time[\Time(\role{r}') | \role{r} \addcost \ccc]$ \quad and \quad
$\Time[\role{r} \addcost \ccc]$ for $\Time[0|\role{r} \addcost \ccc]$.
}
\end{center}
\noindent For the cost of actions, we define:
\begin{enumerate*}
\item[(1)] $\crecv(\tau)$ is the time required for \emph{receiving}
a value of type $\tau$;
\item[(2)] $\csend(\tau)$ is the time required for \emph{sending} a value of type
$\tau$;
\item[(3)] $\ccc^\tau$ is cost associated to type $\tau$; and
\item[(4)] the cost of labels $l$ is calculated as unit type $1$.
\end{enumerate*}

\begin{definition}[Cost of a Trace and Action Cost]\label{def:cost_trace}
The cost of a trace $\vec{\act}$ takes as an input an initial cost $\Time$, an
input dependency queue $\Queue$, and produces a pair of a final cost $\Time'$
and queue $\Queue'$.
\begin{center}
$
  \cost(\epsilon)(\Time, \Queue) = (\Time, \Queue)
  \hspace{.5cm}
  \cost(\act \cdot \vec{\act}) = \cost(\vec{\act})(\cost(\act)(\Time, \Queue))
$
\end{center}
\noindent The initial cost is
$\cost(\vec{\act}) =
\cost(\vec{\act})([\role{r} \mapsto 0]_{\role{r} \in \vec{\act}},
                  [\role{p}\role{q} \mapsto
                  \epsilon]_{\role{p}\role{q} \in \vec{\act}})$
with empty initial queues and zero costs.
The cost of individual actions is defined below:
\begin{center}
  {\small
    \noindent%
  $\begin{array}{@{}r@{\;}l}
     \cost(\Rp\Rq\tRecv \tau) (\Time, \Queue[\Rp\Rq \mapsto \ccc \cdot w])
     & = (\Time[\ccc | \Rq \addcost \crecv(\tau)], \Queue[\Rp\Rq \mapsto w]))
     \\[1mm]
     \cost(\Rp\Rq\tBranch l_k) (\Time, \Queue[\Rp\Rq \mapsto \ccc \cdot w])
     & = (\Time[\ccc | \Rq \addcost \crecv(1)], \Queue[\Rp\Rq \mapsto w]))
     \\[1mm]
     \cost(\Rp\Rq\tSend \tau) (\Time, \Queue[\Rp\Rq \mapsto w]))
     & = (\Time[\Rp \addcost \csend(\tau)], \Queue[\Rp\Rq \mapsto w \cdot (\Time(\Rp) + \csend(\tau))])
     \\[1mm]
     \cost(\Rp\Rq\tSelect l_k) (\Time, \Queue[\Rp\Rq \mapsto w]))
     & = (\Time[\Rp \addcost \csend(1)], \Queue[\Rp\Rq \mapsto w \cdot (\Time(\Rp) + \csend(1))])
     \\[1mm]
     \cost(\Rp \tRun \ccc) (\Time, \Queue)
     & = (\Time[\Rp \addcost \ccc], \Queue)
   \end{array}$}
  \end{center}
\end{definition}

\myparagraph{Example of Trace Cost}
We show an example of calculating the cost of a trace. Consider the following
global type:
{\small
\begin{center}
$
G =
\Rp \gMsg \Rq \gTy {\kw{str}^n \hasCost n \times 3 \text{ms}} .
\Rq \gMsg \Rp \gTy {\kw{int}^i \hasCost 6 \text{ms}} .
\gEnd
$
\end{center}
}
\noindent
In this protocol, there are two participants, $\Rp$ and $\Rq$. First, $\Rp$
sends a string of size $n$ to $\Rq$, that requires $n \times 3 \text{ms}$ of
local computation time. Then, $\Rq$ replies with an integer of size $i$ (i.e.\
smaller than $i$) to $\Rp$, that takes a constant computation time of $6
\text{ms}$. We represent this scenario as a trace of actions:
{\small
\noindent
\begin{center}
$
\begin{array}{@{}l@{\;}l@{}}
  tr = &
  \Rp\Rq \tSend \kw{str}^n \cdot
  \Rp\Rq \tRecv \kw{str}^n \cdot
  \Rq \tRun (n \times 3 \text{ms}) \cdot
  \Rq\Rp \tSend \kw{int}^i \cdot
  \Rq\Rp \tRecv \kw{int}^i \cdot
  \Rp \tRun (6 \text{ms})
\end{array}
$
\end{center}
}
\noindent
To compute the cost, we traverse the trace, record at which time each event
happened in the message queue, and add the cost of each action to the total
execution time accumulated by the subject of the action. For example,
$\cost(\Rp\Rq \tSend \kw{str}^n)([], []) = ([\Rp \mapsto
\csend(\kw{str}^n)], [\Rp\Rq \mapsto \csend(\kw{str}^n)]) $, i.e.\ the cost
of sending a string of size $n$ is added to the cost for $\Rp$, and the message
queue now records that this message was sent after $\csend(\kw{str}^n)$ time.
Then, $\cost(\Rp\Rq\tRecv \kw{str}^n)([\Rp \mapsto
\csend(\kw{str}^n)], [\Rp\Rq \mapsto \csend(\kw{str}^n)])
=
([\Rp \mapsto \csend(\kw{str}^n); \; \Rq \mapsto \csend(\kw{str}^n) + \crecv(\kw{str}^n)]
, [])$. That means that the cost of receiving a string of size $n$ is added to
the cost of $\Rq$, after the time recorded in the queue $\Rp\Rq$, in this case
the cost of sending a string of size $n$, and the message queue would now be
empty. By following the cost rules with the remaining actions, we produce the following cost equation:
{\small
\noindent
\begin{center}
$
\cost(\mathsf{tr}) =
\left[
\begin{array}{@{}l@{\;}l@{\;}l}
  \Rp
  & \mapsto
  & \csend(\kw{str}^n) + \crecv(\kw{str}^n) + n \times 3 \text{ms} +
    \csend(\kw{int}^i) + \crecv(\kw{int}^i) + 6 \text{ms}
  \\
  \Rq
  & \mapsto
  & \csend(\kw{str}^n) + \crecv(\kw{str}^n) + n \times 3 \text{ms}
    + \csend(\kw{int}^i)
\end{array}
\right]
$
\end{center}
}
\noindent
By instantiating the sizes of the messages and the send/receive costs with e.g.\
profiling information, we can now estimate how much time it will take the
protocol to complete.

\begin{example}[Scatter/Gather]\label{exn:sc}
  This global type represents a scatter/gather protocol, where $\role{p}$
  distributes tasks to $\role{q}$ and $\role{r}$, and $\role{s}$ collects the
  results. We omit the cost on the receiving end of $\role{s}$ to represent that
  $\role{s}$ simply gathers the results, and that has computation cost $0$.

  \vskip.2cm
  {\noindent
  \small
  \begin{center}
  $
  G =
  \role{p} \gMsg \role{q} \gTy{ \tau_1 \hasCost \ccc_1 }.
  \role{p} \gMsg \role{r} \gTy{ \tau_1 \hasCost \ccc_1 }.
  \role{q} \gMsg \role{s} \gTy{ \tau_2 }.
  \role{r} \gMsg \role{s} \gTy{ \tau_2 }.
  \gEnd
  $
  \end{center}
  }
  \vskip.2cm

  \noindent
  We show below two examples of the possible traces:
  {\small
    \noindent
    \begin{center}
    $
    \begin{array}{l@{}}
      tr_1 =
      \role{p}\role{q}\tSend\tau_1 \cdot
      \role{p}\role{q}\tRecv\tau_1 \cdot
      \role{q}\tRun   \ccc_1 \cdot
      \role{q}\role{s}\tSend\tau_2 \cdot
      \role{q}\role{s}\tRecv\tau_2 \cdot
      \role{s} \tRun 0 \cdot
      \role{p}\role{r}\tSend\tau_1 \cdot
      \role{p}\role{r}\tRecv\tau_1 \cdot
      \role{r}\tRun \ccc_1 \cdot
      \role{r}\role{s}\tSend\tau_2 \cdot
      \role{r}\role{s}\tRecv\tau_2 \cdot
      \role{s} \tRun 0
      \\
      tr_2 =
      \role{p}\role{q}\tSend\tau_1 \cdot
      \role{p}\role{r}\tSend\tau_1 \cdot
      \role{p}\role{r}\tRecv\tau_1 \cdot
      \role{r}\tRun \ccc_1 \cdot
      \role{p}\role{q}\tRecv\tau_1 \cdot
      \role{q}\tRun \ccc_1 \cdot
      \role{q}\role{s}\tSend\tau_2 \cdot
      \role{r}\role{s}\tSend\tau_2 \cdot
      \role{q}\role{s}\tRecv\tau_2 \cdot
      \role{s} \tRun 0 \cdot
      \role{r}\role{s}\tRecv\tau_2 \cdot
      \role{s} \tRun 0
    \end{array}
    $
    \end{center}
  }
  Since we assume that each participant can run in parallel to the remaining of
  the participants, the cost of both traces yield the same result:
  {\small
    \noindent
    \begin{center}
    $
    \begin{array}{l@{}}
      \cost(tr_1)([], [])
      = \cost(\role{p}\role{q}\tRecv\tau_1\cdots)
      ([\role{p}\mapsto\csend(\tau_1)], [\role{p}\role{q}\mapsto \csend(\tau_1)])
      \\
      \quad
      = \cost(\role{q}\tRun \ccc_1 \cdots)
      ([\role{p}\mapsto\csend(\tau_1)]; \role{q}\mapsto \csend(\tau_1) + \crecv(\tau_1), [])
      \\
 \quad     = \left[
      \begin{array}{@{}l@{}}
        \role{p}\mapsto
        2\times \csend(\tau_1) ; \;
        \role{q}\mapsto
        \csend(\tau_1) + \crecv(\tau_1) + \ccc_1 + \csend(\tau_2) ; \;
        \;\;
        \role{r}\mapsto
        2 \times \csend(\tau_1) + \crecv(\tau_1) + \ccc_1 + \csend(\tau_2) ; \;
        \\
        \role{s}\mapsto
        \max(\crecv(\tau_2), \csend(\tau_1))
        + \csend(\tau_1) + \crecv(\tau_1) + \ccc_1 + \csend(\tau_2) + \crecv(\tau_2)
      \end{array}
      \right]
    \end{array}
    $
    \end{center}
  }

\end{example}

\begin{example}[Parallel Pipeline]\label{exn:pipe}
  We now show the cost of a fragment of the trace that corresponds with
  $G = \gFix X.
  \Rp \to \Rq \gTy{\tau_1\hasCost \ccc_1}.
  \Rq \to \Rr \gTy{\tau_2 \hasCost \ccc_2}. X$.
  Two possible traces for two iterations of this protocol are as follows:

  {\small
    \vskip.2cm
    \noindent
    \begin{center}
    $
    \begin{array}{@{}l@{}}
      tr_1 =
      \Rp\Rq\tSend\tau_1 \cdot
      \Rp\Rq\tSend\tau_1 \cdot
      \Rp\Rq\tRecv\tau_1 \cdot
      \Rq\tRun \ccc_1 \cdot
      \Rq\Rr\tSend\tau_2 \cdot
      \Rp\Rq\tRecv\tau_1 \cdot
      \Rq\tRun \ccc_1 \cdot
      \Rq\Rr\tSend\tau_2 \cdot
      \Rq\Rr\tRecv\tau_2 \cdot
      \Rr\tRun \ccc_2 \cdot
      \Rq\Rr\tRecv\tau_2 \cdot
      \Rr\tRun \ccc_2
      \\
      tr_2 =
      \Rp\Rq\tSend\tau_1 \cdot
      \Rp\Rq\tRecv\tau_1 \cdot
      \Rq\tRun \ccc_1 \cdot
      \Rq\Rr\tSend\tau_2 \cdot
      \Rq\Rr\tRecv\tau_2 \cdot
      \Rr\tRun \ccc_2 \cdot
      \Rp\Rq\tSend\tau_1 \cdot
      \Rp\Rq\tRecv\tau_1 \cdot
      \Rq\tRun \ccc_1 \cdot
      \Rq\Rr\tSend\tau_2 \cdot
      \Rq\Rr\tSend\tau_2 \cdot
      \Rr\tRun \ccc_2
    \end{array}
    $
    \end{center}
    \vskip.2cm
  }

  \noindent
  In the first trace, $\Rp$ sends first two messages to $\Rq$. Then, $\Rq$
receives, computes them and sends the results to $\Rr$. Finally, $\Rr$ receives
the results, and performs their computation with cost $\ccc_2$. The second trace,
instead, is the repetition of two single iterations of the protocol, where $\Rp$
sends one message, $\Rq$ receives, processes and sends the result to $\Rr$, and
$\Rr$ performs its local computation. Note, however, that since the cost models
assume that each participant runs at a different CPU, the costs of both traces
is the same. To help readability, we name $\Time_\role{p} = \csend(\tau_1)$,
$\Time_\role{q} = \crecv(\tau_1) + \ccc_1 + \csend(\tau_2)$ and $\Time_\role{q} =
\crecv(\tau_2) + \ccc_2$. The trace cost is:

{\small
    \noindent
    \begin{center}
    $
    \begin{array}{l@{}}
      \cost(tr_1)([], [])
      = \cost(\role{p}\role{q}\tRecv\tau_1\cdots)
      ([\role{p}\!\mapsto\!\csend(\tau_1)], [\role{p}\role{q}\!\mapsto\! \csend(\tau_1)])
      = \cost(\role{q}\tRun \ccc_1 \cdots)
      ([\role{p}\!\mapsto\!\csend(\tau_1)]; \role{q}\!\mapsto\! \csend(\tau_1) + \crecv(\tau_1), [])
      \\
      \quad = \left[
      \begin{array}{@{}l@{}}
        \role{p}\mapsto
        2\times \Time_\role{p} ; \;
        \hspace{.2cm}
        \role{q}\mapsto
        \Time_\role{p} + \Time_\role{q} + \max(\Time_\role{p}, \Time_\role{q})
        \;\;
        \role{r}\mapsto
        \Time_\role{p} + \Time_\role{q} + \Time_\role{r} +
        \max(\Time_\role{p}, \Time_\role{q}, \Time_\role{r})
      \end{array}
      \right]
    \end{array}
    $
    \end{center}
  }
  \vskip.2cm

  \noindent
  We can see that the cost is the expected one, where the cost includes the
  initialisation and finalisation of the protocol, where the costs are added, and
  a pipeline steady state, where the cost is the maximum of the costs of each
  participant.
\end{example}

\begin{example}[Dependency Cycle]\label{exn:dep_cy}
  We change slightly the pipeline example, to illustrate what happens to the
trace cost when we introduce a dependency cycle in the protocol. The protocol
that we show below is a recursive ping-pong, where $\Rp$ sends to $\Rq$, and
then $\Rq$ replies to $\Rp$: $\gFix X. \Rp \to \Rq \gTy{\tau_1\hasCost \ccc_1}. \Rq
\to \Rp \gTy{\tau_2 \hasCost \ccc_2}. X$. There is only one possible trace for such
protocol, due to the input/output dependencies between $\Rq$ and $\Rp$ (see
conditions $\Rp, \Rq \notin \subj(\act)$ in Def. \ref{def:lts-global}, e.g.\
[GR4a]). The trace and cost in this instance is:
{\small
\begin{center}
  $
  \begin{array}{@{}c@{}}
    tr =
    \Rp\Rq\tSend\tau_1 \cdot
    \Rp\Rq\tRecv\tau_1 \cdot
    \Rq\tRun \ccc_1 \cdot
    \Rq\Rp\tSend\tau_2 \cdot
    \Rq\Rp\tRecv\tau_2 \cdot
    \Rp\tRun \ccc_2 \cdot
    \Rp\Rq\tSend\tau_1 \cdot
    \Rp\Rq\tRecv\tau_1 \cdot
    \Rq\tRun \ccc_1 \cdot
    \Rq\Rp\tSend\tau_2 \cdot
    \Rq\Rp\tRecv\tau_2 \cdot
    \Rp\tRun \ccc_2
    \\[.1cm]
    \cost(tr)([], [])
      = \left[
      \begin{array}{@{}l@{}}
        \role{p}\mapsto 2 \times (\Time_\Rp + \Time_\Rq)
        \hspace{.2cm}
        \role{q}\mapsto \csend(\tau_1) + \Time_\Rp + 2 \times \Time_\Rq
      \end{array}
      \right]
    \end{array}
    $
  \end{center}
  }
  \noindent
  Here, $\Time_\Rp = \csend(\tau_1) + \crecv(\tau_2) + \ccc_2$ and $\Time_\Rq =
  \crecv(\tau_1) + \ccc_1 + \csend(\tau_2)$. Participant $\Rp$ needs to send
  $\tau_1$, then wait for $\Rq$ to complete its part of the protocol, and then
  receive $\tau_2$ and process it.  Therefore, the cost per iteration is $\Time_\Rp +
  \Time_\Rq$ in all cases.  For participant $\Rq$, the situation is slightly
  different.  A single iteration of $\Rq$ only requires it to wait until $\Rp$
  sends $\tau_1$, and then perform its part of the protocol. Hence, the cost is
  $\csend(\tau_1) + \Time_\Rq$. However, on the next iteration, $\Rq$ needs to wait
  until $\Rp$ finishes with its actions for the previous iteration. This implies
  that the cost of a single iteration for $\Rq$ ($\csend(\tau_1) + \Time_\Rq$) is
  \emph{less} than the average cost per iteration ($\Time_\Rp + \Time_\Rq$).
\end{example}

\subsection{Cost of Global Protocols}
\label{sec:}
We have introduced a way to compute the cost of \emph{one} trace. This cost is
useful to statically analyse the potential execution times of particular
executions of a protocol. However, it is in general not feasible to produce all
possible traces to analyse the cost of a
concurrent/distributed system.
Our \emph{global type cost} addresses this issue,
by providing a syntactic method to estimate an upper bound of the execution
cost.

The global type cost produces, just like Def.\ \ref{def:cost_trace}, a
\emph{cost environment}, with a per-participant estimation.
The protocol will complete when all the
participants have finished their tasks, and so the overall cost is the
maximum of the cost per participant. The global type cost is a function from a
global type, an estimation of the number of iterations for the recursive
protocols, and an initial cost environment.
For proving completeness,
we use a dependency queue as an input to the global type cost,
that will only be used at intermediate
stages of the execution.

\begin{definition}[Global Type Cost]\label{def:global-cost}
Let the maximum operation that combines two cost environments compute a per participant
maximum $\max(\Time, \Time') = [\role{p} \mapsto \max(\Time(\role{p}),
\Time'(\role{p}))]_{\role{p}\in\mathsf{dom}(\Time)\cup\mathsf{dom}(\Time')}$\footnote{The
maximum operation is defined even if $\role{p}$ is not in one of the environments:
recall that $\Time(\role{p}) = 0$ if $\role{p} \notin \mathsf{dom}(\Time)$}. We define
the function $\unroll$, that unrolls the recursive
protocol $\gFix X. G$
$k$ times:

{\small
\begin{center}
$
  \unroll^{k+1}(X,G) = [\unroll^k(X, G) / X] G
  \hspace{.4cm}
  \unroll^0(X, G) = \gEnd
$
\end{center}
}

Then the \emph{global type cost} is defined recursively on the structure of
global types:\\[1mm]
{\small
$\begin{array}{l}
\bullet \ \ \cost(\role{p}\gMsg \role{q} \gTy{\tau\hasCost \ccc}. G, \vec{k}) (\Time,
\Queue)
= \cost(G, \vec{k}) (\Time[\role{p} \addcost \csend(\tau)][\role{p} | \role{q} \addcost \crecv(\tau) + c], \Queue)
   \\[1mm]
\bullet \ \   \cost(\role{p}\gMsg \role{q} \{l_{i}. G_i\}_{i\in I}, \vec{k}) (\Time, \Queue)
= \max\{\cost(G_i, \vec{k}) (\Time[\role{p} \addcost \csend(1)][\role{p} | \role{q} \addcost \crecv(1)], W)\}_{i\in I}
   \\[1mm]
\bullet \   \ \cost(\gFix X. G, k \cdot \vec{k}) (\Time, \Queue)
   = \cost(\unroll^k(X, G), \vec{k}) (\Time, \Queue)
   \\[1mm]
\bullet \   \ \cost(\gEnd) (\Time, \Queue) = (\Time, \Queue)
\end{array}
$
}

\noindent
For completeness, and for the proofs, we define the cost rules for the extended
global types used in the semantics.

\noindent
{\small
$\begin{array}{l}
\bullet \ \   \cost(\role{p}\gMsgt \role{q} \gTy{\tau\hasCost \ccc}. G ,\vec{k}) (\Time, \Queue[\role{p}\role{q} \mapsto \ccc_\role{p} \cdot w])
= \cost(G, \vec{k}) (\Time[\ccc_\role{p} | \role{q} \addcost \crecv(\tau) + c], \Queue[\role{p}\role{q} \mapsto w])
   \\[1mm]
\bullet \  \  \cost(\role{p}\gMsgt \role{q} \; j \; \{l_{i}. G_i\}_{i\in I}, \vec{k}) (\Time, \Queue[\role{p}\role{q} \mapsto \ccc_\role{p} \cdot w])
 = \cost(G_j, \vec{k}) (\Time[\ccc_\role{p} | \role{q} \addcost \crecv(1)], \Queue[\role{p}\role{q} \mapsto w])
   \\[1mm]
\bullet \  \  \cost(\role{p}\gEval (\tau\hasCost \ccc). G, \vec{k})(\Time, \Queue)
 = \cost(G, \vec{k}) (\Time[\role{p} \addcost c], \Queue)
   \\[1mm]
\end{array}
$
}

\end{definition}

\noindent
We write $\cost(G, \vec{k})$ to represent $\cost(G, \vec{k})([], [])$.
Since the dependency queue is only used in the definitions for the intermediate
stages of the execution ($\gMsgt$), we can write
$\cost(G, \vec{k})(T)$. When we compare the output of the cost
functions, we refer to the per-participant cost, ignoring the dependency queue:
$(\Time, \Queue) \leq (\Time', \Queue')$,
$\forall \role{p} \in \Time, \Time(\role{p}) \leq \Time(\role{p}')$.

The first rule in Def. \ref{def:global-cost} explains the cost of an interaction
from $\role{p}$ to $\role{q}$. Participant $\role{p}$ needs to send a message, and this is what the
cost $\Time[\role{p} \addcost \csend(\tau)]$ reflects. Participant $\role{q}$ will receive a
value from $\role{p}$, and then take $\ccc$ time performing a computation. Since $\role{q}$ needs
to wait until $\role{p}$ finishes, we add this dependency to the cost:
$[\role{p} | \role{q} \addcost \crecv(\tau) + \ccc]$.
The cost of a choice is computed similarly, but to produce
an upper bound of the cost of all branches, we compute the maximum cost
per-participant. The cost of the intermediate stages of the execution of the
protocol ($\role{p} \gMsgt \role{q}$) requires accessing the information in $\Queue$, and
retrieving when $\role{p}$ completed the send operation. The cost of a computation
$(\role{p}\gEval(\tau \hasCost \ccc))$ is added to accumulated cost of
participant $\role{p}$.
The cost of a
recursive protocol uses parameter $k$ to first unroll the recursion, and then
compute the cost. We go back to the Examples \ref{exn:sc}, \ref{exn:pipe} and
\ref{exn:dep_cy} and show the computed cost by their global type.

\begin{example}[Scatter/Gather]
  We illustrate the global type cost using the scatter/gather protocol:
  $G =
  \Rp \gMsg \Rq \gTy{\tau_1 \hasCost \ccc_1}.
  \Rp \gMsg \Rr \gTy{\tau_1 \hasCost \ccc_1}.
  \Rq \gMsg \Rs \gTy{\tau_2}.
  \Rr \gMsg \Rs \gTy{\tau_2}.
  \gEnd$.

  {\small%
    \noindent
    \begin{center}
    $
    \begin{array}{@{}l@{}}
      \cost(\Rp \gMsg \Rq \gTy{\tau_1 \hasCost \ccc_1}\ldots)([])
      = \cost(\Rp \gMsg \Rr \gTy{\tau_1 \hasCost \ccc_1}\ldots)%
      \left(
      \left[
      \begin{array}{@{}l@{}}
        \Rp \mapsto \csend(\tau_1); \;
        \Rq \mapsto \csend(\tau_1) + \crecv(\tau_1) + \ccc_1
      \end{array}
      \right]
      \right)
      \\
      = \cost(\Rq \gMsg \Rs \gTy{\tau_2}\ldots)%
      \left(
      \left[
      \begin{array}{@{}l@{}}
        \Rp \mapsto 2 \times \csend(\tau_1); \;
        \Rq \mapsto \csend(\tau_1) + \crecv(\tau_1) + \ccc_1;
        \Rr \mapsto 2 \times \csend(\tau_1) + \crecv(\tau_1) + \ccc_1
      \end{array}
      \right]
      \right)
      \\
      = \cost(\Rr \gMsg \Rs \gTy{\tau_2}\ldots)%
      \left(
      \left[
      \begin{array}{@{}l@{}}
        \Rp \mapsto 2 \times \csend(\tau_1); \;
        \Rq \mapsto \csend(\tau_1) + \crecv(\tau_1) + \ccc_1 + \csend(\tau_2);
        \Rr \mapsto 2 \times \csend(\tau_1) + \crecv(\tau_1) + \ccc_1;
        \\
        \Rs \mapsto \csend(\tau_1) + \crecv(\tau_1) + \ccc_1 + \csend(\tau_2) + \crecv(\tau_2);
      \end{array}
      \right]
      \right)
      \\
      =
      \left[
      \begin{array}{@{}l@{}}
        \Rp \mapsto 2 \times \csend(\tau_1); \;
        \Rq \mapsto \csend(\tau_1) + \crecv(\tau_1) + \ccc_1 + \csend(\tau_2);
        \Rr \mapsto 2 \times \csend(\tau_1) + \crecv(\tau_1) + \ccc_1 + \csend(\tau_2);
        \\
        \Rs \mapsto \csend(\tau_1) +
                    \crecv(\tau_1) + \ccc_1 +
                    \csend(\tau_2) +
                    \max(\crecv(\tau_2), \csend(\tau_1)) + \crecv(\tau_2);
      \end{array}
      \right]
    \end{array}
    $
  \end{center}
  }
\vspace{0.1cm}

  \noindent
  The final cost produced by the global type predicts the same as the one
  taking any possible trace.
\end{example}

\begin{example}[Parallel Pipeline]
  $G = \gFix X. \Rp \to \Rq \gTy{\tau_1\hasCost \ccc_1}. \Rq \to \Rr \gTy{\tau_2 \hasCost \ccc_2}. X$.
  Applying the cost models with $\vec{k}= 2$, $\cost(G, 2)([])$, produces the
  same cost as the trace cost. Particularly, for any arbitrary $k$, $\cost(G, k)([])$
  produces:

  \vskip.2cm
  {\small
    \noindent
    \begin{center}
    $
    \left[
      \begin{array}{@{}l@{}}
        \role{p}\mapsto
        k\times \Time_\role{p} ; \;
        \hspace{.2cm}
        \role{q}\mapsto
        \Time_\role{p} + \Time_\role{q} + (k-1)\times\max(\Time_\role{p}, \Time_\role{q})
        \role{r}\mapsto
        \Time_\role{p} + \Time_\role{q} + \Time_\role{r} +
        (k-1)\times\max(\Time_\role{p}, \Time_\role{q}, \Time_\role{r})
      \end{array}
      \right]
      $
    \end{center}
    }
\end{example}

\begin{example}[Dependency Cycle]\label{exn:pingpong-cost}
  $G = \gFix X. \Rp \to \Rq \gTy{\tau_1\hasCost \ccc_1}.
  \Rq \to \Rp \gTy{\tau_2 \hasCost \ccc_2}. X$.
  For any arbitrary $k>1$, $\cost(G, k)([])$
  produces the following cost, which corresponds to the trace cost:

\vskip.2cm
  {\small
    \noindent
    \begin{center}
    $
    \begin{array}{l@{}}
      \left[
      \begin{array}{@{}l@{}}
        \role{p}\mapsto k \times (\Time_\Rp + \Time_\Rq)
        \hspace{.2cm}
        \role{q}\mapsto \csend(\tau_1) + \Time_\Rq + (k-1) \times (\Time_\Rp + \times \Time_\Rq)
      \end{array}
      \right]
    \end{array}
    $
  \end{center}
  }
  \vskip.2cm

\end{example}

\noindent
We showed in the previous examples that
function $\cost(G)$ accurately predicts an upper
bound of the cost obtained from any trace of the protocol. We formalise this
statement below in Theorem \ref{thm:global-cost-sound}, and provide a full proof in
the Appendix \ref{app:global-cost-sound}. In the formalisation, we use
$\unroll(G, \vec{k})$ to unroll all recursion variables, using the parameters
$\vec{k}$, i.e. $\unroll(G, \vec{k})$ is defined recursively on $G$, with the
only interesting case $\unroll(\gFix X. G, k \cdot \vec{k}) = \unroll^k(X,
(\unroll(G, \vec{k})))$. Function $\unroll(G, \vec{k})$ is only defined if there
are enough the size of $\vec{k}$ is that of the amount of recursion variables in
$G$.

\begin{definition}[Well-Formedness of Dependency Queues]\label{def:wf_dep_queue}
  A dependency queue $\Queue$ is well formed with respect to a global type $G$ if it
  \emph{only} contains the values required to compute the cost of $G$.

\vskip.2cm
$\begin{array}{r l}
   \WF(G, \Queue[\Rp\Rq \mapsto w])
   & \Longrightarrow \WF(\Rp \gMsgt \Rq \gTy{\tau \hasCost \ccc}. G, \Queue[\Rp\Rq \mapsto \ccc_\Rp \cdot w])
   \\
   \WF(G, \Queue[\Rp\Rq \mapsto \epsilon])
   & \Longrightarrow \WF(\Rp \gMsg \Rq \gTy{\tau \hasCost \ccc}. G, \Queue[\Rp\Rq \mapsto \epsilon])
   \\
   \WF(G, \Queue)
   & \Longrightarrow \WF(\Rp \gEval \gTy{\tau \hasCost \ccc}. G, \Queue)
   \\
   j \in I \wedge \WF(G_j, \Queue[\Rp\Rq \mapsto w])
   & \Longrightarrow \WF(\Rp \gMsgt \Rq \; j \; \{\aInj{i}. G_i\}_{i \in I}, \Queue[\Rp\Rq \mapsto \ccc \cdot w])
   \\
   \forall (i \in I),  \WF(G_i, \Queue[\Rp\Rq \mapsto \epsilon])
   & \Longrightarrow \WF(\Rp \gMsg \Rq \; \{\aInj{i}. G_i\}_{i \in I}, \Queue[\Rp\Rq \mapsto \epsilon])
   \\
   & \phantom{\Longrightarrow\;} \WF(\gEnd, [])
 \end{array}$
 \vskip.2cm

 \noindent
 We generally write $\WF(G, (\Time, \Queue)) = \WF(G, \Queue)$.

\end{definition}

\begin{restatable}[Preservation of $\WF$]{lemma}{queueWf}\label{lm:step_wf}
  If $\WF(G, (\Time, \Queue))$ and $G \xrightarrow{\ell} G'$, then $\WF(G' ,
  \cost(\ell)(\Time, \Queue))$.
\end{restatable}

\begin{proof}
  By induction on the structure of $G \xrightarrow{\ell} G'$. See Appendix
\ref{app:global-cost-sound}.
\end{proof}

This lemma states that if $\Queue$ is well-formed with respect to $G$, and $G'$
results from taking a step in $G$, then the queue that results from
$\cost(\ell)(\Time, \Queue)$ is well formed with respect to $G'$.

\begin{restatable}[Cost Preservation]{lemma}{costPreserv}\label{lm:cost_preservation}
  If $G\xrightarrow{\ell} G'$, then $\cost(G')(\cost(\ell)(\Time, \Queue)) \leq
\cost(G)(\Time, \Queue)$.
\end{restatable}

\begin{proof}
  By induction on the structure of the derivation of $G \xrightarrow{\ell} G'$.
\end{proof}

This is the main lemma, that states that if $G$ transitions to $G'$ with action
$\ell$, then, given an initial cost/queue $(\Time, \Queue)$, the cost of $G'$ on
an initial cost after running $\ell$ on $(\Time, \Queue)$ is less or equal than
the cost of $G$ with initial cost $(\Time, \Queue)$. The reason why this cost is
less or equal, rather than equal, is that a branching may take a lower cost path
in the protocol.

\begin{restatable}[Bounded-Cost Soundness]{theorem}{globalcostsound}%
  \label{thm:global-cost-sound}
  If $\unroll(G, \vec{k}) \xrightarrow{\vec{\act}} \gEnd$, then
  $\cost(\vec{\act}) \leq \cost(G, \vec{k})$.
\end{restatable}

\begin{proof}
We prove the following generalised statement.  If $G \xrightarrow{\vec{\ell}}
\gEnd$ and $\WF(G, (\Time, \Queue))$ then $\cost(\vec{\ell}) \leq \cost(G)$. To
recover the original statement, we need to specialise this statement with $G =
\unroll(G',\vec{k})$, and $\Queue=[]$.  By induction on the length of
$\vec{\ell}$:

\noindent
\underline{\textbf{Case} $\vec{\ell}=\epsilon$}: $G \xrightarrow{\epsilon} G'$
implies that $G = G'$, therefore $G = \gEnd$. $\cost(\epsilon) (\Time, \Queue) =
\cost(\gEnd) (\Time, \Queue) = (\Time, \Queue)$.

\noindent
\underline{\textbf{Case} $\vec{\ell}=\ell_1\cdot\vec{\ell'}$}: $G
\xrightarrow{\ell_1\cdot\vec{\ell'}} G'$ implies that there is a $G''$ s.t. $G
\xrightarrow{\ell_1} G'' \xrightarrow{\vec{\ell'}} G'$. By Lemma
\ref{lm:step_wf}, we know that $\WF(G'',(\cost(\ell_1)(\Time,
\Queue)))$. Therefore, by the IH, $\cost(\vec{\ell'})(\cost(\ell_1)(\Time,
\Queue)) \leq \cost(G'', \cost(\ell_1)(\Time, \Queue))$.  The proof is completed
by Lemma \ref{lm:cost_preservation}, that allows us to derive that $\cost(G'',
\cost(\ell_1)(\Time, \Queue)) \leq \cost(G)(\Time, \Queue)$.
\end{proof}

%

\section{Cost for Multiparty Session Protocols (2): Latency of Recursion}
\label{sec:thro}

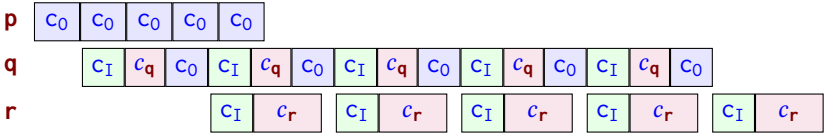
\begin{figure}
    \begin{tikzpicture}
        [ step1/.style={rectangle,minimum height=0.5cm, minimum width=.6cm,draw},
          step2/.style={rectangle,minimum height=0.5cm ,minimum width=.3, draw},
          step3/.style={rectangle,minimum height=0.5cm ,minimum width=.3, draw}
        ]
      \node[minimum width=0.6cm] (p0) at (0, 1.2) { $\Rp$ } ;
      \node[minimum width=0.6cm] (q0) at (0,  .6) { $\Rq$ } ;
      \node[minimum width=0.6cm] (r0) at (0, 0  ) { $\Rr$ } ;

      \node[step1, right=0cm of p0, fill=blue!10] (p1) { $\csend$ };
      \foreach \i [evaluate=\i as \next using int(\i+1)] in {1,...,4}{%
        \node[step1, right=0cm of p\i, fill=blue!10] (p\next) { $\csend$ };
      };

      \node (q1) [right=.39cm of q0] {};
      \foreach \i
        [evaluate=\i as \next using int(\i+1)] in {1,...,5}{%
        \node[step2, right=0cm of q\i, fill=green!10] (q1\next) { $\crecv$ };
        \node[step2, right=0cm of q1\next, fill=purple!10] (q2\next) { $\ccc_\Rq$ };
        \node[step2, right=0cm of q2\next, fill=blue!10] (q\next) { $\csend$ };
      };

      \node (r1) [right=1cm of r0] {};
      \foreach \i
        [evaluate=\i as \next using int(\i+1)] in {1,...,5}{%
        \node[step3, right=1.1cm of r\i, fill=green!10] (r\next) { $\crecv$ };
        \node[step3, minimum width=0.9cm, right=0cm of r\next, fill=purple!10] {$\ccc_\Rr$};
      };

    \end{tikzpicture}
  \caption{Latency of a parallel pipeline.}
  \label{fig:pipe-thro}
\end{figure}

Previous section presented cost models for multiparty session protocols with
bounded recursion. In this section, we extend the cost models for multiparty
session protocols with two notions:
\begin{enumerate}
  \item The \emph{average cost per iteration} of a protocol, which we call
    \textbf{\textsl{latency}} ($\costw(G)$).
  \item The \textbf{\textsl{latency relative to $\role{p}$}}, denoted by
    $\costw_{\role{p}}(G)$, as the latency of a global type, divided by the number of
    messages exchanged by participant $\role{p}$ per iteration.
\end{enumerate} These cost models are useful for scenarios where we do not know
how many iterations the protocol is going to run, or this number of iterations
is large. For example, consider the protocol for a parallel program following a
master-worker pattern, where the master ($\role{m}_1$) distributes a stream of
tasks to a series of workers ($\role{w}_i$), and then collects the results
($\role{m}_2$):
\[
  \small
\begin{array}{@{}l@{}}
mw(n) = \gFix X.
\role{m}_1 \gMsg \{\role{m}_2, \role{w}_1, \ldots, \role{w}_n \}
  \\ \hspace{2cm}
\left\{
  \begin{array}{@{}l@{}}
    \aInj{1}. \;
    \begin{array}[t]{@{}l@{}}
      \role{m}_1 \gMsg \role{w}_1 \gTy{\tau \hasCost \ccc}.
      \ldots
      \role{m}_1 \gMsg \role{w}_n \gTy{\tau \hasCost \ccc}.
      \role{w}_n \gMsg \role{m}_2 \gTy{\tau \hasCost \ccc}.
      \ldots
      \role{w}_1 \gMsg \role{m}_2 \gTy{\tau \hasCost \ccc}.
      X
    \end{array}
    \\
    \aInj{2}. \; \gEnd
  \end{array}
\right\}
\end{array}
\]
\noindent
When such protocols are run in practice, they are aimed at speeding up
some computation on a large, potentially unbounded, stream of tasks. Therefore,
computing $\cost(mw(n), k)$ can be computationally very expensive, or impossible if
$k$ is unknown. These scenarios are where the average cost per iteration, or
latency ($\costw(G)$) is more useful. The key property of $\costw(G)$ is that
approximates to $\cost(G, k)/k$ as $k$ grows.  In the protocol above, it is
clear that $\costw(mw(n)) > \costw(mw(m))$ if $n > m$, due to the greater number
of interactions that involve $\role{m}_1$ and $\role{m_2}$ per
iteration. However, unless the cost of the extra interactions outweigh the cost
of a computation performed by $\role{w}_i$, it is preferable to use $mw(n)$ than
$mw(m)$, subject to the available resources. This is where the latency
relative to a participant is useful, since it provides a better
measurement about how fast is $mw(n)$ processing tasks.

We explain the intuition behind $\costw(G)$ and $\costw_{\role{p}}(G)$ using the
master-worker protocol, and $\cost(mw(n), k)$ from \S \ref{sec:cost}. For
simplicity, we omit branching:
\[
\small
\begin{array}{ll}
mw(n) = \gFix X.\!\!\!&
      \role{m}_1 \gMsg \role{w}_1 \gTy{\tau \hasCost \ccc}.
      \ldots
      \role{m}_1 \gMsg \role{w}_n \gTy{\tau \hasCost \ccc}.
      \role{w}_n \gMsg \role{m}_2 \gTy{\tau \hasCost \ccc}.
      \ldots
      \role{w}_1 \gMsg \role{m}_2 \gTy{\tau \hasCost \ccc}.
      X
\end{array}
\]
\noindent
To simplify the calculations, we also assume that all $\role{w}_i$ only do their
actions after $\role{m}_1$ has finished sending all tasks, although the cost
models in \S \ref{sec:cost} will predict a lower cost for $\role{w_i}$ than
$\role{w_j}$ if $i < j$, since they only need to wait until their required data
has been sent. We start with $mw(2)$ and $k=1$:

\noindent
{
\small
$
\delimiterfactor=500
\delimitershortfall=50pt
\cost(mw(2), 1)
\leq
\left[
\begin{array}{@{}l@{\;}l@{}}
  \role{m_1} & \mapsto 2 \times \overbracket[.5pt][2pt]{\csend(\tau_1)}^{\Time_1};
  \role{w_i} \mapsto 2 \times \Time_1 + \overbracket[.5pt][2pt]{\crecv(\tau_1) + \ccc_1 + \csend(\tau_2)}^{\Time};
  \role{m_2} \mapsto 2 \times \Time_1 + \Time + 2 \times \underbracket[0.5pt][2pt]{(\crecv(\tau_2) + \ccc_2)}_{\Time_2}
\end{array}
\right]
$
}

\noindent
We named the relevant parts as $\Time_1$, $\Time$ and $\Time_2$ for readability's
sake.
%
%
For any arbitrary $k>1$, we compute the cost $\cost(mw(2), k)$ as:
\[
  \small
\cost(mw(2), k)
\!\! \leq \!\!
\left[
  \begin{array}{@{}l@{\;}l@{}}
  \role{m_1} & \mapsto 2 \times \Time_1 + (k - 1) \times 2 \times \Time_1;
  \role{w_i} \mapsto 2 \times \Time_1 + \Time + (k - 1) \times \max(2 \times \Time_1, \Time);
  \\
  \role{m_2} & \mapsto 2 \times \Time_1 + \Time + 2 \times \Time_2 +
  (k - 1) \times \max(2 \times \Time_1, \Time, 2 \times \Time_2)
\end{array}
\right]
\]
There are two parts in this cost that can be distinguished, the \emph{fixed}
cost that corresponds to the initial and final stages of the protocol:
$ [
  \role{m_1} \mapsto 2 \times \Time_1; \;
  \role{w_i} \mapsto 2 \times \Time_1 + \Time; \;
  \role{m_2} \mapsto 2 \times \Time_1 + \Time + 2 \times \Time_2]
$, and the \emph{latency}, which is the cost that increases the more
iterations we take.
%
%
In general, for an arbitrary $n$, the latency is:
\[
\small
\begin{array}{ll}
  \costw(mw(n))
  = &
  [
    \role{m_1} \mapsto n \times \Time_1; \;
    \role{w_i} \mapsto \max(n \times \Time_1, \Time);
    \role{m_2} \mapsto \max(n \times \Time_1, \Time, n \times \Time_2)
  ]
\end{array}
\]
\noindent
If we keep increasing the number of workers, the latency will indicate a
greater cost. However, in this particular protocol what matters is the cost
\emph{per message interaction of} $\role{m}_i$, which are the workers that
respectively distribute tasks and collect
the results. We use $\costw_{\role{m_2}}(mw(n)) = \costw(mw(n)) / i$,
where $i$ is the number of message exchanged by $\role{m_2}$ per iteration:
\[
\small
  \costw_{\role{m}_2}(mw(n))
  =
  \left[
    \begin{array}{@{}l@{}}
    \role{m_1} \mapsto \Time_1; \;
    \role{w_i} \mapsto \max(\Time_1, \Time / n);
    \role{m_2} \mapsto \max(\Time_1, \Time / n, T_2)
    \end{array}
  \right]
\]
\noindent Since $\costw_{\role{m}_2}(mw(n))$ is less than
$\costw_{\role{m}_2}(mw(m))$ if $m < n$, then the latency relative to
$\role{m_2}$ is a
better measurement to compare how \emph{fast} a protocol processes tasks. In the
remainder of this section, we define $\costw$ and $\costw_{\role{p}}$; and prove
that they approximate $\cost(G, k)$ for a sufficiently large $k$.

\subsection{Latency of Nested Recursive Protocols}
\label{sec:nested}
The master-worker protocol above contains only one recursion variable. In
general, recursive protocols can have multiple nested recursive
sub-protocols. Intuitively, to compute $\costw$ of $\mu X. G$, we need to
estimate the total execution time of a single iteration of $G$. If $G$ contains
recursive sub-protocols, this implies that we need to know how many iterations
they will run, before recursive variable $X$ is found. We illustrate this with
the recursive global type below:
\[
\small
\begin{array}{@{}l@{}}
  G = \mu X . \; \Rp \gMsg \Rq \gTy{\tau \hasCost \ccc_1} .
  \mu Y. \; \Rq \gMsg \Rp \left\{l_1. \; Y; \; l_2. \; X \right\}
\end{array}
\]
\noindent
To compute $\costw(G)$, we need to know how many times the branch that ends in
recursion variable $Y$ will be taken. Since this depends on the particular
implementation of the protocol, we parameterise such recursion variables with
some $k$, and defer its instantiation. To produce an equation to estimate the
latency that is parametric in this $k$, we split the protocol $G$ into two
sub-protocols:
\[
\small
\begin{array}{@{}l@{}}
  G_Y = \mu Y . \; \Rq \gMsg \Rp \left\{l_1. \; Y; \; l_2. \; Z \right\}
\quad
  G_X = \mu X . \; \Rp \gMsg \Rq \gTy{\tau \hasCost \ccc_1} . [X/Z] G_Y
\end{array}
\]
\noindent
If $\costw(G_Y)$ contains another recursion variable, then we keep splitting it
until we have a set of global types, each of which defined using at most one
bound recursion variable.  To compute $\costw(G_X)$ we require a parameter $k$,
and we will use $k \times \costw([\gEnd/Z] G_Y)$ for the cost of any participant
in the inner sub-protocol. In the remainder of this section, we focus on
recursive protocols with at most one recursion variable.

\subsection{Cost Recurrences}
\label{sec:recurrences}

Computing $\costw$ is done in two steps. First, we build a system of
recurrence equations, $\Time(n)$, that capture the execution costs after $n$
iterations of the recursive protocol. Then, we build the difference equations
$\Delta(n)$, where $\Delta(n)(\Rp) = \Time(n+1)(\Rp) - \Time(n)(\Rp)$, and
estimate the value of $\Delta(n)$, as $n$ grows. We observe that, for the
recurrences that we generate, with $n \geq 2$, $\Delta(n)$ stabilises.

\begin{definition}[Cost Recurrences]\label{def:cost-recur}
  We use $\cost(G)$ from Def. \ref{def:global-cost}.
  Given a recursive global type $\gFix X. G$, we define its cost recurrence, $\Time(n)$, as follows:
  $\Time(n+1) = \cost([\gEnd/X]G, \Time(n))$, $\Time(0) = []$.
\end{definition}

Consider the following parallel pipeline:
\[
\small
\begin{array}{@{}l@{}}
  G = \mu Y . \;
  \Rp \gMsg \Rq \gTy{\tau_1 \hasCost \ccc_1}. \;
  \Rq \gMsg \Rr \gTy{\tau_2 \hasCost \ccc_2}. \;
  Y
  \end{array}
\]
We show below an example of the system of recurrence equations that we
generate. We take the resulting cost environment, and we produce a different
equation $\Time_\Rp$, $\Time_\Rq$, and $\Time_\Rr$ for every participant in the
protocol:
\[
\small
\begin{array}[t]{@{}r@{\;}l@{}}
  \Time_\Rp(n+1)
   = \Time_\Rp(n) + \csend(\tau_1) \quad
  \Time_\Rq(n+1)
  & = \max (\Time_\Rp(n+1), \Time_\Rq(n)) + \crecv(\tau_1) + \ccc_1 + \csend(\tau_2)
  \\
  \Time_\Rr(n+1)
  & = \max (\Time_\Rq(n+1), \Time_\Rr(n)) + \crecv(\tau_2) + \ccc_2
  \end{array}
\]

\begin{definition}[Cost Difference Equations]\label{def:delta}
  Given a recursive global type $\gFix X. G$, with cost recurrence $\Time$, we
define its cost difference equation $\Delta$, as $\Delta(n) = \Time(n+1) -
\Time(n)$.
\end{definition}

The cost difference equation provides an estimate on how much the cost increases
for each participant after running the protocol one additional iteration.

\begin{definition}[Latency per Iteration]\label{def:latency}
  The latency of a recursive protocol $\gFix X. G$ with cost difference
  $\Delta(n)$ is defined as the cost expression $\ccc$ that is the least upper bound
  of the difference equation $\Delta(n)$, for a sufficiently large $n$:
  \[
  \begin{array}{@{}l@{}}
    \costw(\mu X. G) =
    \ccc \quad \text{s.t.}\; \exists k, \forall n \geq k, \; \ccc \geq \Delta(n)
  \end{array}
\]
\end{definition}
Suppose that we want to compute the latency of the previous parallel pipeline.
On average, excluding the initialisation of the protocol, the latency for
$\Rr$ must be the maximum of the times for $\Rp$, $\Rq$ and $\Rr$, as usual in
parallel pipelines. This is because the actions of $\Rp$, $\Rq$ and $\Rr$ are
independent \emph{across iterations}.
The solution of $\Delta(0)$ shows that the cost is the addition of all
individual costs. However, by solving $\Delta(1)$, we obtain the expected
result, where $\Time_\Rp = \csend(\tau_1)$, $\Time_\Rq = \max(\Time_\Rp, \crecv(\tau_1) +
\ccc_1 + \csend(\tau_2))$, and $\Time_\Rr = \max(\Time_\Rq, \crecv(\tau_2) + \ccc_2)$.

When the actions of a recursive protocol are not independent across iterations,
i.e.\ the send/receive dependency graph forms a cycle, then all participants
will need to synchronise. An example of this is the protocol:
\[
\small
\begin{array}{@{}l@{}}
  G = \mu Y . \;
  \Rp \gMsg \Rq \gTy{\tau_1 \hasCost \ccc_1}. \;
  \Rq \gMsg \Rp \gTy{\tau_2 \hasCost \ccc_2}. \;
  Y
\end{array}
\]
\noindent
In the first iteration, we will have that $\Rp$ sends $\tau_1$ to $\Rq$, which
needs to wait for the message, and then takes $\ccc_1$ time. At this point, we have
that $\Rp$ spent $\csend(\tau_1)$, and $\Rq$ took $\csend(\tau_1) +
\crecv(\tau_1) + \ccc_1$. Next, $\Rq$ sends $\tau_2$ to $\Rp$, which is completed
after $\csend(\tau_1) + \Time_{\role{q}}$, where $\Time_{\role{q}} = \crecv(\tau_1) + \ccc_1 +
\csend(\tau_2)$. Then, $\Rp$ needs to receive $\tau_2$ and take $\ccc_2$ of local
computation time. Since the accumulated time by $\Rp$ is $\csend(\tau_1) <
\csend(\tau_1) + \Time_{\role{q}}$, we increase the total time spent by $\Rp$: $\Time_{\role{q}} + \Time_{\role{p}}$,
where $\Time_{\role{p}} = \csend(\tau_1) + \crecv(\tau_2) + \ccc_2$.  In the next iteration, we
have that $\Rp$ takes $\Time_{\role{p}} + \Time_{\role{q}} + \csend(\tau_1)$.  Next, $\Rq$ takes
$\max(\csend(\tau_1) + \Time_{\role{q}}, \Time_{\role{p}} + \Time_{\role{q}} + \csend(\tau_1)) + \Time_{\role{q}} = \Time_{\role{p}} + \Time_{\role{q}} +
\csend(\tau_1) + \Time_{\role{q}}$, and finally $\Rp$ will take $2 \times (\Time_{\role{p}} + \Time_{\role{q}})$. After
$k$ iterations, the cost for $\Rp$ is $k \times (\Time_{\role{p}} + \Time_{\role{q}})$,
while the cost for $\Rq$ is $(k-1) \times (\Time_{\role{p}} +
\Time_{\role{q}}) + \csend(\tau_1) + \Time_{\role{q}}$, which approximates $\Time_{\role{p}} + \Time_{\role{q}}$.

\begin{definition}[Latency with respect to $\Rp$]
  We define $\costw_\Rp(\gFix X. G) = \costw(\gFix X. G)(\Rp) / \mathsf{count}(\Rp,
G)$, where $\mathsf{count}(\Rp, G)$ is the number of interactions of $G$ in
which $\Rp$ occurs.
\end{definition}

\subsection{Correctness}
We guarantee that the latency correctly approximates the bounded cost of a
protocol. Moreover, given an arbitrary trace that is the
result of a $k$-unrolling of a recursive global type $G$, $k \times \costw(G)$
will approximate the cost of the full trace.

\begin{restatable}[Cost Latency Correspondence]{theorem}{costThro}\label{thrm:cost-thro}
  Given a sufficiently large $k_2$, for all $k_1 > k_2$, $\cost(G, k_1) - \cost(G, k_2) \leq (k_1 - k_2) \times \costw(G)$.
\end{restatable}

This result follows directly from our definition of $\costw$, since the latency
approximates $\Delta(n)$ (with $\Delta(n) = \Time(n+1) - \Time(n)$) for a
sufficiently large $n$, and that $\Time(n+1)$ is the recurrence that
approximates $\cost(G, n)$. We need to show that $k \times \Delta(n) =
\Time(n+k) - \Time(n)$, and then take $k_2 = n$ and $k_1 = n + k$.

\begin{proposition}\label{prop:recur}
  Given $\mu X. G$, let $\Time(n+1) = \cost([\gEnd/X] G, \Time(n))$ and $\Time(0)=[]$. Then,
$\cost(\mu X. G, n) = T(n)$.
\end{proposition}

\begin{proof}
  By induction on $n$, the base case is straightforward: $\Time(0) = \cost(\mu
  X. G, 0) = \cost(\gEnd) = []$. If $n = m + 1$, then $\Time(m + 1) = \cost([\gEnd/X] G, \Time(m))
  = \cost([\gEnd/X] G, \cost(\mu X. G, m)) = \cost(\mu X. G, m + 1)$.
\end{proof}

Proposition \ref{prop:recur} states that given a recursive protocol,
instantiating its recurrence with some number $n$ yields the same cost as
unrolling the protocol $n$ times and computing its cost. We use Proposition
\ref{prop:recur} in combination with Definition \ref{def:delta} to derive the
following. Assume $\Delta$ is the difference equation for recursive protocol
$G$. Then, the following equality holds:
\begin{equation}\label{eqn:recur}
\Delta(n) = \cost(G, n + 1) - \cost(G, n)
\end{equation}

\begin{restatable}[Latency Soundness]{theorem}{throSoundness}\label{thm:thro-sound}
  There exists $k'$ such that for all $k$, if $\unroll^{k}(G) \xrightarrow{\vec{\act}}\gEnd$,
  then $\cost(\vec{\act}) \leq k \times \costw(G) + k'$.
\end{restatable}

\begin{proof}
  By Definition \ref{def:delta}, we know that there exists some $k_0$ such that
  for all $n \geq k_0$,
\begin{equation}\label{eqn:costw-delta}
\costw(G) \geq \Delta(n).
\end{equation}

We show that $k'$ is $\cost(G, k_0)$.  By Theorem
\ref{thm:global-cost-sound}, we know that $\cost(\vec{\act}) \leq \cost(G,
k)$. Therefore, it is sufficient to show that for all $k$, $\cost(G, k) \leq k
\times \costw(G) + \cost(G, k_0)$. We proceed by case analysis:

\noindent
\textbf{Case} $k \leq k_0$ straightforward, since $\cost(G, k) \leq \cost(G,
k_0)$ if $k \leq k_0$.

\noindent
\textbf{Case} $k > k_0$:  By induction on $k$. All
cases $\leq k_0$ are straightforwardly true.

$\bullet$ Case $k = k_0 + 1$ follows from $\cost(G, k_0) \leq \cost(G, k_0 + 1)$:
\[
  \begin{array}{@{}r@{\;}l@{\hspace{-.3cm}} r@{}}
    \cost(G, k_0)
    & \leq \cost(G, k_0 + 1)
    & {\small\{\text{multiply $k_0$}\}}
    \\
    k_0 \times \cost(G, k_0)
    & \leq k_0 \times \cost(G, k_0 + 1)
    & {\small\{\text{add $\cost(G, k_0 + 1)$}\}}
    \\
    \cost(G, k_0 + 1) + k_0 \times \cost(G, k_0)
    & \leq \cost(G, k_0 + 1) + k_0 \times \cost(G, k_0 + 1)
    & {\small\{\text{sub $k_0 \times \cost(G, k_0)$}\}}
    \\
    \cost(G, k_0 + 1)
    & \leq (k_0 + 1) \times \cost(G, k_0 + 1) - k_0 \times \cost(G, k_0)
    & {\small\{\text{cancel $\cost(G, k_0)$}\}}
    \\
    \cost(G, k_0 + 1)
    & \leq (k_0 + 1) \times (\cost(G, k_0 + 1) - \cost(G, k_0)) + \cost(G, k_0)
    & {\small\{\text{by (\ref{eqn:recur}) and (\ref{eqn:costw-delta})}\}}
    \\
    \cost(G, k_0 + 1)
    & \leq (k_0 + 1) \times \costw(G) + \cost(G, k_0)
    &
  \end{array}
\]

$\bullet$ $k = k_2 + 1$, with $k_2 > k_0.$ Assume the induction hypothesis
$\cost(G, k_2) \leq k_2 \times \Delta(k_2) + \cost(G, k_0)$:
\[
  \begin{array}{@{}r@{\;}l r@{}}
    \cost(G, k_2)
    & \leq k_2 \times \Delta(k_2) + \cost(G, k_0)
    & {\small\{\text{by (\ref{eqn:recur})}\}}
    \\
    \cost(G, k_2)
    & \leq k_2 \times \cost(G, k_2 + 1) - k_2 \times \cost(G, k_2) + \cost(G, k_0)
    & {\small\{\text{add $\cost(G, k_2 + 1)$}\}}
    \\
    \cost(G, k_2 + 1) + \cost(G, k_2)
    & \leq (k_2 + 1) \times \cost(G, k_2 + 1) - k_2 \times \cost(G, k_2) + \cost(G, k_0)
    & {\small\{\text{sub $\cost(G, k_2)$}\}}
    \\
    \cost(G, k_2 + 1)
    & \leq (k_2 + 1) \times (\cost(G, k_2 + 1) - \cost(G, k_2)) + \cost(G, k_0)
    & {\small\{\text{by (\ref{eqn:recur}) and (\ref{eqn:costw-delta})}\}}
    \\
    \cost(G, k_2 + 1)
    & \leq (k_2 + 1) \times \costw(G) + \cost(G, k_0)
    &
  \end{array}
\]
\end{proof}

This implies that the latency approximates the cost of a trace of a
$k$-unrolling of a recursive protocol, and it follows from Theorems
\ref{thrm:cost-thro} and \ref{thm:global-cost-sound}. To illustrate this,
consider the
average cost per recursion iteration, $\cost(\vec{\act}) / k$. By Theorem
\ref{thm:thro-sound}, we know that $\cost(\vec{\act}) / k \leq \costw(G) + k' /
k$. Since $k'$ does not depend on $k$, for a sufficiently large $k$, the term
$k' / k$ will become smaller, and the upper bound of $\cost(\vec{\act}) / k$
will be approximately $\costw(G)$.

\section{Asynchronous Message Optimisation}
\label{sec:extn}
This section illustrates one of the key features of \CAMP, the
formulation and its soundness of \emph{asynchronous message optimisations}.
We extend the cost equations in \S\ref{sec:cost} and \S\ref{sec:thro} to
tackle protocols in which certain actions have been permuted for optimisation
purposes. Parallel programs often make use of parallel pipelines to overlap
computation and communication, as far as the overlapping does not interfere
with data dependencies. The overlapping can reduce stall time due to blocking
wait in the asynchronous communication model. Under the \CAMP{} theory,
optimisation should preserve the deadlock-freedom and produce the same
outcome, while ensuring less cost for the overall calculation.

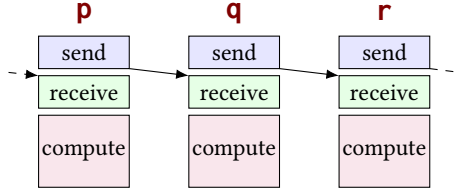
\begin{wrapfigure}{r}{.5\textwidth}
  \vskip-.6cm
  \begin{tikzpicture}[every node/.style={font=\small,
                                         minimum height=0.42cm,
                                         minimum width=1.2cm}]

   \node [matrix, very thin,column sep=0.4cm,row sep=0.1cm] (matrix) at (0,0) {
    &
    \node(0,0) (stage10) {}; & &
    \node(0,0) (stage20) {}; & &
    \node(0,0) (stage30) {}; & \\
    &
    \node(0,0) (stage11) {}; & &
    \node(0,0) (stage21) {}; & &
    \node(0,0) (stage31) {}; & \\
    &
    \node(0,0) (stage12) {}; & &
    \node(0,0) (stage22) {}; & &
    \node(0,0) (stage32) {}; & \\
    &
    \node(0,0) (stage13) {}; & &
    \node(0,0) (stage23) {}; & &
    \node(0,0) (stage33) {}; & \\
    &
    \node(0,0) (stage14) {}; & &
    \node(0,0) (stage24) {}; & &
    \node(0,0) (stage34) {}; & \\
  };

  \fill
    (stage10) node[fill=white] {\Large $\Rp$}
    (stage20) node[fill=white] {\Large $\Rq$}
    (stage30) node[fill=white] {\Large $\Rr$};

  \filldraw[fill=blue!10]
    (stage11.north west) rectangle (stage11.south east)

    (stage21.north west) rectangle (stage21.south east)

    (stage31.north west) rectangle (stage31.south east)
    ;
  \filldraw[fill=green!10]
    (stage12.north west) rectangle (stage12.south east)

    (stage22.north west) rectangle (stage22.south east)

    (stage32.north west) rectangle (stage32.south east)
    ;
  \filldraw[fill=purple!10]
    (stage13.north west) rectangle (stage14.south east)

    (stage23.north west) rectangle (stage24.south east)

    (stage33.north west) rectangle (stage34.south east)
    ;

  \fill
    (stage11) node { send }
    (stage21) node { send }
    (stage31) node { send }
    (stage12) node { receive }
    (stage22) node { receive }
    (stage32) node { receive }
    (stage13) node [yshift=-.25cm] { compute }
    (stage23) node [yshift=-.25cm] { compute }
    (stage33) node [yshift=-.25cm] { compute }
    ;

  \draw [-Latex] (stage11.south east) -- (stage22.north west);
  \draw [-Latex] (stage21.south east) -- (stage32.north west);
  \draw [-, dashed] (stage31.south east) -- ($(stage31.south east) + (.4, -.05)$);
  \draw [-Latex, dashed]  ($(stage11.south west) + (-.4, -.05)$) -- (stage12.north west);
  \end{tikzpicture}
  \vskip-0.4cm
  \caption{Optimised ring trace.}
  \label{fig:ring-async}
\end{wrapfigure}

\noindent
Fig.\ref{fig:ring-async} shows a safe and efficient ring protocol, in which stage
$i$ shares data with stage $(i+1) \mod 3$, and then proceed to do some local
computation. This protocol behaves similarly to that of Fig.
\ref{fig:ring-proto} in \S \ref{sec:mpst}, but the output actions have been
permuted so that they are performed first, thus reducing the amount of
synchronisation required. The optimised version, however, is more difficult
to check against a standard global type, because of the permuted actions.
This can be illustrated by comparing the optimised and un-optimised local
types of $\Rq$:

\vskip.2cm
\begin{center}
$
\begin{array}{@{}r@{\hspace{1cm}}l@{}}
  \begin{array}{@{}l@{}}
    L_{\Rq} = \lFix X. \Rp \lRecv \tau \hasCost \ccc. \Rr \lSend \tau. X
  \end{array}
  &
  \begin{array}{@{}l@{}}
    L_{\Rq}' = \lFix X. \Rr \lSend \tau. \Rp \lRecv \tau \hasCost \ccc. X
  \end{array}
\end{array}
$
\end{center}
\vskip.2cm

\noindent
$L_{\Rq}$ is the unoptimised local type, and $L_{\Rq}'$ is the optimised
version. Both local types represent a similar communication pattern. However,
in the left version $L_{\Rq}$, the send action only happens after receiving,
and computing (with cost $c$), while the right version first sends a value of
type $\tau$, and then performs the receive and local computation. This
removes unnecessary synchronisation, and allows $\Rr$ to continue with its
interactions before $\Rq$ finishes its own local computation.

Only certain message permutations are valid. For example, if instead of swapping the
send and receive actions for $L_{\Rq}$, we permute the actions for participant $\Rp$,
then we end up in the following (incorrect) situation:

\vskip.2cm
\begin{center}
$
\begin{array}{@{}c@{\hspace{1cm}}c@{\hspace{1cm}}c@{}}
  \begin{array}{@{}l@{}}
    L_{\Rp}' = \lFix X. \Rr \lRecv \tau \hasCost \ccc. \Rq \lSend \tau. X
  \end{array}
  &
  \begin{array}{@{}l@{}}
    L_{\Rq}' = \lFix X. \Rp \lRecv \tau \hasCost \ccc. \Rr \lSend \tau. X
  \end{array}
  &
  \begin{array}{@{}l@{}}
    L_{\Rr}' = \lFix X. \Rp \lRecv \tau \hasCost \ccc. \Rp \lSend \tau. X
  \end{array}
\end{array}
$
\end{center}
\vskip.2cm

\noindent
This is a clear deadlock situation, since all participants are waiting for a
message from each other. To avoid such situations, we define the
\emph{Asynchronous Message Optimisation} for global types, and show its
soundness:


\begin{definition}[Asynchronous Message Optimisation]
\label{def:async}
We first extend the syntax of global types to include
send ($\gSend$) and receive ($\gRecv$) actions as:
$G \Coloneqq \Rp \Rq \gSend \gTy{\tau}.G \mid \Rp \Rq \gRecv
\gTy{\tau}.G \mid \ldots$
The \emph{asynchronous optimisation} relation, $G_1 \preceq G_2$
(read: $G_1$ is more optimal than $G_2$), with
$\Rp_1 \neq \Rp_2$ or $\Rq_1 \neq \Rq_2$
is the transitive closure of the rules below:
{\small
\[
    \begin{array}{@{} l @{\quad} r @{\;} c @{\;} l @{}}
\rulename{Init}      &
        \Rp\Rq \gSend \gTy{\tau}. \Rq\Rp  \gRecv \gTy{\tau} . G
      & \gSubtype &
        \Rp \gMsg \Rq \gTy{\tau} . G
    \\[1mm]
\rulename{Out}    & \Rp_1\Rq_1  \gSend \gTy{\tau_1}. \Rp_2\Rq_2  \gSend \gTy{\tau_2}. G
      & \gSubtype &
        \Rp_2\Rq_2  \gSend \gTy{\tau_2}. \Rp_1\Rq_1  \gSend \gTy{\tau_1}. G
    \\[1mm]
\rulename{In}  & \Rp_1\Rq_1  \gRecv \gTy{\tau_1}. \Rp_2\Rq_2  \gRecv \gTy{\tau_2}. G
      & \gSubtype &
        \Rp_2\Rq_2  \gRecv \gTy{\tau_2}. \Rp_1\Rq_1  \gRecv \gTy{\tau_1}. G
        \\[1mm]
\rulename{Opt}
      & \Rp_1\Rq_1  \gSend \gTy{\tau_1}. \Rq_2\Rp_2  \gRecv \gTy{\tau_2}. G
      & \gSubtype &
        \Rq_2\Rp_2  \gRecv \gTy{\tau_2} . \Rp_1\Rq_1  \gSend \gTy{\tau_1}. G
        \\[1mm]
\rulename{OBra} & \Rp_2\Rq_2  \gSend \gTy{\tau_2}. \Rp_1 \gMsg \Rq_1 \{l_i . G_i \}_{i\in I}
      & \gSubtype
      & \Rp_1 \gMsg \Rq_1 \{l_i . \Rp_2\Rq_2  \gSend \gTy{\tau_2}. G_i \}_{i\in I}
    \\[1mm]
\rulename{IBra}  & \Rp_1 \gMsg \Rq_1 \{l_i . \Rq_2\Rp_2  \gRecv \gTy{\tau_2}. G_i \}_{i\in I}
      & \gSubtype
      & \Rq_2\Rp_2  \gRecv \gTy{\tau_2}. \Rp_1 \gMsg \Rq_1 \{l_i . G_i \}_{i\in I}
    \\[1mm]
\rulename{Cong}      & G \gSubtype G' \Rightarrow \quad\quad\ E[G] &
\gSubtype & E[G']
\end{array}
\]
}
where $E::= \ [ \ ] \mid \Rp \gMsg \Rq \gTy{\tau} . E
\mid \Rp \gMsg \Rq (\{l.  E \}\cup \{l_k.  G_k \}_{k \in K})
\mid \gFix X. E$.
\end{definition}
The optimisation starts first splitting, by \rulename{Init}, the message to a
sending and receiving operation; \rulename{Out} permute two outputs to two
different participants; \rulename{In} is its dual; \rulename{OBra} permutes a
send and a branch; \rulename{IBra} is dual; and \rulename{Cong} is a
congruence rule. The key rules
are \rulename{Opt},
\rulename{OBra}
and \rulename{IBra},
that perform permutations that allow communication and computation to overlap. This is because
the rules permute send actions to the left, and receive actions to the right.
We prove that whenever $G_2$ is deadlock-free, then $G_1
\gSubtype G_2$ must also be deadlock free. Moreover, we show that $G_1
\gSubtype G_2$ is decidable.
Notice that: (1) our definition is different from the literature
asynchronous subtyping for session types, motivated from more practical
use cases; (2) our cost models can be applied even whenever we do not have
that $G_1 \gSubtype G_2$, in which case, safety can be guaranteed by using any
method from the literature. See \S\ref{sec:relw}.

\begin{restatable}[Asynchronous Message Optimisation]{theorem}{asyncOpt}
\label{thm:opt-sound}
\phantom{linebreak}

\begin{enumerate}
\item {\em (Soundness)} Suppose
$G_2$ is a deadlock-free global type and $G_1 \gSubtype G_2$. Then
$G_1$ is deadlock-free.
\item {\em (Decidability)}
Given $G_1$ and $G_2$, it is decidable whether $G_1 \gSubtype G_2$ or
not.
\end{enumerate}
\end{restatable}

\begin{proof}

  (1) By induction on the derivation of $G_1 \gSubtype G_2$.
Assume $G_1 \gSubtype G_2$ and $G_2$ is deadlock-free and
$C_i= [\role{p} \mapsto G_i \project \role{p}]_{\role{p} \in \roles}$
with $i=1,2$. We prove
if $\langle C_2, Q \rangle$ is deadlock-free, then
$\langle C_1, Q \rangle$ is deadlock-free.
To do this proof, we extend the projection for global types as follows. $\Rp
\Rq \gSend \gTy{\tau}.G \project \Rr = \Rq \lSend \gTy{\tau}.(G \project
\Rr)$ if $\Rp = \Rr$, and $G \project \Rr$ otherwise. $\Rp \Rq \gRecv
\gTy{\tau}.G \project \Rr = \Rq \lRecv \gTy{\tau}.(G \project \Rr)$ if $\Rp =
\Rr$, and $G \project \Rr$ otherwise.
All cases except \rulename{Opt} is obvious. The \rulename{Opt}
states:
$\Rp_1\Rq_1 \gSend \gTy{\tau_1}. \Rq_2\Rp_2 \gRecv \gTy{\tau_2}. G \gSubtype
\Rq_2\Rp_2 \gRecv \gTy{\tau_2} . \Rp_1\Rq_1 \gSend \gTy{\tau_1}. G$.
We know that $\langle C_2, Q \rangle$ is deadlock free.
Note that $\Rp_i \neq \Rq_i$, otherwise $G_2$ cannot be proven
deadlock free (it is either ill-formed, or the optimisation of an ill-formed
global type).
There are two cases, considering the side conditions for the rules:
\begin{enumerate*}[label=\alph*)]
\item if $\Rp_1 \neq \Rq_2$, straightforward
  since these subject of both interactions are different;
\item if $\Rp_1 = \Rq_2 = \Rp$, then we have $C_1\setminus \Rp = C_2 \setminus \Rp$,
  $C_2(\Rp) =\Rp_2 \gRecv \gTy{\tau_2} . \Rq_1 \gSend \gTy{\tau_1}. L$, and
  $C_1(\Rp) =\Rq_1 \gSend \gTy{\tau_1}. \Rp_2 \gRecv \gTy{\tau_2} . L$. Since
  $\langle C_2, Q \rangle$ is deadlock free, then $Q(\Rp_2\Rp) = w \cdot
  \tau_2$.
  Therefore,
  $\langle C_1, Q \rangle \xrightarrow{\ast} \langle C_1', Q' \rangle$,
  $\langle C_2, Q \rangle \xrightarrow{\ast} \langle C_2', Q' \rangle$,
  with $C_1'(\Rp) = C_2'(\Rp) = L$, and $Q'(\Rp_2\Rp) = w$. Since
  $\langle C_2, Q \rangle$ is deadlock free, then $\langle C_2', Q' \rangle$ must
  also be deadlock free, and $\langle C_1', Q' \rangle$ as well.
\end{enumerate*}

(2) We consider a normal form which is derived applying \rulename{Out,In}
with the side condition $\Rp_1 < \Rp_2$; and all other rules except
\rulename{Init} as much as possible until no rule is applicable, and finally
applying \rulename{Init} to all pairs of send/receive.
Then if $G_1 \gSubtype G_2$, there exists a unique global type $G$ such that
$G_i \gSubtype G$ derivable applying the above rules finitely. This means
interpreting $G_1 \gSubtype G_2$ as a term rewriting system. The term
rewriting system is terminating because
\begin{enumerate*}[label=\alph*)]
  \item the terms are finite, since we do not unroll recursion; and
  \item the only potential rewrite cycle appears in rules \rulename{Out,In},
  which is prevented by the additional side condition that $\Rp_1 < \Rp_2$
\end{enumerate*}
The repeated application of these rules permute the send and receive actions
to their rightmost and leftmost position respectively. By the side conditions
of the rules, it is straightforward to show that any critical pairs can be
unified, since no rule can prevent another rule from being applied.
\end{proof}

Finally, we prove that, ignoring sending costs, if $G_1 \gSubtype G_2$, then
the cost of $G_1$ is less than the cost of $G_2$. The reason why we need to
ignore sending costs for this proof is that permuting two output actions may
introduce delays in a later computation stage. Note that this property is a
statement about the \emph{synchronisation} costs, not an algorithm for
optimising a protocol. To illustrate this case, consider the following global
types:

\vskip.2cm
{\small
\begin{center}
$
\begin{array}{@{}l@{}}
  G_1 = {\Rp \Rq_1 \gSend \tau .
        \Rp \Rq_2 \gSend \tau .}
        \Rq_1 \Rp \gRecv \tau .
        \Rq_2 \Rp \gRecv \tau .
        \Rq_1 \tRun \ccc_1 .
        \Rq_2 \tRun \ccc_2 .
        \Rq_2 \Rr \gSend \tau' .
        \Rr\Rq_2  \gRecv \tau' .
        \gEnd
\\
  G_2 = {\Rp \Rq_2 \gSend \tau .
        \Rp \Rq_1 \gSend \tau .}
        \Rq_1 \Rp \gRecv \tau .
        \Rq_2 \Rp \gRecv \tau .
        \Rq_1 \tRun \ccc_1 .
        \Rq_2 \tRun \ccc_2 .
        \Rq_2 \Rr \gSend \tau' .
        \Rr\Rq_2  \gRecv \tau' .
        \gEnd
\end{array}
$
\end{center}
}
\vskip.2cm

\noindent
It is clear that $G_1 \gSubtype G_2$, by \rulename{Out}. However, whenever
$\ccc_2 \geq \ccc_1$, then $\cost(G_1) \geq \cost(G_2)$, since $\Rq_2$ must wait
longer in $G_1$ than in $G_2$ before receiving the message of type $\tau$.
Note that, even if $\ccc_1 = \ccc_2$, the cost of the global protocol will be
greater in $G_1$, since $\Rr$ is the participant that takes longer in the
protocol, and needs to wait for $\Rq_2$.
The implications of this result are twofold:
\begin{enumerate*}[label=(\alph*)]
  \item we know that whenever $G_1 \gSubtype G_2$, $G_1$ contains \emph{less}
  overhead due to synchronisation; and
  \item for a given $G_2$, choosing an optimal $G_1 \gSubtype G_2$ is not
  straightforward, and depends on actual local computation costs and
  communication latencies.
\end{enumerate*}

\begin{restatable}[Optimisation Cost]{theorem}{optCost}
\label{thm:opt}
Suppose $G_2$ is a well-formed global type and $G_1 \gSubtype G_2$. If
the sending cost is $0$, then $\cost(G_1) \leq \cost(G_2)$.
\end{restatable}

\begin{proof}
  By induction on the derivation of $G_1 \gSubtype G_2$. Most cases are
  permutations of independent interactions, and all independent interactions
  can be permuted with no effect on the cost. Since we assume zero send
  costs, the cost of sending two actions is the same, independently of the
  order. The reasoning is similar for receiving interactions. The only rules
  that we need to consider are \rulename{Opt}, \rulename{OBra} and
  \rulename{IBra}. Notice that in all the cases, the left hand side contains
  a sending (or choice) at an earlier position than the right hand side. We
  show the proof for case \rulename{Opt}, but all cases follow a similar
  structure. The cost of
   $\Rp_1\Rq_1 \gSend \{\tau_1 \} . \Rp_2\Rq_2 \gRecv \{\tau_2 \} . G$
    is the cost of $G$, where the message queue for
  $\Rp_1\Rq_1$ contains the current execution time for $\Rp_1$. If $\Rp_1
  \neq \Rp_2$, then the cost will be the same in both cases. But if $\Rp_1 =
  \Rp_2$, then the cost in the right hand side will contain the accumulated
  cost for $\Rp_1$, plus the cost of receiving from $\Rq_2$. Since the costs
  recorded at the message queue are greater, then the cost of the
  continuation must also be greater.
\end{proof}

We illustrate how this optimisation reduces synchronisation time with one
iteration of a ring protocol of size 2: {\small
$
\Rp \gMsg \Rq \gTy{\tau_1 \hasCost \ccc_1} . \Rq \gMsg \Rp \gTy {\tau_2 \hasCost \ccc_2} . \gEnd.
$
}
The only possible trace for running such protocol is: { \small
$
\Rp\Rq \tSend \tau_1 \cdot
\Rp\Rq \tRecv \tau_1 \cdot
\Rq \tRun \ccc_1 \cdot
\Rq\Rp \tSend \tau_2 \cdot
\Rq\Rp \tRecv \tau_2 \cdot
\Rp \tRun \ccc_2.
$
}
This trace and the derived cost imply that computation costs $\ccc_1$
and $\ccc_2$ cannot happen in parallel: { \small
$
\Time_\Rp = \csend(\tau_1) + \crecv(\tau_1) + \ccc_1 + \csend(\tau_2) + \crecv(\tau_2) + \ccc_2.
$
}

\noindent
In cases where such interactions are independent, we can permute the
send/receive actions of $\Rq$ to remove the synchronisation cost from $\Rp$, and
allow \emph{any trace that is an interleaving of the following sub-traces}, where
the send operations happen before the matching receive:

\vskip.1cm
\begin{center}
$
\begin{array}{@{}l@{}}
  tr_\Rp = \Rp\Rq \tSend \tau_1 \cdot \Rq\Rp \tRecv \tau_2 \cdot \Rp \tRun \ccc_2
  \hspace{1cm}
  tr_\Rq = \Rq\Rp \tSend \tau_2 \cdot \Rp\Rq \tRecv \tau_1 \cdot \Rq \tRun \ccc_1
\end{array}
$
\end{center}
\vskip.1cm

\noindent
Such optimisations is represented by the following type: {\small
$
\Rp \Rq \gSend \gTy{\tau_1} . \Rq \Rp \gSend \gTy {\tau_2} .
\Rq \Rp \gRecv \gTy{\tau_1 \hasCost \ccc_1} . \Rp \Rq \gRecv \gTy {\tau_2 \hasCost \ccc_2} .
\gEnd.
$
}
\vskip.1cm

\noindent
This scenario will have the cost that we show below, which is smaller than
the original cost.

\vskip.1cm
\begin{center}
$
\begin{array}{@{}l@{}}
  \Time_\Rp = \max(\csend(\tau_1), \csend(\tau_2)) + \crecv(\tau_2) + \ccc_2
  \quad
  \Time_\Rq = \max(\csend(\tau_1), \csend(\tau_2)) + \crecv(\tau_1) + \ccc_1
\end{array}
$
\end{center}
\vskip.1cm

\section{Implementation}
\label{sec:impl}
We implemented a library in Haskell for describing global types augmented
with size and cost information, from which we can derive cost equations for
protocols.

\subsection{Resource Contention}
\label{subsec:contention}
\CAMP{} addresses the issue that multiple participants may need to
share computational resources. We model the cases in which the
participants of a protocol are mapped to distinct \emph{nodes} of a
distributed system, where each node may contain multiple \emph{cores}. This
requires:
\begin{enumerate*}[label=\alph*)]
  \item a \emph{target hardware} specification, and
  \item a \emph{mapping} from participants to nodes.
\end{enumerate*}
The target hardware specification describes the amount of nodes available,
the cores per-node, and the communication latencies between nodes. The
mapping from participants to nodes assigns each participant of the
distributed system to a different node. Our assumptions are:
\begin{enumerate*}[label=\alph*)]
  \item there is no mechanism for process migration;
  \item processes can be pinned to specific nodes, but not to specific cores; and
  \item an optimistic scheduling scenario, in which participants will run as
  soon as possible, whenever a core becomes available.
\end{enumerate*}

\begin{definition}[Target Hardware Specification]
  The target hardware is specified as an indexed set of node descriptions,
  and the communication latencies between nodes:
  $\{\core_\node\}_{\node \in \nodes}$
  and
  $\{\latency_{\node_1\node_2}\}_{\node_1, \node_2 \in \nodes}$.
  Here, $\nodes$ is the set of \emph{node identifiers}, $\core_\node$ is a
  natural number that describes the number of available cores for node
  $\node$, and $\latency_{\node_1\node_2}$ is a function from a size to the
  amount of time it takes to transmit a value from $\node_1$ to $\node_2$.
\end{definition}

\begin{definition}[Participant Mapping]
  The participant mapping associates each participant with a specific node.
  We say that participants are \emph{pinned} to nodes
  $\mapping : \roles \to \nodes$.
\end{definition}

For example, consider the master-worker example, where we have $1$ master and
$5$ workers:
\[
\gFix X. \; \role{m} \gMsg \role{w_1}\gTy{\tau_1}. \ldots . \role{m} \gMsg
\role{w_5}\gTy{\tau_1}. \role{w}_1 \gMsg \role{m}\gTy{\tau_2}. \ldots .
\role{w_5} \gMsg \role{m}\gTy{\tau_2}. X
\]
\noindent
First, we need to know which is the target hardware. In our case, this is a
distributed system with two nodes, $\node_1$ and $\node_2$, with $1$ and $4$
cores respectively. That is: with $\core_{\node_1} = 1$ and $\core_{\node_2}
= 4$. Suppose that the communication latency between $\node_1$ and $\node_2$
is a known function on the size of the messages, $l$. Then,
$\latency_{\node_1\node_2} = \latency_{\node_2\node_1} = l$. Our hardware
description is completed by $\{\core_\node\}_{\node \in \{\node_1,
\node_2\}}$, and $\{\latency_{\node\node'}\}_{\node, \node' \in \{\node_1,
\node_2\}}$.
Finally, we require to map our participants to the different nodes in the
architecture. In our example, we may want to run $\role{m}$ in $\node_1$, and
$\role{w}_i$ in $\node_2$: $M(\role{m}) = \node_1$ and $M(\role{w}_i) =
\node_2$.

To compute the cost in this specific scenario, we use the \emph{resource
bounded} cost equations. The key difference is that, as well as keeping track
of the accumulated time per-role, we keep the accumulated time per node,
using a \emph{core-availability} time, which is the earliest time at which a
core becomes available. The resource-bounded cost equations are obtained
using $\cost(G)(\Time, \Sys, \Queue)$, where $\Sys$ accumulates the cost at
each core and each node of the system. We assume a hardware specification and
mapping. The rules are now modified in the following way:
\[
\cost(\role{p}\gMsg \role{q} \gTy{\tau}. G, \vec{k}) (\Time, \Sys, \Queue)
= \cost(G, \vec{k}) ( \Time[ \Rp \mapsto \Sys_1(\mapping(\Rp)_\ccc)
                          , \Rq \mapsto \Sys_2(\mapping(\Rq)_\ccc)
                          ]
                    , \Sys_2
                    , \Queue
\]
where
    $\Sys_1 = \Sys[\mapping(\Rp)_{\ccc_1} \addcost \csend(\tau)]$,
    $\Sys_2 = \Sys_1[\mapping(\Rq)_{\ccc_2} \addcost \crecv(\tau) + \latency_{M(\Rp)M(\Rq)}]$,
  {\small
    $\forall \ccc, \Sys[\mapping(\Rp)_{\ccc_1}] \leq \Sys[\mapping(\Rp)_\ccc]$
  }
and
    $\forall \ccc, \Sys_1[\mapping(\Rq)_{\ccc_2}] \leq \Sys[\mapping(\Rq)_\ccc]$.
In this definition, we update the accumulated cost of $\Rp$ and $\Rq$ to the
total accumulated cost of the lowest cost core of the node to which they are
mapped. The definition of $\Sys[\node_\ccc \addcost \ccc]$ is the same as in
\S\ref{sec:cost}.

\subsection{A Monadic Interface for Global Types}
\label{sec:eval-monadic}

We develop a deep embedding of the global types of \S\ref{sec:mpst} in Haskell,
and provide a monadic interface on top as a simpler interface for
representing protocols. We call this monadic interface \lstinline|GTM|, for
\emph{Global Type Monad}. In \lstinline|GTM|, there is an implicit $\gEnd$ at the end
of each sequence of interactions. An interaction is specified using function
\lstinline|message|, and participants are created using \lstinline|mkRole|. Function
\lstinline|gclose| runs the code in the \lstinline|GTM| monad, and produces the resulting
global type (\texttt{CGT}). We show below the Haskell code that generates an
$n$-stage pipeline, and a recursive $2$-stage pipeline generated using the
following code:

\vskip.2cm
\noindent
{\small
$
\begin{array}{l}
  \mathtt{pipe} :: [(\mathtt{SType}, \mathtt{Cost})] \to
    \mathtt{Role} \to \mathtt{GTM}\;() \to \mathtt{GTM}\;()
  \\
  \begin{array}{@{}l@{\;}l@{\;}l@{\;}l}
    \mathtt{pipe} \; [] & p & k & = k
    \\
    \mathtt{pipe} \; ((t, c) : r) & p & k & =
    \mathtt{mkRole} >>= \lambda q \to  \mathtt{message} \; p \; q \; t \; (\mathtt{cost} \; t)
     >> \mathtt{pipe} \; r \; q \; k
  \end{array}

  \\[.4cm]
  \mathtt{rpipe2} :: \mathtt{CGT}
  \\
  \mathtt{rpipe2} = \mathtt{gclose} \; \$ \; \mathtt{mkRole}
    >>= \lambda r \to \mathtt{grec} \; \$ \; \lambda x \to \mathtt{pipe}
      \; [(t_1, c_1), (t_2, c_2)] \; r \; x
\end{array}
$
}
\vskip.2cm

\noindent
The code for \lstinline|rpipe2| produces the following global type:
{
$\gFix X. \role{p} \gMsg
\role{q} \gTy{\text{\tt t1}\hasCost \text{\tt c1}}. \role{q} \gMsg \role{r}
\gTy{\text{\tt t2}\hasCost \text{\tt c2}}. X$.
}
Notice that embedding a global type language in Haskell allows us to compute
topologies based on any input parameters, such as the number of stages of a
pipeline, that would otherwise require the use of extensions to \MPST, e.g.\
parameterised roles \cite{CHJNY:2019, DYBH12}. However, to check
well-formedness, we need to instantiate the parameters.

We provide functions \lstinline|cost| and \lstinline|latency|, both of
type $\mathtt{CGT} \to \mathtt{Time}$, to compute the set of equations that
describe the cost (latency) of an input global type. To obtain a
particular prediction, the user needs to provide an instantiation of all free
size and cost variables in the equations, including the transmission costs
between participants.

\section{Evaluation}
\label{sec:eval}
This section presents a number of benchmarks used to evaluate the predictive
power of \CAMP. Our benchmarks are taken from multiple different sources,
mostly \MPST-based tools \cite{NCY2015, CY:2020, CHJNY:2019,INY19,OOPSLA20FStar},
but also a subset of the Savina actor benchmarking suite \cite{IS14}. We
categorise our benchmarks following the structure of the Savina benchmarking
suite:
\begin{enumerate*}[label=(\roman*)]
      \item microbenchmarks,
      \item concurrency benchmarks, and
      \item parallel algorithms.
\end{enumerate*}
Microbenchmarks focus on different structures and protocols, and are aimed at
testing and evaluating the different features of \CAMP. Concurrency
benchmarks are aimed at evaluating the impact on communication and
synchronisation. This can be useful to, e.g.\ estimate server response times,
and set the appropriate timeouts in larger systems. In the context of
parallel algorithms, the main use of the cost models is to predict the
parallel speedups achieved by a particular parallelisation, without needing
to run or profile the application.

\subsection{Methodology}
We follow a series of steps in order to make our results as consistent as
possible. We will detail now these steps, highlighting which part is
automated, and which needs to be provided by the developer. Our methodology
is divided in two parts:
\begin{enumerate*}
\item characterising the target architecture; and
\item benchmark cost analysis.
\end{enumerate*}

\myparagraph{Characterising the target architecture.}
To tailor a cost analysis to a specific target architecture, we need to
characterise the costs of sending/receiving data between nodes. This requires
three steps:
\begin{enumerate*}
      \item specifying the amount of nodes, and the amount of processors/cores
      per node;
      \item estimate message latencies between nodes; and,
      \item profiling send/receive operations in the required
      languages/frameworks with inputs of different sizes.
\end{enumerate*}

We require the results of these steps to be stored in a \verb|.hs| file, as
an architecture description, that will be imported and used by \CAMP's cost
models. These steps must be performed only once per architecture and
programming language.

Additionally to our theory, the implementation allows programmers to specify
an overhead for running multiple participants in a single node. This is to
account for all factors that \CAMP{} is currently not considering for
deriving cost equations. See \S \ref{sec:concl} for a discussion.

\myparagraph{Benchmark cost analysis.} This is the main part of the cost
analysis. This part \emph{does not require that the target application is
implemented using an \MPST-based framework}. Assume that we start with a
target application, already implemented. The steps of our methodology are the
following:

\begin{enumerate}
\item \textbf{Write its global type.} Since most of our benchmarks are
derived from implementations in other \MPST-based tools, this step is
straightforward. For non \MPST-based implementations, the developer needs
to analyse the communication protocol and write it as a global type.

\item \textbf{Extract the sequential parts.} The sequential parts
must be extracted as self-contained implementations, that can
be run independently of the whole distributed system.

\item \textbf{Run the profiler on the sequential parts.} Our profiler
requires multiple input sizes, measures the execution costs of the sequential
parts on these input sizes, and performs cubic spline interpolation on the
gathered data. Note that the sequential cost is only valid for inputs of
sizes that are within the measured range. This part can be omitted when using
a static cost analysis, or the cost equations are known and provided
manually.

\item \textbf{Annotate the global type} and extract cost equations.

\item \textbf{Instantiate the cost} by feeding the profiling information for
both the target architecture and the sequential parts.
\end{enumerate}

\subsection{Benchmark Structure}
\label{sec:benchstructure}
We list and provide a brief explanation of all the benchmarks that we used
for the cost models. We used two different hardware configurations for the
evaluation. We name them \desktop{} and \hpc{}: \desktop{} is a 4-core
Intel(R) Core(TM) i7-6700 CPU @ 3.40GHz with hyperthreading, and \hpc{}
comprises 2 NUMA nodes, 12 cores per node and 62GB of memory, using Intel
Xeon CPU E5-2650 v4 @ 2.20GHz chips. \hpc{} is an HPC cluster that uses PBS
queuing mechanism. We made sure that we consistently selected the same
hardware for every execution. In the remainder of this section, we will
specify whether the benchmarks were run on \desktop{} or \hpc{}.

We used the benchmarks as defined in the different sources from where we took
the source code.  Overall, we used averages of > 50 repetitions for benchmarks
with large computation costs, and linear regression ($95\%$ CI) for smaller
(micro-benchmarks such as ping-pong, all-to-all, etc).

\myparagraph{Microbenchmarks.}
\textbf{\emph{Recursive Ping-Pong}} is the recursive ping-pong example.
We run both the Scala benchmark (\bench{pp-akka}) from the Savina
benchmarking suite \cite{IS14}, and the OCaml version taken from
\citet{INY19} on \desktop, on three different transports
(\bench{pp-ev}, \bench{pp-lwt}, \bench{pp-ipc-}$n$). Since
the cost of sending in the \bench{ipc} transport depends on the input
size, we use $n$ to differentiate different runs of this benchmark with
different input sizes.
We introduced an arbitrary computation to the Scala version to increase the
local computation costs. {\bf\emph{Thread Ring}} (\bench{ring}) is the
Scala version from \citet{IS14}, both with and without asynchronous message
optimisations, on
\desktop{}. {\bf\emph{Counting Actor}} (\bench{count}) is a
benchmark with two actors, one of which counts the number of messages
received from the other. This is the Savina microbenchmark \citet{INY19} on
\desktop.
{\bf\emph{One-to-All}},
{\bf\emph{All-to-One}} and {\bf\emph{All-to-All}}: we use the Go one-to-all, all-to-one and all-to-all Go
implementations
(\bench{1a}, \bench{a1} and \bench{aa})
in \cite{CHJNY:2019}, all run on \desktop.

\myparagraph{Concurrency Benchmarks.}
{\bf\emph{Two-Buyer Protocol}}  (\bench{twobuy}). We use an F$\star$ implementation taken from
\cite{OOPSLA20FStar}, and extracted into OCaml. {\bf\emph{Sleeping Barber}}x,
  {\bf\emph{Dining
Philosophers}} and {\bf\emph{Cigarette Smoker}} (\bench{barb},
\bench{dphil},
\bench{csmok}). These are the Savina Akka benchmarks in
\cite{IS14}, run on \desktop. {\bf\emph{K-Nucleotide}}, {\bf\emph{Spectral-Norm}} and
{\bf\emph{Regex-DNA}}
(\bench{kn}, \bench{sn}, \bench{dna}).
 These benchmarks are Go implementations taken from
\cite{CHJNY:2019}, based on the Computer Language Benchmarks Game, and use
different combinations of scatter, gather, choices and recursion.


\myparagraph{Parallel Algorithms.}
All these benchmarks were run on \hpc, and they were taken from two sources:
\citet{NCY2015}(NBody, linear equation solver, wordcount and adpredictor) and
\citet{CY:2020}(dot product, fast fourier transform and mergesort).

\myparagraph{\citet{NCY2015}.}
The work by \citet{NCY2015} has implemented representative
parallel benchmarks from \cite{dwarfs-cacm}.
{\bf\emph{NBody}}
(\bench{nb}) is a 2D NBody simulation in C+MPI which
is implemented as a thread ring with asynchronous communication
optimisations. {\bf\emph{Linear equation solver}}
(\bench{ls}) is parallelised using a
wraparound mesh. Similarly to the NBody example, we required the extension
with asynchronous communication optimisation. {\bf\emph{WordCount}}
(\bench{wc}) and
{\bf\emph{AdPredictor}}
(\bench{ap}) are parallelised using map-reduce.

\begin{figure}
  \footnotesize
      \begin{tikzpicture}[%
            ,source/.style={%
                  ,rectangle
                  ,draw
                  ,minimum height=1.4cm
                  ,minimum width=.15cm
                  ,xshift=-.7cm
            }
            ,sink/.style={%
                  ,rectangle
                  ,draw
                  ,minimum height=1.4cm
                  ,minimum width=.15cm
                  ,xshift=.7cm
            }
            ,buff1/.style={%
                  ,draw
                  ,rectangle
                  ,rounded corners
                  ,yshift=.5cm
                  ,minimum width=.3cm
                  ,minimum height=.5cm
            }
            ,buff2/.style={%
                  ,draw
                  ,rectangle
                  ,rounded corners
                  ,yshift=-.5cm
                  ,minimum width=.3cm
                  ,minimum height=.5cm
            }
            ,lblup/.style={%
                  ,above
                  ,midway
                  ,yshift=.55cm
            }
            ,lblbelow/.style={%
                  ,below
                  ,midway
                  ,yshift=-.65cm
            }
      ]

      \node (orig0) at (-.5,0) {};
      \node[buff1]  (ba0)  at (orig0) {\color{blue}\textbf{A}} ;
      \node[buff2]  (bb0)  at (orig0)  {\color{blue}\textbf{B}} ;
      \node[source] (src0) at (orig0) {};
      \node[sink]   (snk0) at (orig0) {};
      \draw ($(ba0.north west)+(-.15cm,.2cm)$)
            rectangle
            ($(bb0.south east)+(.15cm,-.2cm)$) ;
      \node[left=0.7cm of orig0] { (a) } ;

      \quad\quad      \quad\quad
      \draw[->,thick, blue]
            (ba0) -- ($(src0.center)+(0,.2)$)
            node[lblup]{\color{blue} \scriptsize ready } ;

      \node (orig1) at (2.375,0) {};
      \node[buff1]   (ba1) at (orig1) {\scriptsize \color{blue}\textbf{A}} ;
      \node[buff2]   (bb1) at (orig1)  {\scriptsize \color{blue}\textbf{B}} ;
      \node[source] (src1) at (orig1) {};
      \node[sink]   (snk1) at (orig1) {};
      \draw ($(ba1.north west)+(-.15cm,.2cm)$)
            rectangle
            ($(bb1.south east)+(.15cm,-.2cm)$) ;
      \node[left=0.7cm of orig1] { (b) } ;

      \quad\quad
      \draw[->,thick, purple]
            ($(src1.center)+(0,.2)$)
            --
            (ba1)
            node[lblup]{\color{purple} \scriptsize copy } ;
      \draw[->,thick, blue]
            (bb1)
            --
            ($(src1.center)+(0,-.2)$)
            node[lblbelow]{\color{blue} \scriptsize ready } ;
      \draw[->,thick, blue]
            ($(snk1.center)+(0,.2)$)
            --
            (ba1)
            node[lblup]{\color{blue} \scriptsize ready } ;

      \node (orig2) at (5.25,0) {};
      \node[buff1]   (ba2) at (orig2) {\color{blue}\textbf{A}} ;
      \node[buff2]   (bb2) at (orig2)  {\color{blue}\textbf{B}} ;
      \node[source] (src2) at (orig2) {};
      \node[sink]   (snk2) at (orig2) {};
      \draw ($(ba2.north west)+(-.15cm,.2cm)$)
            rectangle
            ($(bb2.south east)+(.15cm,-.2cm)$) ;
      \node[left=0.7cm of orig2] { (c) } ;

      \quad\quad
      \draw[->,thick,purple]
            ($(src2.center)+(0,-.2)$)
            --
            (bb2)
            node[lblbelow, yshift=-.06cm]{\color{purple} \scriptsize copy } ;
      \draw[->,thick, purple]
            (ba2)
            --
            ($(snk2.center)+(0,.2)$)
            node[lblup]{\color{purple} \scriptsize copy } ;

      \node (orig3) at (8.125,0) {};
      \node[buff1]   (ba3) at (orig3) {\color{blue}\textbf{A}} ;
      \node[buff2]   (bb3) at (orig3)  {\color{blue}\textbf{B}} ;
      \node[source] (src3) at (orig3) {};
      \node[sink]   (snk3) at (orig3) {};
      \draw        ($(ba3.north west)+(-.15cm,.2cm)$) rectangle
                   ($(bb3.south east)+( .15cm,-.2cm)$) ;
      \node[left=0.7cm of orig3] { (d) } ;

      \quad\quad
      \draw[->,thick,purple]
            ($(src3.center)+(0,-.2)$)
            --
            (bb3)
            node[lblbelow, yshift=-.06cm]{\color{purple} \scriptsize copy } ;

      \draw[->,thick,blue]
            (ba3)
            --
            ($(src3.center)+(0,.2)$)
            node[lblup]{\color{blue} \scriptsize ready} ;
      \draw[->,thick, blue]
            ($(snk3.center)+(0,-.2)$)
            --
            (bb3)
            node[lblbelow]{\color{blue} \scriptsize ready} ;

      \node (orig4) at (11,0) {};
      \node[buff1]   (ba4) at (orig4) {\color{blue}\textbf{A}} ;
      \node[buff2]   (bb4) at (orig4)  {\color{blue}\textbf{B}} ;
      \node[source] (src4) at (orig4) {};
      \node[sink]   (snk4) at (orig4) {};
      \draw        ($(ba4.north west)+(-.15cm,.2cm)$) rectangle
                   ($(bb4.south east)+( .15cm,-.2cm)$) ;
      \node[left=0.7cm of orig4] { (e) } ;

      \draw[->,thick,purple]
            ($(src4.center)+(0,.2)$)
            --
            (ba4)
            node[lblup]{\color{purple} \scriptsize copy} ;
      \draw[->,thick, purple]
            (bb4)
            --
            ($(snk4.center)+(0,-.2)$)
            node[lblbelow, yshift=-.06cm]{\color{purple} \scriptsize copy} ;

      \end{tikzpicture}
      \caption{Double-Buffering}
      \label{fig:double-buff}
\end{figure}
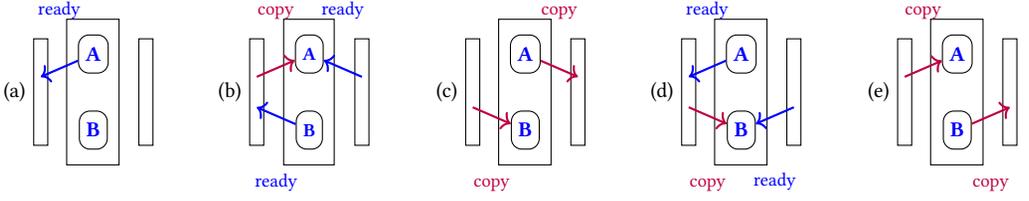

\myparagraph{Double-Buffering Algorithm} (\bench{dbuff})
\cite{doublebuffer} is a well-known technique for increasing the throughput
of a device that has two buffers. To accurately represent a double-buffering
protocol, we use \CAMP's extension with asynchronous message
optimisations. We show the protocol below, using participants $\Rp$ for \emph{source},
$\Rq$ for \emph{sink} and $\Rr$ for the \emph{service}:\\[1mm]
\centerline{
{\small
  $
  \begin{array}{lrl}
\Rr \Rp \gSend \gTy{r_1} .
\Rr \Rp \gSend \gTy{r_2} .
\gFix X. &\hspace{-2mm}
   \Rp \Rr \gRecv \gTy{r_1} . 
   \Rp \gMsg \Rr \gTy{s_1} .
   \Rq \gMsg \Rr \gTy{t_1} .
   \Rr \gMsg \Rq \gTy{u_1} .
   \Rr \Rp \gSend \gTy{r_1} .
   \Rp \Rr \gRecv \gTy{r_2} . 
   \\
   &
   \Rp \gMsg \Rr \gTy{s_2} .
   \Rq \gMsg \Rr \gTy{t_2} .
   \Rr \gMsg \Rq \gTy{u_2} .
   \Rr \Rp \gSend \gTy{r_2} .
   & \hspace{-2mm} X
   \end{array}
$
}}
\vspace{1mm}
Fig. \ref{fig:double-buff} illustrates this protocol. Suppose a
streaming service with two buffers ({\color{blue} \textbf{A}} and
{\color{blue} \textbf{B}}), a source (left) and a sink (right). First
\textbf{(a)}, buffer {\color{blue} \textbf{A}} is ready to copy an element
(message $r_1$),
and so it notifies the source. Then \textbf{(b)}, an element is copied into
{\color{blue} \textbf{A}} (message $s_1$). Meanwhile, both the sink and
{\color{blue} \textbf{B}} can notify the service and the source
(respectively) that they are ready to copy (messages $r_2$ and $t_1$). This
implies that, next \textbf{(c)}, both the service and
the sink can copy an element \emph{in parallel}
(messages $s_1$ and $u_1$). Note that using a single
buffer, this would not be possible, since we would risk overwriting the
buffer before the sink copied it. In the next iteration \textbf{(d)},
we can swap the buffers, and repeat the process. By swapping the buffers,
both the service (buffer {\color{blue} \textbf{A}}) and the sink can notify
that they are ready, even if data is still being copied to buffer
{\color{blue} \textbf{B}} (messages $r_2$ and $t_2$). Finally \textbf{(e)},
buffer {\color{blue} \textbf{A}} and the sink can copy the respective next
elements, again in parallel (messages $s_2$ and $u_2$).

\begin{myremark}[Double-buffering]\rm
  \label{rem:DB}
{\bf (1)}
Definition~\ref{def:async} does not directly
check the asynchronous subtyping of local types from the above global type
as our rules do not include unrolling recursive global types (to obtain
the decidability result). However we can apply any (sound) asynchronous
subtyping relation from the literature since the cost
calculation does not related to well-formedness of
global types. For example,
local types that behave as the projections of this global type are known as
deadlock-free \cite{YoshidaVPH08,mostrous_yoshida_honda_esop09}.
\
{\bf (2)}
The syntax of the global type in \cite{mostrous_yoshida_honda_esop09} uses
the explicit channels in global types. We translated them to
corresponding labels, which does not affect the cost calculation, and
our end-point implementation is essentially as identical as
one in \cite{YoshidaVPH08}.
\end{myremark}

\myparagraph{\citet{CY:2020}.} {\bf\emph{Mergesort}}
(\bench{ms}) follows a divide-and-conquer
protocol. {\bf\emph{Fast Fourier Transform}}
(\bench{fft}) is the Cooley-Tukey
fast-fourier transform algorithm, implemented in C using pthreads,
parallelised using a butterfly topology as illustrated in
Fig.~\ref{fig:butterfly}. It uses a divide-and-conquer strategy based on the
following equation (we use $\omega_N^{2k}=\omega_{N/2}^k$):\\[1mm]
\centerline
{\small
$
\begin{array}{rcl}
  X_k  =  \sum_{j=0}^{N-1}x_j\,\omega_N^{jk}
   = \sum_{j=0}^{N/2-1}x_{2j}\,\omega_{N/2}^{jk}
  + \omega_N^k\sum_{j=0}^{N/2-1}x_{2j+1}\,\omega_{N/2}^{jk}
\end{array}
$
}
\vspace{1mm}
\noindent
Each of the two separate sums are DFT of half of the original vector members,
separated into even and odd. Recursive calls can then divide the input set
further based on the value of the next binary bits.
Fig.~\ref{fig:butterfly}(a) illustrates this
recursive principle, called {\em butterfly}, where two
different intermediary values can be computed in constant time from the results
of the same two recursive calls.
The complete algorithm for a size-8 is illustrated by the diagram from
Fig.\ref{fig:butterfly}(b).
The global type in Fig.\ref{fig:butterfly}(c) shows the
resulting global type, in terms of the indices of the participants that need
to communicate. We use keyword $\mathsf{foreach}$ to
represent that the body must be expanded for all natural numbers
that satisfy the condition (similarly to
parameterised \MPST{} \cite{DYBH12}).
\CAMP{} uses a recursive definition that expands into a butterfly of the
required size.
Fig.\ref{fig:butterfly}(d) shows the (abstract) code of our implementation
for participants $0$ and $1$. We show the high-level structure, in terms of
$\mathsf{send}$ and $\mathsf{receive}$. Suppose that participants receive as
initial value $v$, the deinterleaving of the input vector. Then, they all
start applying a sequential $\mathsf{fft}$, and communicate the result to the
appropriate participants. Then, they apply the necessary addition and
subtraction to compute their part of the result, and communicate it
accordingly.

\begin{figure}
\begin{tabular}{@{}c@{}l@{}}
\begin{minipage}{17.5em}
{\bf \textsf (a) Butterfly pattern}

\hspace{0em}
\xymatrix@C=15pt@R=15pt{
{\footnotesize \text{$x_{k-N/2}$}}
\ar[dr]\ar@{.>}[r]
 & {\hole}\save[]+<1.3cm,-0.2cm>*\txt<4cm>{\small
   $X_{k-N/2}=x_{k-N/2}+$\\ \hspace{6em} $x_{k}*\omega_N^{k-N/2}$}\restore\\
{\footnotesize \text{$x_{k}$}}\ar[ur]
\ar@{.>}[r] 
& {\hole}\save[]+<1.6cm,0cm>*{\text{\small 
   $X_k=x_{k-N/2}+x_{k}*\omega_N^{k}$}}\restore\\
}

\smallskip
{\bf \textsf (b) FFT diagram}\\ 
{\footnotesize
$   
\begin{minipage}{15em}
\centering
\xymatrix@C=18pt@R=15pt{
  \ar[r]^ {x_0}
  & *+[o][F]{0} \ar[dr]\ar@{.>}[r]^1
  & *+[o][F]{0} \ar[ddr]\ar@{.>}[r]^2
  & *+[o][F]{0} \ar[ddddr]\ar@{.>}[r]^3
  & *+[o][F]{0} \ar[r]^{X_0}
  & \\
  \ar[r]^{x_4}
  & *+[o][F]{1} \ar[ur]\ar@{.>}[r]
  & *+[o][F]{1} \ar[ddr]\ar@{.>}[r]
  & *+[o][F]{1} \ar[ddddr]\ar@{.>}[r]
  & *+[o][F]{1} \ar[r]^{X_1}
  & \\
  \ar[r]^{x_2}
  & *+[o][F]{2} \ar[dr]\ar@{.>}[r]
  & *+[o][F]{2} \ar[uur]\ar@{.>}[r]
  & *+[o][F]{2} \ar[ddddr]\ar@{.>}[r]
  & *+[o][F]{2} \ar[r]^{X_2}
  & \\
  \ar[r]^{x_6}
  & *+[o][F]{3}  \ar[ur]\ar@{.>}[r]
  & *+[o][F]{3}  \ar[uur]\ar@{.>}[r]
  & *+[o][F]{3} \ar[ddddr]\ar@{.>}[r]
  & *+[o][F]{3} \ar[r]^{X_3}
  & \\
  \ar[r]^{x_1}
  & *+[o][F]{4}  \ar[dr]\ar@{.>}[r]
  & *+[o][F]{4}  \ar[ddr]\ar@{.>}[r]
  & *+[o][F]{4} \ar[uuuur]\ar@{.>}[r]
  & *+[o][F]{4} \ar[r]^{X_4}
  & \\
  \ar[r]^{x_5}
  & *+[o][F]{5}  \ar[ur]\ar@{.>}[r]
  & *+[o][F]{5} \ar[ddr]\ar@{.>}[r]
  & *+[o][F]{5} \ar[uuuur]\ar@{.>}[r]
  & *+[o][F]{5} \ar[r]^{X_5}
  & \\
  \ar[r]^{x_3}
  & *+[o][F]{6}  \ar[dr]\ar@{.>}[r]
  & *+[o][F]{6}  \ar[uur]\ar@{.>}[r]
  & *+[o][F]{6}  \ar[uuuur]\ar@{.>}[r]
  & *+[o][F]{6} \ar[r]^{X_6}
  & \\
  \ar[r]^{x_7}
  & *+[o][F]{7}  \ar[ur]\ar@{.>}[r]
  & *+[o][F]{7}  \ar[uur]\ar@{.>}[r]
  & *+[o][F]{7} \ar[uuuur]\ar@{.>}[r]
  & *+[o][F]{7} \ar[r]^{X_7}
  & \\
}
\end{minipage} 
$}
\end{minipage}

&

      \begin{minipage}{.49\textwidth}
            \noindent
            \textbf{(c) Global type}
            \[ \begin{array}{@{}l@{}}
                  \Pi n. \mathsf{foreach} (i < 2^n)\lbrace
                  \\ \;
                  \mathsf{foreach}(l < n)\lbrace
                  \\ \;\;
                  \mathsf{foreach}(i < 2^l)\lbrace
                  \\ \;\;\;
                  \mathsf{foreach}(j < 2^{n-l-1})\lbrace
                  \\ \;\;\;\;
                  \mathsf{foreach}(k < 2)\lbrace
                  \\ \;\;\;\;\;
                  \mathsf{foreach}(k' < 2)\lbrace
                  \\ \;\;\;\;\;\;
                  \Rp_{i \times 2^{n - l} + k \times 2 ^ {n - l - 1} + j}
                  \\ \;\;\;\;\;\;\;
                  \gMsg
                  \Rp_{i \times 2^{n - l} + k' \times 2 ^ {n - l - 1} + j}
                  \rbrace
                  \rbrace
                  \rbrace
                  \rbrace
                  \rbrace
                  \rbrace
               \end{array}
            \]

            \noindent
            \textbf{(d) Programs}
            {\small
            \[\begin{array}{@{}l l@{}}
                  \mathsf{P}_0(v) \Coloneqq
                  &
                  \mathsf{P}_1(v) \Coloneqq
                  \\
                  \quad
                  \begin{array}{@{}l@{}}
                        x \leftarrow \mathtt{fft}(v); \\
                        \mathtt{send}\;\mathsf{P_1} \; x; \\
                        y \leftarrow \mathtt{recv} \; \mathsf{P_1} ; \\
                        x \leftarrow \mathtt{zip}_{+}(x,y) ; \\
                        \mathtt{send}\;\mathsf{P_2} \; x; \\
                        y \leftarrow \mathtt{recv}\;\mathsf{P_2}; \\
                        x \leftarrow \mathtt{zip}_{+}(x,y) ; \\
                        \mathtt{send}\;\mathsf{P_4} \; x; \\
                        y \leftarrow \mathtt{recv}\;\mathsf{P_4}; \\
                        x \leftarrow \mathtt{zip}_{+}(x,y) ; \\
                        \mathtt{return} (x)
                  \end{array}
                  &
                  \quad
                  \begin{array}{@{}l@{}}
                        x \leftarrow \mathtt{fft'}(v); \\
                        \mathtt{send}\;\mathsf{P_0} \; x; \\
                        y \leftarrow \mathtt{recv} \; \mathsf{P_0} ; \\
                        x \leftarrow \mathtt{zip}_{-}(x,y) ; \\
                        \mathtt{send}\;\mathsf{P_3} \; x; \\
                        y \leftarrow \mathtt{recv}\;\mathsf{P_3}; \\
                        x \leftarrow \mathtt{zip}_{+}(x,y) ; \\
                        \mathtt{send}\;\mathsf{P_5} \; x; \\
                        y \leftarrow \mathtt{recv}\;\mathsf{P_5}; \\
                        x \leftarrow \mathtt{zip}_{+}(x,y) ; \\
                        \mathtt{return} (x);
                  \end{array}
              \end{array}
            \]
            }
      \end{minipage}
      \end{tabular}
\caption{Butterfly Network Topology for Fast Fourier Transform}
      \label{fig:butterfly}
\end{figure}

\begin{figure}
   {
    \scriptsize
    \begin{tabular}{@{}l@{\qquad\qquad}l@{}}
      \toprule
    \begin{tabular}[t]{@{}l@{\quad} l @{\quad} l@{\quad} l@{\quad} l @{}}
      \textbf{\textsl{Benchmark}}
      &
        \textbf{\textsl{Protocol}}
      &
        \textbf{\textsl{Cost (s)}}
      &
        Real (s)
      &
        Diff (\%)
      \\ \midrule
      \multicolumn{5}{@{}l}{\OCaml{}~\cite{INY19}}
      \\ \midrule
      \benchscript{pp-ev}
      & PP
      & 6.39e-6
      & 6.29e-6
      & $2.07$
      \\
      \benchscript{pp-lwt}
      & PP
      & 4.20e-7
      & 4.07e-7
      & $3.29$
      \\
      \benchscript{pp-ipc-0}
      & PP
      & 6.27e-6
      & 5.95e-6
      & $5.40$
      \\
      \benchscript{pp-ipc-1}
      & PP
      & 6.28e-6
      & 6.12e-6
      & $2.66$
      \\
      \benchscript{pp-ipc-2}
      & PP
      & 6.42e-6
      & 6.19e-6
      & $3.67$
      \\
      \benchscript{pp-ipc-3}
      & PP
      & 7.96e-6
      & 7.80e-6
      & $2.08$
      \\
      \benchscript{pp-ipc-4}
      & PP
      & 2.54e-5
      & 2.09e-5
      & $21.9$
      \\
      \benchscript{pp-ipc-5}
      & PP
      & 2.20e-4
      & 2.19e-4
      & $0.62$
      \\ \midrule
      \multicolumn{5}{@{}l}{Go \cite{CHJNY:2019}}
      \\ \midrule
      \benchscript{aa-2}
      & AA
      & 2.42e-6
      & 2.13e-6
      & $14$
      \\
      \benchscript{aa-4}
      & AA
      & 4.85e-6
      & 4.45e-6
      & $8.8$
      \\
      \benchscript{1a-2}
      & S
      & 6.28e-6
      & 4.46e-6
      & $0.41$
      \\
      \benchscript{1a-3}
      & S
      & 8.12e-6
      & 7.64e-6
      & $6.42$
      \\
      \benchscript{1a-4}
      & S
      & 9.98e-6
      & 9.85e-6
      & $1.34$
      \\
      \benchscript{a1-2}
      & G
      & 2.8e-6
      & 2.14e-6
      & 30.66
      \\
      \benchscript{a1-3}
      & G
      & 3.27e-6
      & 2.86e-6
      & 14.09
      \\
      \benchscript{a1-4}
      & G
      & 3.74e-6
      & 3.30e-6
      & 13.22
      \\
      \benchscript{sn-1}
      & SG, CR
      & 11.62
      & 11.58
      & 0.37
      \\
      \benchscript{sn-2}
      & SG, CR
      & 5.87
      & 5.81
      & 1.05
      \\
      \benchscript{sn-3}
      & SG, CR
      & 3.98
      & 3.95
      & 0.79
      \\
      \benchscript{sn-4}
      & SG, CR
      & 3.06
      & 3.05
      & 0.08
      \\
      \benchscript{kn-1}
      & SG
      & 10.88
      & 10.65
      & 2.16
      \\
      \benchscript{kn-2}
      & SG
      & 11.93
      & 11.13
      & 7.15
      \\
      \benchscript{kn-3}
      & SG
      & 14.01
      & 13.01
      & 7.69
      \\
     \benchscript{kn-4}
      & SG
      & 17.28
      & 17.17
      & 0.66
      \\
      \benchscript{dna-1}
      & SG
      & 3.00
      & 2.93
      & 2.38
      \\
      \benchscript{dna-2}
      & SG
      & 3.34
      & 3.39
      & 1.38
      \\
      \benchscript{dna-3}
      & SG
      & 3.68
      & 3.66
      & 0.48
      \\
      \benchscript{dna-4}
      & SG
      & 4.02
      & 4.01
      & 0.24
      \\ \midrule
      \multicolumn{5}{@{}l}{Savina \cite{IS14}}
      \\ \midrule
      \benchscript{pp-akka}
      & PP
      & 4.4e-5
      & 3.99e-5
      & 10.28
      \\
      \benchscript{ring}
      & Ring
      & 7.09-3
      & 5.04e-3
      & 40.67
      \\
      \benchscript{ring-opt}
      & Ring
      & 5.24e-4
      & 5.4e-4
      & 2.8
      \\
      \benchscript{count}
      & CR
      & 1.98e-4
      & 1.53e-4
      & 29.41
      \\
      \benchscript{barb}
      & CR
      & 3.5e-4
      & 3.36e-4
      & 4.16
      \\
      \benchscript{dphil}
      & CR
      & 2.03e-4
      & 1.92e-4
      & 5.75
      \\
      \benchscript{csmok}
      & CR
      & 1.05e-4
      & 1.03e-4
      & 1.6
      \\
    \end{tabular} &
    \begin{tabular}[t]{@{}l@{\quad} l @{\quad} l@{\quad} l@{\quad} l @{}}
      \textbf{\textsl{Benchmark}}
      &
      \textbf{\textsl{Protocol}}
      &
        \textbf{\textsl{Cost (s)}}
      &
        Real (s)
      &
        Diff (\%)
      \\ \midrule
      \multicolumn{5}{@{}l}{C-MPI \cite{NCY2015}}
      \\
      \midrule
      \benchscript{nb-1}
      & Ring
      & 177.91
      & 177.91
      & 2e-6
      \\
      \benchscript{nb-4}
      & Ring
      & 45.17
      & 44.71
      & 1.02
      \\
      \benchscript{nb-16}
      & Ring
      & 12.07
      & 11.10
      & 8.79
      \\
      \benchscript{nb-32}
      & Ring
      & 6.69
      & 7.84
      & 15
      \\
      \benchscript{nb-64}
      & Ring
      & 4.29
      & 4.28
      & 0.086
      \\
      \benchscript{ls-1}
      & Mesh
      & 10.98
      & 10.58
      & 3.78
      \\
      \benchscript{ls-4}
      & Mesh
      & 4.34
      & 4.44
      & 2.23
      \\
      \benchscript{ls-16}
      & Mesh
      & 1.88
      & 1.72
      & 9.67
      \\
     \benchscript{ls-32}
      & Mesh
      & 1.19
      & 1.30
      & 8.79
      \\
      \benchscript{ls-64}
      & Mesh
      & 0.87
      & 0.72
      & 0.20
      \\
      \benchscript{wc-1}
      & MR
      & 57
      & 57
      & 1e-5
      \\
      \benchscript{wc-2}
      & MR
      & 31.8
      & 27.5
      & 17
      \\
      \benchscript{wc-8}
      & MR
      & 17
      & 16
      & 6.26
      \\
      \benchscript{wc-24}
      & MR
      & 17.5
      & 19.5
      & 10
      \\
      \benchscript{wc-64}
      & MR
      & 20.6
      & 23.0
      & 10
      \\
      \benchscript{ap-1}
      & MR
      & 657
      & 656
      & 7e-2
      \\
      \benchscript{ap-2}
      & MR
      & 330
      & 284
      & 16
      \\
      \benchscript{ap-8}
      & MR
      & 67
      & 65
      & 3.4
      \\
      \benchscript{ap-24}
      & MR
      & 51
      & 45
      & 13
      \\
      \benchscript{ap-64}
      & MR
      & 74
      & 64
      & 17
      \\ \midrule
      \multicolumn{5}{@{}l}{C-pthreads \cite{CY:2020}}
      \\
      \midrule
       \benchscript{fft}
      & Btfly
& 143.1
& 143.0
& 5.8e-2
      \\
      \benchscript{fft-2}
      & Btfly
& 74.3
& 74.1
& 1.7e-1
      \\
      \benchscript{fft-4}
      & Btfly
& 40.5
& 40.8
& 7.2e-1
      \\
      \benchscript{fft-8}
      & Btfly
& 24.3
& 21.8
& 12
      \\
      \benchscript{fft-32}
      & Btfly
& 13.6
& 12.4
& 9.3
      \\
      \benchscript{ms-2}
      & d\&c
& 53.6
& 53.2
& 7.3-1
      \\
      \benchscript{ms-4}
      & d\&c
& 31.39
& 31.33
& 1.3-1
      \\
      \benchscript{ms-8}
      & d\&c
& 20.1
& 18.1
& 11.3
      \\
      \benchscript{ms-16}
      & d\&c
& 14.6
& 14.2
& 2.5
      \\ \midrule
      \multicolumn{5}{@{}l}{OCaml \cite{OOPSLA20FStar}}
      \\
      \midrule
      \benchscript{twobuy}
      & CR
      & 4.0133
      & 4.0035
      & 0.24
      \\
      \\ \midrule
      \multicolumn{5}{@{}l}{C \cite{YoshidaVPH08}}
      \\
      \midrule
      \benchscript{dbuff}
      & Double Buffer
      & 2.54e-1
      & 2.12e-1
      & 19.7
      \\
    \end{tabular}
      \\ \bottomrule
    \end{tabular}
  }
\caption{Predicted vs real execution times: PP = Ping-Pong,
AA = All-to-All, S = Scatter, G = Gather,
    SG = Scatter-Gather, CR = choice with
recursion, MR = MapReduce, D\&C = parallel divide and
conquer.}
\label{fig:evaluation}
\vspace{-0.4cm}
\end{figure}
\subsection{Discussion of Predicted Execution Times}
\label{subsec:result}
Fig. \ref{fig:evaluation} shows a comparison, for each benchmark, of the
execution times compared with the predictions by our cost models. For most of
our examples, we get predictions with $< 15\%$ of error. Examples include
\bench{pp-ipc-4}, \bench{a1-2}, \bench{ring},
\bench{count},
\bench{nb-32}, \bench{wc-2}, \bench{ap-2},
\bench{ap-64}, and \bench{dbuff}.  We observe that the worst
predictions are those of the microbenchmarks, with very small execution times.
Here, communication costs dominate, and are repeated a large number of times.
With such small costs, a small error is amplified after a large enough number of
iterations. An example of this is \bench{ring}, that is a recursive
ring protocol that is run for $10^5$ iterations.

When we consider examples with larger local computation costs,
most of the predictions are with less than $10\%$ error.
There are a small number of examples above than $10\%$
where errors in the prediction are due to factors
that \CAMP's cost models do
not take into account, such as
scheduler costs, cost of thread creation, or
resource contention such as shared caches.
These details that the cost models do not take into account
can also explain why, in some cases, the cost models do not predict an upper
bound of the cost, since the real executions include slowdowns due to these
factors. Note, however, that \CAMP{} offers a
quick and static first assessment of the performance behaviour of concurrent and
distributed systems which use different transports and topologies, without the
need to deploy or profile the application.

\myparagraph{Asynchronous Communication Optimisations}
Algorithms
\bench{fft},
\bench{dbuff}, \bench{nb} and \bench{ring} all rely on
asynchronous communication optimisations. Both \bench{fft} and
\bench{dbuff} require to be specified using this extension. For
\bench{ring}, we take measurements to compare the optimised and unoptimised
global types. We can observe a speedup in the execution of the protocol that is
predicted by the cost models, which is consistent with Theorem \ref{thm:opt}.

\section{Related Work}
\label{sec:relw}

\textbf{Resource Analysis and Session Types.}
\Citet{DBLP:conf/lics/Das0P18} combine session types with amortised resource
analysis in a linear type system, to reason about resource usage of
message-passing processes, but their work focuses on binary sessions in a linear
type system, while we focus on multiparty session types, and the global
execution times of the protocol.  \Citet{DBLP:journals/pacmpl/Das0P18} extend a
system of binary session types in a Curry-Howard correspondence with
intuitionistic linear logic \cite{DBLP:conf/concur/CairesP10,
DBLP:journals/mscs/CairesPT16} with temporal modalities \emph{next},
\emph{always}, and \emph{eventually}, to prescribe the timing of the
communication.  A fundamental difference with our work is that
\citet{DBLP:journals/pacmpl/Das0P18} require the introduction of delays into the
processes, to match the specified cost. In our case, the processes are left
unmodified, and the cost is computed from the protocol descriptions. Finally,
their work are limited to theory,
while our work are readily applied to real use cases.

\textbf{Asynchronous Communication Optimisation.} The first idea of
asynchronous communication optimisation was found in Scribble \cite{scribble}
where a multiparty financial protocol with message ordering permutations is
informally described. Later this idea was formalised as \emph{asynchronous
session subtyping} for the $\pi$-calculus
\cite{mostrous_yoshida_honda_esop09,MostrousY15,DBLP:conf/tlca/MostrousY09,CDSY2017,cdy14}
and its denotational properties were studied in \cite{Dezani16,DemangeonY15}.
Concurrently, because of the need of asynchronous optimisation in multiparty
protocols, several applications inspired by asynchronous subtyping have been
developed in Java \cite{H2017}, C \cite{YoshidaVPH08} and MPI-C
\cite{NgYH12,NCY2015}, but without any formal theories.
Recently, this subtyping relation was found undecidable for \emph{binary}
session types \cite{BravettiCZ17,BravettiCZ18,LY2017} and its sound algorithm
for binary session communicating automata was proposed in \cite{BCLYZ2019}.
We have implemented a different and more practical
decidable optimisation relation based on
asynchronous subtyping for multiparty session types, recently proposed
in \cite{GPPSY2021}. None of the above work
has (1) developed a formal cost theory which can justify the optimisation;
and (2) measured and compared the cost of optimised/unoptimised applications
with a formal justification. \CAMP{} is the first framework which (1)
proposes a formal cost theory with asynchronous optimisation (Theorems
\ref{thm:opt-sound} and \ref{thm:opt}) and (2)
justifies the optimisation cost against real benchmarks using (1).

\textbf{Timed Session Types.} The notion of \emph{time} has been
introduced to session types \cite{DBLP:conf/esop/BocchiMVY19,DBLP:journals/lmcs/BartolettiCM17,DBLP:conf/concur/BocchiYY14,BLY2015}, to account for protocols
that require time specifications, originated from communicating timed
automata (CTA) \cite{CTA06}.
Session types and the $\pi$-calculus processes have been related in
terms of static
typing \cite{DBLP:conf/esop/BocchiMVY19,DBLP:conf/concur/BocchiYY14},
or timed session types are linked with compliments relations
\cite{DBLP:journals/lmcs/BartolettiCM17} or CTA
\cite{BLY2015}. Among them,
\cite{DBLP:conf/esop/BocchiMVY19,DBLP:journals/lmcs/BartolettiCM17}
are limited to binary or server-client session types.
All of the above works are theoretical only,
while the work in \cite{DBLP:conf/concur/BocchiYY14} was applied to
the runtime monitoring in Python \cite{NBY2017}.
The main difference is that
the above timed session types focus on ensuring that
deadlines or time constraints are
satisfied. In contrast,
our work does not enforce any time constraints, since we
are interested on the static estimation of execution costs, but not on enforcing
that timeouts and deadlines are respected.


\textbf{Type-Based and Amortised Cost Analysis.}
\citet{POPL20Refinement} use refinement types to reason about efficiency,
cost, of Haskell programs, but they do not consider concurrency or
parallelism. Sized types \cite{DBLP:conf/popl/HughesPS96} are one of the
successful techniques for cost analysis of programs
\cite{DBLP:journals/pacmpl/AvanziniL17, DBLP:conf/ifl/PortilloHLV02,
DBLP:phd/ethos/Vasconcelos08}. Most of the uses of sized types do not deal
with concurrency and distribution. Exceptions are
\cite{DBLP:conf/popl/GimenezM16}, that address space and space-time
complexity of parallel reductions of interaction-net programs using sized and
scheduled types, but they do not address message-passing and distributed
environments. The work \cite{DBLP:conf/esop/0002S15} extends earlier
amortised cost analyses \cite{DBLP:journals/toplas/0002AH12} to parallel
reductions. Their work focuses on parallel functional programs with explicit
parallel composition, but does not address message-passing. To our best
knowledge, none of the work above addresses the cost of message-passing
constructs or distributed environments.


\section{Conclusions and Future Work}
\label{sec:concl}
We have presented \CAMP{}, a framework for statically predicting the
\emph{cost}, execution times, of concurrent and distributed systems. \CAMP{}
augments global types from the theory of multiparty session types with
\emph{local computation costs}, and its trace semantics is extended with
local computation observable actions. We have developed a way to extract cost
equations from these instrumented protocol descriptions, that we can use for
estimating upper-bounds of the execution times required by the participants
of a protocol. \CAMP{} can be used to predict the \emph{latency}, i.e.\
the execution times that the participants of a protocol will require, on
average, per iteration of the protocol. Furthermore, we extended \CAMP{} to
address \emph{asynchronous communication optimisation}.
\CAMP's cost
analysis on top of multiparty session types gives us several benefits.
Firstly, we can use global types to reason about both \emph{correctness} and
\emph{performance} of concurrent and distributed systems. Secondly, the cost
analysis can be readily applied and integrated into any \MPST{} framework.
Thirdly, it can be used in non-session-based concurrency benchmarks
by simply providing \MPST{} protocols.
It suffices to describe the global type instrumented with cost, and instantiate
the derived cost equations with measured or estimated communication latencies,
and local computation costs. And, fourthly our prototype accounts for CPU/CORE
availability of the target hardware.

\CAMP{} addresses two main concerns when estimating execution costs of
concurrent and distributed systems: communication overheads, and
synchronisation. Although these factors are a main source of inefficiency,
there are more that we still do not take into account, such as the cost of
starting new threads, the cost of context switching/scheduling, or the cost
of resource contention such as shared caches \cite{Lea97}. We plan to study how
to extend \CAMP{} to take such factors into account as future work.
\CAMP{} considers distributed systems comprised of multiple nodes,
each of which with a number of CPUs/cores. We plan to extend \CAMP's hardware
descriptions to consider heterogeneity, e.g. CPU clusters, FPGAs, etc.
\CAMP's cost models take the maximum cost of the different possible branches
in a protocol. This is sufficient to compute a worst-case execution time of
non-recursive protocols. However, we can extend our costs to take into
account the \emph{weight} of different branches, so that our cost models
would compute an average cost based on the probability to take the different
branches. Moreover, since communication latencies may not be very
predictable, we plan to study the extension of \CAMP{} to use probabilistic
cost estimations.
Finally we plan to study the development of a performance analysis tool for
existing code, based on the \emph{inference} or \emph{extraction} of the
communication protocol followed by non-session-typed implementations such as
\cite{LNTY2018,NY16,GY2020}.

\begin{acks}                            
  We thank the OOPSLA reviewers for their careful reviews and suggestions; 
  and Lorenzo Gheri and Fangyi Zhou for their comments. Francisco Ferreira
  and Keigo Imai helped testing our artifact submission.  
  The work is supported by
  EPSRC EP/T006544/1, EP/K011715/1, EP/K034413/1, EP/L00058X/1, EP/N027833/1,
  EP/N028201/1, EP/T006544/1, EP/T014709/1 and EP/V000462/1, and NCSS/EPSRC VeTSS. 
\end{acks}

\bibliography{bibliography}


\begin{thebibliography}{59}


\ifx \showCODEN    \undefined \def \showCODEN     #1{\unskip}     \fi
\ifx \showDOI      \undefined \def \showDOI       #1{#1}\fi
\ifx \showISBNx    \undefined \def \showISBNx     #1{\unskip}     \fi
\ifx \showISBNxiii \undefined \def \showISBNxiii  #1{\unskip}     \fi
\ifx \showISSN     \undefined \def \showISSN      #1{\unskip}     \fi
\ifx \showLCCN     \undefined \def \showLCCN      #1{\unskip}     \fi
\ifx \shownote     \undefined \def \shownote      #1{#1}          \fi
\ifx \showarticletitle \undefined \def \showarticletitle #1{#1}   \fi
\ifx \showURL      \undefined \def \showURL       {\relax}        \fi
\providecommand\bibfield[2]{#2}
\providecommand\bibinfo[2]{#2}
\providecommand\natexlab[1]{#1}
\providecommand\showeprint[2][]{arXiv:#2}

\bibitem[\protect\citeauthoryear{Asanovic, Bod{\'{\i}}k, Demmel, Keaveny,
  Keutzer, Kubiatowicz, Morgan, Patterson, Sen, Wawrzynek, Wessel, and
  Yelick}{Asanovic et~al\mbox{.}}{2009}]%
        {dwarfs-cacm}
\bibfield{author}{\bibinfo{person}{Krste Asanovic}, \bibinfo{person}{Rastislav
  Bod{\'{\i}}k}, \bibinfo{person}{James Demmel}, \bibinfo{person}{Tony
  Keaveny}, \bibinfo{person}{Kurt Keutzer}, \bibinfo{person}{John Kubiatowicz},
  \bibinfo{person}{Nelson Morgan}, \bibinfo{person}{David~A. Patterson},
  \bibinfo{person}{Koushik Sen}, \bibinfo{person}{John Wawrzynek},
  \bibinfo{person}{David Wessel}, {and} \bibinfo{person}{Katherine~A. Yelick}.}
  \bibinfo{year}{2009}\natexlab{}.
\newblock \showarticletitle{A view of the parallel computing landscape}.
\newblock \bibinfo{journal}{\emph{Commun. {ACM}}} \bibinfo{volume}{52},
  \bibinfo{number}{10} (\bibinfo{year}{2009}), \bibinfo{pages}{56--67}.
\newblock
\urldef\tempurl%
\url{https://doi.org/10.1145/1562764.1562783}
\showDOI{\tempurl}


\bibitem[\protect\citeauthoryear{Avanzini and {Dal Lago}}{Avanzini and {Dal
  Lago}}{2017}]%
        {DBLP:journals/pacmpl/AvanziniL17}
\bibfield{author}{\bibinfo{person}{Martin Avanzini} {and} \bibinfo{person}{Ugo
  {Dal Lago}}.} \bibinfo{year}{2017}\natexlab{}.
\newblock \showarticletitle{Automating sized-type inference for complexity
  analysis}.
\newblock \bibinfo{journal}{\emph{{PACMPL}}} \bibinfo{volume}{1},
  \bibinfo{number}{{ICFP}} (\bibinfo{year}{2017}),
  \bibinfo{pages}{43:1--43:29}.
\newblock
\urldef\tempurl%
\url{https://doi.org/10.1145/3110287}
\showDOI{\tempurl}


\bibitem[\protect\citeauthoryear{Bartoletti, Cimoli, and Murgia}{Bartoletti
  et~al\mbox{.}}{2017}]%
        {DBLP:journals/lmcs/BartolettiCM17}
\bibfield{author}{\bibinfo{person}{Massimo Bartoletti},
  \bibinfo{person}{Tiziana Cimoli}, {and} \bibinfo{person}{Maurizio Murgia}.}
  \bibinfo{year}{2017}\natexlab{}.
\newblock \showarticletitle{{Timed Session Types}}.
\newblock \bibinfo{journal}{\emph{Logical Methods in Computer Science}}
  \bibinfo{volume}{13}, \bibinfo{number}{4} (\bibinfo{year}{2017}).
\newblock
\urldef\tempurl%
\url{https://doi.org/10.23638/LMCS-13(4:25)2017}
\showDOI{\tempurl}


\bibitem[\protect\citeauthoryear{Bocchi, Lange, and Yoshida}{Bocchi
  et~al\mbox{.}}{2015}]%
        {BLY2015}
\bibfield{author}{\bibinfo{person}{Laura Bocchi}, \bibinfo{person}{Julien
  Lange}, {and} \bibinfo{person}{Nobuko Yoshida}.}
  \bibinfo{year}{2015}\natexlab{}.
\newblock \showarticletitle{{Meeting Deadlines Together}}. In
  \bibinfo{booktitle}{\emph{26th International Conference on Concurrency
  Theory}} \emph{(\bibinfo{series}{LIPIcs})}, Vol.~\bibinfo{volume}{42}.
  \bibinfo{publisher}{Schloss Dagstuhl}, \bibinfo{pages}{283--296}.
\newblock


\bibitem[\protect\citeauthoryear{Bocchi, Murgia, Vasconcelos, and
  Yoshida}{Bocchi et~al\mbox{.}}{2019}]%
        {DBLP:conf/esop/BocchiMVY19}
\bibfield{author}{\bibinfo{person}{Laura Bocchi}, \bibinfo{person}{Maurizio
  Murgia}, \bibinfo{person}{Vasco~Thudichum Vasconcelos}, {and}
  \bibinfo{person}{Nobuko Yoshida}.} \bibinfo{year}{2019}\natexlab{}.
\newblock \showarticletitle{Asynchronous Timed Session Types - From Duality to
  Time-Sensitive Processes}. In \bibinfo{booktitle}{\emph{28th European
  Symposium on Programming, {ESOP} 2019}} \emph{(\bibinfo{series}{LNCS})},
  \bibfield{editor}{\bibinfo{person}{Lu{\'{\i}}s Caires}} (Ed.),
  Vol.~\bibinfo{volume}{11423}. \bibinfo{publisher}{Springer},
  \bibinfo{pages}{583--610}.
\newblock
\showISBNx{978-3-030-17183-4}
\urldef\tempurl%
\url{https://doi.org/10.1007/978-3-030-17184-1\_21}
\showDOI{\tempurl}


\bibitem[\protect\citeauthoryear{Bocchi, Yang, and Yoshida}{Bocchi
  et~al\mbox{.}}{2014}]%
        {DBLP:conf/concur/BocchiYY14}
\bibfield{author}{\bibinfo{person}{Laura Bocchi}, \bibinfo{person}{Weizhen
  Yang}, {and} \bibinfo{person}{Nobuko Yoshida}.}
  \bibinfo{year}{2014}\natexlab{}.
\newblock \showarticletitle{Timed Multiparty Session Types}. In
  \bibinfo{booktitle}{\emph{{CONCUR} 2014 - Concurrency Theory - 25th
  International Conference, {CONCUR} 2014, Rome, Italy, September 2-5, 2014.
  Proceedings}} \emph{(\bibinfo{series}{LNCS})},
  \bibfield{editor}{\bibinfo{person}{Paolo Baldan} {and}
  \bibinfo{person}{Daniele Gorla}} (Eds.), Vol.~\bibinfo{volume}{8704}.
  \bibinfo{publisher}{Springer}, \bibinfo{pages}{419--434}.
\newblock
\showISBNx{978-3-662-44583-9}
\urldef\tempurl%
\url{https://doi.org/10.1007/978-3-662-44584-6\_29}
\showDOI{\tempurl}


\bibitem[\protect\citeauthoryear{Bravetti, Carbone, Lange, Yoshida, and
  Zavattaro}{Bravetti et~al\mbox{.}}{2019}]%
        {BCLYZ2019}
\bibfield{author}{\bibinfo{person}{Mario Bravetti}, \bibinfo{person}{Marco
  Carbone}, \bibinfo{person}{Julien Lange}, \bibinfo{person}{Nobuko Yoshida},
  {and} \bibinfo{person}{Gianluigi Zavattaro}.}
  \bibinfo{year}{2019}\natexlab{}.
\newblock \showarticletitle{{A Sound Algorithm for Asynchronous Session
  Subtyping}}. In \bibinfo{booktitle}{\emph{30th International Conference on
  Concurrency Theory}} \emph{(\bibinfo{series}{LIPIcs})},
  Vol.~\bibinfo{volume}{140}. \bibinfo{publisher}{Schloss Dagstuhl -
  Leibniz-Zentrum f{\"{u}}r Informatik}.
\newblock


\bibitem[\protect\citeauthoryear{Bravetti, Carbone, and Zavattaro}{Bravetti
  et~al\mbox{.}}{2017}]%
        {BravettiCZ17}
\bibfield{author}{\bibinfo{person}{Mario Bravetti}, \bibinfo{person}{Marco
  Carbone}, {and} \bibinfo{person}{Gianluigi Zavattaro}.}
  \bibinfo{year}{2017}\natexlab{}.
\newblock \showarticletitle{Undecidability of asynchronous session subtyping}.
\newblock \bibinfo{journal}{\emph{Inf. Comput.}}  \bibinfo{volume}{256}
  (\bibinfo{year}{2017}), \bibinfo{pages}{300--320}.
\newblock


\bibitem[\protect\citeauthoryear{Bravetti, Carbone, and Zavattaro}{Bravetti
  et~al\mbox{.}}{2018}]%
        {BravettiCZ18}
\bibfield{author}{\bibinfo{person}{Mario Bravetti}, \bibinfo{person}{Marco
  Carbone}, {and} \bibinfo{person}{Gianluigi Zavattaro}.}
  \bibinfo{year}{2018}\natexlab{}.
\newblock \showarticletitle{On the boundary between decidability and
  undecidability of asynchronous session subtyping}.
\newblock \bibinfo{journal}{\emph{Theor. Comput. Sci.}}  \bibinfo{volume}{722}
  (\bibinfo{year}{2018}), \bibinfo{pages}{19--51}.
\newblock
\urldef\tempurl%
\url{https://doi.org/10.1016/j.tcs.2018.02.010}
\showDOI{\tempurl}


\bibitem[\protect\citeauthoryear{Caires and Pfenning}{Caires and
  Pfenning}{2010}]%
        {DBLP:conf/concur/CairesP10}
\bibfield{author}{\bibinfo{person}{Lu{\'{\i}}s Caires} {and}
  \bibinfo{person}{Frank Pfenning}.} \bibinfo{year}{2010}\natexlab{}.
\newblock \showarticletitle{Session Types as Intuitionistic Linear
  Propositions}. In \bibinfo{booktitle}{\emph{{CONCUR} 2010 - Concurrency
  Theory, 21th International Conference, {CONCUR} 2010, Paris, France, August
  31-September 3, 2010. Proceedings}} \emph{(\bibinfo{series}{LNCS})},
  \bibfield{editor}{\bibinfo{person}{Paul Gastin} {and}
  \bibinfo{person}{Fran{\c{c}}ois Laroussinie}} (Eds.),
  Vol.~\bibinfo{volume}{6269}. \bibinfo{publisher}{Springer},
  \bibinfo{pages}{222--236}.
\newblock
\showISBNx{978-3-642-15374-7}
\urldef\tempurl%
\url{https://doi.org/10.1007/978-3-642-15375-4\_16}
\showDOI{\tempurl}


\bibitem[\protect\citeauthoryear{Caires, Pfenning, and Toninho}{Caires
  et~al\mbox{.}}{2016}]%
        {DBLP:journals/mscs/CairesPT16}
\bibfield{author}{\bibinfo{person}{Lu{\'{\i}}s Caires}, \bibinfo{person}{Frank
  Pfenning}, {and} \bibinfo{person}{Bernardo Toninho}.}
  \bibinfo{year}{2016}\natexlab{}.
\newblock \showarticletitle{Linear logic propositions as session types}.
\newblock \bibinfo{journal}{\emph{Mathematical Structures in Computer Science}}
  \bibinfo{volume}{26}, \bibinfo{number}{3} (\bibinfo{year}{2016}),
  \bibinfo{pages}{367--423}.
\newblock
\urldef\tempurl%
\url{https://doi.org/10.1017/S0960129514000218}
\showDOI{\tempurl}


\bibitem[\protect\citeauthoryear{Castro, Hu, Jongmans, Ng, and Yoshida}{Castro
  et~al\mbox{.}}{2019}]%
        {CHJNY:2019}
\bibfield{author}{\bibinfo{person}{David Castro}, \bibinfo{person}{Raymond Hu},
  \bibinfo{person}{Sung-Shik Jongmans}, \bibinfo{person}{Nicholas Ng}, {and}
  \bibinfo{person}{Nobuko Yoshida}.} \bibinfo{year}{2019}\natexlab{}.
\newblock \showarticletitle{Distributed Programming using Role-Parametric
  Session Types in Go} \emph{(\bibinfo{series}{POPL'19})}.
  \bibinfo{publisher}{ACM}, \bibinfo{address}{New York, NY, USA}, 12.
\newblock


\bibitem[\protect\citeauthoryear{Castro-Perez and Yoshida}{Castro-Perez and
  Yoshida}{2020}]%
        {CY:2020}
\bibfield{author}{\bibinfo{person}{David Castro-Perez} {and}
  \bibinfo{person}{Nobuko Yoshida}.} \bibinfo{year}{2020}\natexlab{}.
\newblock \showarticletitle{{Compiling First-Order Functions to Session-Typed
  Parallel Code}}. In \bibinfo{booktitle}{\emph{Proc. of the 29th Int. Conf. on
  Compiler Construction (CC2020)}} \emph{(\bibinfo{series}{CC 2020})}.
  \bibinfo{publisher}{ACM}, \bibinfo{address}{New York, NY, USA},
  \bibinfo{pages}{143–154}.
\newblock
\showISBNx{9781450371209}
\urldef\tempurl%
\url{https://doi.org/10.1145/3377555.3377889}
\showDOI{\tempurl}


\bibitem[\protect\citeauthoryear{Chen, Dezani-Ciancaglini, Scalas, and
  Yoshida}{Chen et~al\mbox{.}}{2017}]%
        {CDSY2017}
\bibfield{author}{\bibinfo{person}{Tzu-Chun Chen}, \bibinfo{person}{Mariangiola
  Dezani-Ciancaglini}, \bibinfo{person}{Alceste Scalas}, {and}
  \bibinfo{person}{Nobuko Yoshida}.} \bibinfo{year}{2017}\natexlab{}.
\newblock \showarticletitle{{On the Preciseness of Subtyping in Session
  Types}}.
\newblock \bibinfo{journal}{\emph{LMCS}}  \bibinfo{volume}{13}
  (\bibinfo{year}{2017}), \bibinfo{pages}{1--62}.
\newblock
Issue 2.


\bibitem[\protect\citeauthoryear{Chen, Dezani-Ciancaglini, and Yoshida}{Chen
  et~al\mbox{.}}{2014}]%
        {cdy14}
\bibfield{author}{\bibinfo{person}{Tzu-Chun Chen}, \bibinfo{person}{Mariangiola
  Dezani-Ciancaglini}, {and} \bibinfo{person}{Nobuko Yoshida}.}
  \bibinfo{year}{2014}\natexlab{}.
\newblock \showarticletitle{{On the Preciseness of Subtyping in Session
  Types}}. In \bibinfo{booktitle}{\emph{PPDP}}. \bibinfo{publisher}{ACM Press},
  \bibinfo{pages}{135--146}.
\newblock


\bibitem[\protect\citeauthoryear{Coppo, Dezani-Ciancaglini, Padovani, and
  Yoshida}{Coppo et~al\mbox{.}}{2015}]%
        {CDPY2015}
\bibfield{author}{\bibinfo{person}{Mario Coppo}, \bibinfo{person}{Mariangiola
  Dezani-Ciancaglini}, \bibinfo{person}{Luca Padovani}, {and}
  \bibinfo{person}{Nobuko Yoshida}.} \bibinfo{year}{2015}\natexlab{}.
\newblock \showarticletitle{{A Gentle Introduction to Multiparty Asynchronous
  Session Types}}. In \bibinfo{booktitle}{\emph{15th International School on
  Formal Methods for the Design of Computer, Communication and Software
  Systems: Multicore Programming}} \emph{(\bibinfo{series}{LNCS})},
  Vol.~\bibinfo{volume}{9104}. \bibinfo{publisher}{Springer},
  \bibinfo{pages}{146--178}.
\newblock


\bibitem[\protect\citeauthoryear{Das, Hoffmann, and Pfenning}{Das
  et~al\mbox{.}}{2018a}]%
        {DBLP:journals/pacmpl/Das0P18}
\bibfield{author}{\bibinfo{person}{Ankush Das}, \bibinfo{person}{Jan Hoffmann},
  {and} \bibinfo{person}{Frank Pfenning}.} \bibinfo{year}{2018}\natexlab{a}.
\newblock \showarticletitle{Parallel complexity analysis with temporal session
  types}.
\newblock \bibinfo{journal}{\emph{{PACMPL}}} \bibinfo{volume}{2},
  \bibinfo{number}{{ICFP}} (\bibinfo{year}{2018}),
  \bibinfo{pages}{91:1--91:30}.
\newblock
\urldef\tempurl%
\url{https://doi.org/10.1145/3236786}
\showDOI{\tempurl}


\bibitem[\protect\citeauthoryear{Das, Hoffmann, and Pfenning}{Das
  et~al\mbox{.}}{2018b}]%
        {DBLP:conf/lics/Das0P18}
\bibfield{author}{\bibinfo{person}{Ankush Das}, \bibinfo{person}{Jan Hoffmann},
  {and} \bibinfo{person}{Frank Pfenning}.} \bibinfo{year}{2018}\natexlab{b}.
\newblock \showarticletitle{Work Analysis with Resource-Aware Session Types}.
  In \bibinfo{booktitle}{\emph{Proceedings of the 33rd Annual {ACM/IEEE}
  Symposium on Logic in Computer Science, {LICS} 2018, Oxford, UK, July 09-12,
  2018}}, \bibfield{editor}{\bibinfo{person}{Anuj Dawar} {and}
  \bibinfo{person}{Erich Gr{\"{a}}del}} (Eds.). \bibinfo{publisher}{{ACM}},
  \bibinfo{pages}{305--314}.
\newblock
\urldef\tempurl%
\url{https://doi.org/10.1145/3209108.3209146}
\showDOI{\tempurl}


\bibitem[\protect\citeauthoryear{Demangeon and Honda}{Demangeon and
  Honda}{2012}]%
        {Demangeon2012Nested}
\bibfield{author}{\bibinfo{person}{Romain Demangeon} {and}
  \bibinfo{person}{Kohei Honda}.} \bibinfo{year}{2012}\natexlab{}.
\newblock \showarticletitle{Nested Protocols in Session Types}. In
  \bibinfo{booktitle}{\emph{CONCUR 2012 -- Concurrency Theory}},
  \bibfield{editor}{\bibinfo{person}{Maciej Koutny} {and} \bibinfo{person}{Irek
  Ulidowski}} (Eds.). \bibinfo{publisher}{Springer Berlin Heidelberg},
  \bibinfo{address}{Berlin, Heidelberg}, \bibinfo{pages}{272--286}.
\newblock
\showISBNx{978-3-642-32940-1}


\bibitem[\protect\citeauthoryear{Demangeon and Yoshida}{Demangeon and
  Yoshida}{2015}]%
        {DemangeonY15}
\bibfield{author}{\bibinfo{person}{Romain Demangeon} {and}
  \bibinfo{person}{Nobuko Yoshida}.} \bibinfo{year}{2015}\natexlab{}.
\newblock \showarticletitle{On the Expressiveness of Multiparty Sessions}. In
  \bibinfo{booktitle}{\emph{{FSTTCS} 2015}} \emph{(\bibinfo{series}{LIPIcs})},
  \bibfield{editor}{\bibinfo{person}{Prahladh Harsha} {and}
  \bibinfo{person}{G.~Ramalingam}} (Eds.), Vol.~\bibinfo{volume}{45}.
  \bibinfo{publisher}{Schloss Dagstuhl - Leibniz-Zentrum fuer Informatik},
  \bibinfo{pages}{560--574}.
\newblock
\showISBNx{978-3-939897-97-2}
\urldef\tempurl%
\url{https://doi.org/10.4230/LIPIcs.FSTTCS.2015.560}
\showDOI{\tempurl}


\bibitem[\protect\citeauthoryear{Deni{\'{e}}lou and Yoshida}{Deni{\'{e}}lou and
  Yoshida}{2013}]%
        {DY13}
\bibfield{author}{\bibinfo{person}{Pierre{-}Malo Deni{\'{e}}lou} {and}
  \bibinfo{person}{Nobuko Yoshida}.} \bibinfo{year}{2013}\natexlab{}.
\newblock \showarticletitle{Multiparty Compatibility in Communicating Automata:
  Characterisation and Synthesis of Global Session Types}. In
  \bibinfo{booktitle}{\emph{Automata, Languages, and Programming - 40th
  International Colloquium, {ICALP} 2013, Riga, Latvia, July 8-12, 2013,
  Proceedings, Part {II}}} \emph{(\bibinfo{series}{LNCS})},
  \bibfield{editor}{\bibinfo{person}{Fedor~V. Fomin}, \bibinfo{person}{Rusins
  Freivalds}, \bibinfo{person}{Marta~Z. Kwiatkowska}, {and}
  \bibinfo{person}{David Peleg}} (Eds.), Vol.~\bibinfo{volume}{7966}.
  \bibinfo{publisher}{Springer}, \bibinfo{pages}{174--186}.
\newblock
\urldef\tempurl%
\url{https://doi.org/10.1007/978-3-642-39212-2_18}
\showDOI{\tempurl}


\bibitem[\protect\citeauthoryear{Deni{\'{e}}lou, Yoshida, Bejleri, and
  Hu}{Deni{\'{e}}lou et~al\mbox{.}}{2012}]%
        {DYBH12}
\bibfield{author}{\bibinfo{person}{Pierre{-}Malo Deni{\'{e}}lou},
  \bibinfo{person}{Nobuko Yoshida}, \bibinfo{person}{Andi Bejleri}, {and}
  \bibinfo{person}{Raymond Hu}.} \bibinfo{year}{2012}\natexlab{}.
\newblock \showarticletitle{Parameterised Multiparty Session Types}.
\newblock \bibinfo{journal}{\emph{Logical Methods in Computer Science}}
  \bibinfo{volume}{8}, \bibinfo{number}{4} (\bibinfo{year}{2012}).
\newblock
\urldef\tempurl%
\url{https://doi.org/10.2168/LMCS-8(4:6)2012}
\showDOI{\tempurl}


\bibitem[\protect\citeauthoryear{Dezani{-}Ciancaglini, Ghilezan, Jaksic,
  Pantovic, and Yoshida}{Dezani{-}Ciancaglini et~al\mbox{.}}{2016}]%
        {Dezani16}
\bibfield{author}{\bibinfo{person}{Mariangiola Dezani{-}Ciancaglini},
  \bibinfo{person}{Silvia Ghilezan}, \bibinfo{person}{Svetlana Jaksic},
  \bibinfo{person}{Jovanka Pantovic}, {and} \bibinfo{person}{Nobuko Yoshida}.}
  \bibinfo{year}{2016}\natexlab{}.
\newblock \showarticletitle{Denotational and Operational Preciseness of
  Subtyping: {A} Roadmap}. In \bibinfo{booktitle}{\emph{Theory and Practice of
  Formal Methods - Essays Dedicated to Frank de Boer on the Occasion of His
  60th Birthday}}. \bibinfo{pages}{155--172}.
\newblock
\urldef\tempurl%
\url{https://doi.org/10.1007/978-3-319-30734-3_12}
\showDOI{\tempurl}


\bibitem[\protect\citeauthoryear{Gabet and Yoshida}{Gabet and Yoshida}{2020}]%
        {GY2020}
\bibfield{author}{\bibinfo{person}{Julia Gabet} {and} \bibinfo{person}{Nobuko
  Yoshida}.} \bibinfo{year}{2020}\natexlab{}.
\newblock \showarticletitle{{Static Race Detection and Mutex Safety and
  Liveness for Go Programs}} \emph{(\bibinfo{series}{LIPIcs})}.
  \bibinfo{publisher}{Schloss Dagstuhl - Leibniz-Zentrum f{\"{u}}r Informatik}.
\newblock
\newblock
\shownote{To appear in ECOOP'20.}


\bibitem[\protect\citeauthoryear{Gay and Ravara}{Gay and Ravara}{2017}]%
        {bettytoolbook}
\bibfield{editor}{\bibinfo{person}{Simon Gay} {and} \bibinfo{person}{Antonio
  Ravara}} (Eds.). \bibinfo{year}{2017}\natexlab{}.
\newblock \bibinfo{booktitle}{\emph{Behavioural Types: from Theory to Tools}}.
\newblock \bibinfo{publisher}{River Publishers}.
\newblock


\bibitem[\protect\citeauthoryear{Ghilezan, Pantovic, Prokic, Scalas, and
  Yoshida}{Ghilezan et~al\mbox{.}}{2021}]%
        {GPPSY2021}
\bibfield{author}{\bibinfo{person}{Silvia Ghilezan}, \bibinfo{person}{Jovanka
  Pantovic}, \bibinfo{person}{Ivan Prokic}, \bibinfo{person}{Alceste Scalas},
  {and} \bibinfo{person}{Nobuko Yoshida}.} \bibinfo{year}{2021}\natexlab{}.
\newblock \showarticletitle{Precise Subtyping for Asynchronous Multiparty
  Sessions}.
\newblock \bibinfo{journal}{\emph{Proc. ACM Program. Lang.}}
  \bibinfo{number}{POPL} (\bibinfo{year}{2021}).
\newblock
\newblock
\shownote{To appear in POPL'21.}


\bibitem[\protect\citeauthoryear{Gimenez and Moser}{Gimenez and Moser}{2016}]%
        {DBLP:conf/popl/GimenezM16}
\bibfield{author}{\bibinfo{person}{St{\'{e}}phane Gimenez} {and}
  \bibinfo{person}{Georg Moser}.} \bibinfo{year}{2016}\natexlab{}.
\newblock \showarticletitle{The complexity of interaction}. In
  \bibinfo{booktitle}{\emph{Proceedings of the 43rd Annual {ACM}
  {SIGPLAN-SIGACT} Symposium on Principles of Programming Languages, {POPL}
  2016, St. Petersburg, FL, USA, January 20 - 22, 2016}},
  \bibfield{editor}{\bibinfo{person}{Rastislav Bod{\'{\i}}k} {and}
  \bibinfo{person}{Rupak Majumdar}} (Eds.). \bibinfo{publisher}{{ACM}},
  \bibinfo{pages}{243--255}.
\newblock
\showISBNx{978-1-4503-3549-2}
\urldef\tempurl%
\url{https://doi.org/10.1145/2837614.2837646}
\showDOI{\tempurl}


\bibitem[\protect\citeauthoryear{Goetz, Peierls, Bloch, Bowbeer, Holmes, and
  Lea}{Goetz et~al\mbox{.}}{2006}]%
        {Goetz06}
\bibfield{author}{\bibinfo{person}{Brian Goetz}, \bibinfo{person}{Tim Peierls},
  \bibinfo{person}{Joshua~J. Bloch}, \bibinfo{person}{Joseph Bowbeer},
  \bibinfo{person}{David Holmes}, {and} \bibinfo{person}{Doug Lea}.}
  \bibinfo{year}{2006}\natexlab{}.
\newblock \bibinfo{booktitle}{\emph{Java Concurrency in Practice}}.
\newblock \bibinfo{publisher}{Addison-Wesley}.
\newblock
\showISBNx{978-0-321-34960-6}


\bibitem[\protect\citeauthoryear{Handley, Vazou, and Hutton}{Handley
  et~al\mbox{.}}{2019}]%
        {POPL20Refinement}
\bibfield{author}{\bibinfo{person}{Martin A.~T. Handley}, \bibinfo{person}{Niki
  Vazou}, {and} \bibinfo{person}{Graham Hutton}.}
  \bibinfo{year}{2019}\natexlab{}.
\newblock \showarticletitle{Liquidate Your Assets: Reasoning about Resource
  Usage in Liquid Haskell}.
\newblock \bibinfo{journal}{\emph{Proc. ACM Program. Lang.}}
  \bibinfo{volume}{4}, \bibinfo{number}{POPL}, Article
  \bibinfo{articleno}{Article 24} (\bibinfo{date}{Dec.} \bibinfo{year}{2019}),
  \bibinfo{numpages}{27}~pages.
\newblock
\urldef\tempurl%
\url{https://doi.org/10.1145/3371092}
\showDOI{\tempurl}


\bibitem[\protect\citeauthoryear{Hoffmann, Aehlig, and Hofmann}{Hoffmann
  et~al\mbox{.}}{2012}]%
        {DBLP:journals/toplas/0002AH12}
\bibfield{author}{\bibinfo{person}{Jan Hoffmann}, \bibinfo{person}{Klaus
  Aehlig}, {and} \bibinfo{person}{Martin Hofmann}.}
  \bibinfo{year}{2012}\natexlab{}.
\newblock \showarticletitle{Multivariate amortized resource analysis}.
\newblock \bibinfo{journal}{\emph{{ACM} Trans. Program. Lang. Syst.}}
  \bibinfo{volume}{34}, \bibinfo{number}{3} (\bibinfo{year}{2012}),
  \bibinfo{pages}{14:1--14:62}.
\newblock
\urldef\tempurl%
\url{https://doi.org/10.1145/2362389.2362393}
\showDOI{\tempurl}


\bibitem[\protect\citeauthoryear{Hoffmann and Shao}{Hoffmann and Shao}{2015}]%
        {DBLP:conf/esop/0002S15}
\bibfield{author}{\bibinfo{person}{Jan Hoffmann} {and} \bibinfo{person}{Zhong
  Shao}.} \bibinfo{year}{2015}\natexlab{}.
\newblock \showarticletitle{Automatic Static Cost Analysis for Parallel
  Programs}. In \bibinfo{booktitle}{\emph{24th European Symposium on
  Programming, {ESOP} 2015}} \emph{(\bibinfo{series}{LNCS})},
  \bibfield{editor}{\bibinfo{person}{Jan Vitek}} (Ed.),
  Vol.~\bibinfo{volume}{9032}. \bibinfo{publisher}{Springer},
  \bibinfo{pages}{132--157}.
\newblock
\showISBNx{978-3-662-46668-1}
\urldef\tempurl%
\url{https://doi.org/10.1007/978-3-662-46669-8\_6}
\showDOI{\tempurl}


\bibitem[\protect\citeauthoryear{Honda, Yoshida, and Carbone}{Honda
  et~al\mbox{.}}{2008}]%
        {Honda2008Multiparty}
\bibfield{author}{\bibinfo{person}{Kohei Honda}, \bibinfo{person}{Nobuko
  Yoshida}, {and} \bibinfo{person}{Marco Carbone}.}
  \bibinfo{year}{2008}\natexlab{}.
\newblock \showarticletitle{Multiparty Asynchronous Session Types}. In
  \bibinfo{booktitle}{\emph{Proc. of 35th Symp. on Princ. of Prog. Lang.}}
  \emph{(\bibinfo{series}{POPL '08})}. \bibinfo{publisher}{ACM},
  \bibinfo{address}{New York, NY, USA}, \bibinfo{pages}{273--284}.
\newblock
\showISBNx{978-1-59593-689-9}
\urldef\tempurl%
\url{https://doi.org/10.1145/1328438.1328472}
\showDOI{\tempurl}


\bibitem[\protect\citeauthoryear{Hu}{Hu}{2017}]%
        {H2017}
\bibfield{author}{\bibinfo{person}{Raymond Hu}.}
  \bibinfo{year}{2017}\natexlab{}.
\newblock \showarticletitle{{Distributed Programming Using Java APIs Generated
  from Session Types}}.
\newblock \bibinfo{journal}{\emph{Behavioural Types: from Theory to Tools}}
  (\bibinfo{year}{2017}), \bibinfo{pages}{287--308}.
\newblock


\bibitem[\protect\citeauthoryear{Hu and Yoshida}{Hu and Yoshida}{2017}]%
        {HY2017}
\bibfield{author}{\bibinfo{person}{Raymond Hu} {and} \bibinfo{person}{Nobuko
  Yoshida}.} \bibinfo{year}{2017}\natexlab{}.
\newblock \showarticletitle{Explicit Connection Actions in Multiparty Session
  Types}. In \bibinfo{booktitle}{\emph{20th Int. Conf. on Fundamental
  Approaches to Software Engineering, {FASE} 2017}}
  \emph{(\bibinfo{series}{LNCS})}, \bibfield{editor}{\bibinfo{person}{Marieke
  Huisman} {and} \bibinfo{person}{Julia Rubin}} (Eds.),
  Vol.~\bibinfo{volume}{10202}. \bibinfo{publisher}{Springer},
  \bibinfo{pages}{116--133}.
\newblock
\urldef\tempurl%
\url{https://doi.org/10.1007/978-3-662-54494-5_7}
\showDOI{\tempurl}


\bibitem[\protect\citeauthoryear{Hughes, Pareto, and Sabry}{Hughes
  et~al\mbox{.}}{1996}]%
        {DBLP:conf/popl/HughesPS96}
\bibfield{author}{\bibinfo{person}{John Hughes}, \bibinfo{person}{Lars Pareto},
  {and} \bibinfo{person}{Amr Sabry}.} \bibinfo{year}{1996}\natexlab{}.
\newblock \showarticletitle{Proving the Correctness of Reactive Systems Using
  Sized Types}. In \bibinfo{booktitle}{\emph{Conference Record of POPL'96: The
  23rd {ACM} {SIGPLAN-SIGACT} Symposium on Principles of Programming Languages,
  Papers Presented at the Symposium, St. Petersburg Beach, Florida, USA,
  January 21-24, 1996}}, \bibfield{editor}{\bibinfo{person}{Hans{-}Juergen
  Boehm} {and} \bibinfo{person}{Guy L.~Steele Jr.}} (Eds.).
  \bibinfo{publisher}{{ACM} Press}, \bibinfo{pages}{410--423}.
\newblock
\showISBNx{0-89791-769-3}
\urldef\tempurl%
\url{https://doi.org/10.1145/237721.240882}
\showDOI{\tempurl}


\bibitem[\protect\citeauthoryear{Imai, Neykova, Yoshida, and Yuen}{Imai
  et~al\mbox{.}}{2020}]%
        {INY19}
\bibfield{author}{\bibinfo{person}{Keigo Imai}, \bibinfo{person}{Rumyana
  Neykova}, \bibinfo{person}{Nobuko Yoshida}, {and} \bibinfo{person}{Shoji
  Yuen}.} \bibinfo{year}{2020}\natexlab{}.
\newblock \showarticletitle{{Multiparty Session Programming with Global
  Protocol Combinators}}.
  \bibinfo{howpublished}{\url{https://github.com/keigoi/ocaml-mpst}}
  \emph{(\bibinfo{series}{LIPIcs})}. \bibinfo{publisher}{Schloss Dagstuhl -
  Leibniz-Zentrum f{\"{u}}r Informatik}.
\newblock
\newblock
\shownote{To appear in ECOOP'20.}


\bibitem[\protect\citeauthoryear{Imam and Sarkar}{Imam and Sarkar}{2014}]%
        {IS14}
\bibfield{author}{\bibinfo{person}{Shams~M. Imam} {and} \bibinfo{person}{Vivek
  Sarkar}.} \bibinfo{year}{2014}\natexlab{}.
\newblock \showarticletitle{{Savina - An Actor Benchmark Suite: Enabling
  Empirical Evaluation of Actor Libraries}}. In \bibinfo{booktitle}{\emph{Proc.
  of the 4th Int. Workshop on Programming Based on Actors Agents \&
  Decentralized Control}} \emph{(\bibinfo{series}{AGERE! ’14})}.
  \bibinfo{publisher}{Association for Computing Machinery},
  \bibinfo{address}{New York, NY, USA}, \bibinfo{pages}{67–80}.
\newblock
\showISBNx{9781450321891}


\bibitem[\protect\citeauthoryear{Jin, Song, Shi, Scherpelz, and Lu}{Jin
  et~al\mbox{.}}{2012}]%
        {DBLP:conf/pldi/JinSSSL12}
\bibfield{author}{\bibinfo{person}{Guoliang Jin}, \bibinfo{person}{Linhai
  Song}, \bibinfo{person}{Xiaoming Shi}, \bibinfo{person}{Joel Scherpelz},
  {and} \bibinfo{person}{Shan Lu}.} \bibinfo{year}{2012}\natexlab{}.
\newblock \showarticletitle{Understanding and detecting real-world performance
  bugs}. In \bibinfo{booktitle}{\emph{{ACM} {SIGPLAN} Conference on Programming
  Language Design and Implementation, {PLDI} '12, Beijing, China - June 11 -
  16, 2012}}, \bibfield{editor}{\bibinfo{person}{Jan Vitek},
  \bibinfo{person}{Haibo Lin}, {and} \bibinfo{person}{Frank Tip}} (Eds.).
  \bibinfo{publisher}{{ACM}}, \bibinfo{pages}{77--88}.
\newblock
\showISBNx{978-1-4503-1205-9}
\urldef\tempurl%
\url{https://doi.org/10.1145/2254064.2254075}
\showDOI{\tempurl}


\bibitem[\protect\citeauthoryear{Krommydas et~al\mbox{.}}{Krommydas
  et~al\mbox{.}}{2016}]%
        {Krommydas16}
\bibfield{author}{\bibinfo{person}{K. Krommydas} {et~al\mbox{.}}}
  \bibinfo{year}{2016}\natexlab{}.
\newblock \showarticletitle{{OpenDwarfs}: Characterization of Dwarf-Based
  Benchmarks on Fixed and Reconfigurable Architectures}.
\newblock \bibinfo{journal}{\emph{J Sign Process Syst}}  \bibinfo{volume}{85}
  (\bibinfo{year}{2016}), \bibinfo{pages}{373–--392}.
\newblock


\bibitem[\protect\citeauthoryear{Krčál and yi}{Krčál and yi}{2006}]%
        {CTA06}
\bibfield{author}{\bibinfo{person}{Pavel Krčál} {and} \bibinfo{person}{Wang
  yi}.} \bibinfo{year}{2006}\natexlab{}.
\newblock \showarticletitle{Communicating Timed Automata: The More Synchronous,
  the More Difficult to Verify}. \bibinfo{pages}{249--262}.
\newblock
\urldef\tempurl%
\url{https://doi.org/10.1007/11817963_24}
\showDOI{\tempurl}


\bibitem[\protect\citeauthoryear{Lange, Ng, Toninho, and Yoshida}{Lange
  et~al\mbox{.}}{2018}]%
        {LNTY2018}
\bibfield{author}{\bibinfo{person}{Julien Lange}, \bibinfo{person}{Nicholas
  Ng}, \bibinfo{person}{Bernardo Toninho}, {and} \bibinfo{person}{Nobuko
  Yoshida}.} \bibinfo{year}{2018}\natexlab{}.
\newblock \showarticletitle{{A Static Verification Framework for Message
  Passing in Go using Behavioural Types}}. In \bibinfo{booktitle}{\emph{40th
  International Conference on Software Engineering}}. \bibinfo{publisher}{ACM},
  \bibinfo{pages}{1137--1148}.
\newblock


\bibitem[\protect\citeauthoryear{Lange and Yoshida}{Lange and Yoshida}{2017}]%
        {LY2017}
\bibfield{author}{\bibinfo{person}{Julien Lange} {and} \bibinfo{person}{Nobuko
  Yoshida}.} \bibinfo{year}{2017}\natexlab{}.
\newblock \showarticletitle{{On the Undecidability of Asynchronous Session
  Subtyping}}. In \bibinfo{booktitle}{\emph{20th International Conference on
  Foundations of Software Science and Computation Structures}}
  \emph{(\bibinfo{series}{LNCS})}, Vol.~\bibinfo{volume}{10203}.
  \bibinfo{publisher}{Springer}, \bibinfo{pages}{441--457}.
\newblock


\bibitem[\protect\citeauthoryear{Lea}{Lea}{1997}]%
        {Lea97}
\bibfield{author}{\bibinfo{person}{Doug Lea}.} \bibinfo{year}{1997}\natexlab{}.
\newblock \bibinfo{booktitle}{\emph{Concurrent programming in Java - design
  principles and patterns}}.
\newblock \bibinfo{publisher}{Addison-Wesley-Longman}.
\newblock
\showISBNx{978-0-201-69581-6}


\bibitem[\protect\citeauthoryear{Mostrous and Yoshida}{Mostrous and
  Yoshida}{2009}]%
        {DBLP:conf/tlca/MostrousY09}
\bibfield{author}{\bibinfo{person}{Dimitris Mostrous} {and}
  \bibinfo{person}{Nobuko Yoshida}.} \bibinfo{year}{2009}\natexlab{}.
\newblock \showarticletitle{Session-Based Communication Optimisation for
  Higher-Order Mobile Processes}. In \bibinfo{booktitle}{\emph{Typed Lambda
  Calculi and Applications, 9th International Conference, {TLCA} 2009,
  Brasilia, Brazil, July 1-3, 2009. Proceedings}}
  \emph{(\bibinfo{series}{Lecture Notes in Computer Science})},
  \bibfield{editor}{\bibinfo{person}{Pierre{-}Louis Curien}} (Ed.),
  Vol.~\bibinfo{volume}{5608}. \bibinfo{publisher}{Springer},
  \bibinfo{pages}{203--218}.
\newblock
\showISBNx{978-3-642-02272-2}
\urldef\tempurl%
\url{https://doi.org/10.1007/978-3-642-02273-9\_16}
\showDOI{\tempurl}


\bibitem[\protect\citeauthoryear{Mostrous and Yoshida}{Mostrous and
  Yoshida}{2015}]%
        {MostrousY15}
\bibfield{author}{\bibinfo{person}{Dimitris Mostrous} {and}
  \bibinfo{person}{Nobuko Yoshida}.} \bibinfo{year}{2015}\natexlab{}.
\newblock \showarticletitle{{Session Typing and Asynchronous Subtying for
  Higher-Order $\pi$-Calculus}}.
\newblock \bibinfo{journal}{\emph{Info.\& Comp.}}  \bibinfo{volume}{241}
  (\bibinfo{year}{2015}), \bibinfo{pages}{227--263}.
\newblock


\bibitem[\protect\citeauthoryear{Mostrous, Yoshida, and Honda}{Mostrous
  et~al\mbox{.}}{2009}]%
        {mostrous_yoshida_honda_esop09}
\bibfield{author}{\bibinfo{person}{Dimitris Mostrous}, \bibinfo{person}{Nobuko
  Yoshida}, {and} \bibinfo{person}{Kohei Honda}.}
  \bibinfo{year}{2009}\natexlab{}.
\newblock \showarticletitle{Global Principal Typing in Partially Commutative
  Asynchronous Sessions}. In \bibinfo{booktitle}{\emph{ESOP}}
  \emph{(\bibinfo{series}{LNCS})}, Vol.~\bibinfo{volume}{5502}.
  \bibinfo{publisher}{Springer}, \bibinfo{pages}{316--332}.
\newblock


\bibitem[\protect\citeauthoryear{Neykova, Bocchi, and Yoshida}{Neykova
  et~al\mbox{.}}{2017}]%
        {NBY2017}
\bibfield{author}{\bibinfo{person}{Rumyana Neykova}, \bibinfo{person}{Laura
  Bocchi}, {and} \bibinfo{person}{Nobuko Yoshida}.}
  \bibinfo{year}{2017}\natexlab{}.
\newblock \showarticletitle{{Timed Runtime Monitoring for Multiparty
  Conversations}}.
\newblock \bibinfo{journal}{\emph{FAOC}} (\bibinfo{year}{2017}),
  \bibinfo{pages}{1--34}.
\newblock


\bibitem[\protect\citeauthoryear{Ng, de~Figueiredo~Coutinho, and Yoshida}{Ng
  et~al\mbox{.}}{2015}]%
        {NCY2015}
\bibfield{author}{\bibinfo{person}{Nicholas Ng},
  \bibinfo{person}{Jos{\'{e}}~Gabriel de Figueiredo~Coutinho}, {and}
  \bibinfo{person}{Nobuko Yoshida}.} \bibinfo{year}{2015}\natexlab{}.
\newblock \showarticletitle{Protocols by Default - Safe {MPI} Code Generation
  Based on Session Types}. In \bibinfo{booktitle}{\emph{24th Int. Conf. on
  Compiler Construction, {CC} 2015}} \emph{(\bibinfo{series}{LNCS})},
  \bibfield{editor}{\bibinfo{person}{Bj{\"{o}}rn Franke}} (Ed.),
  Vol.~\bibinfo{volume}{9031}. \bibinfo{publisher}{Springer},
  \bibinfo{pages}{212--232}.
\newblock
\urldef\tempurl%
\url{https://doi.org/10.1007/978-3-662-46663-6_11}
\showDOI{\tempurl}


\bibitem[\protect\citeauthoryear{Ng and Yoshida}{Ng and Yoshida}{2016}]%
        {NY16}
\bibfield{author}{\bibinfo{person}{Nicholas Ng} {and} \bibinfo{person}{Nobuko
  Yoshida}.} \bibinfo{year}{2016}\natexlab{}.
\newblock \showarticletitle{Static deadlock detection for concurrent {Go} by
  global session graph synthesis}. In \bibinfo{booktitle}{\emph{Proceedings of
  the 25th International Conference on Compiler Construction, {CC} 2016,
  Barcelona, Spain, March 12-18, 2016}},
  \bibfield{editor}{\bibinfo{person}{Ayal Zaks} {and}
  \bibinfo{person}{Manuel~V. Hermenegildo}} (Eds.). \bibinfo{publisher}{{ACM}},
  \bibinfo{pages}{174--184}.
\newblock
\urldef\tempurl%
\url{https://doi.org/10.1145/2892208.2892232}
\showDOI{\tempurl}


\bibitem[\protect\citeauthoryear{Ng, Yoshida, and Honda}{Ng
  et~al\mbox{.}}{2012}]%
        {NgYH12}
\bibfield{author}{\bibinfo{person}{Nicholas Ng}, \bibinfo{person}{Nobuko
  Yoshida}, {and} \bibinfo{person}{Kohei Honda}.}
  \bibinfo{year}{2012}\natexlab{}.
\newblock \showarticletitle{Multiparty Session {C:} Safe Parallel Programming
  with Message Optimisation}. In \bibinfo{booktitle}{\emph{Objects, Models,
  Components, Patterns - 50th International Conference, {TOOLS} 2012, Prague,
  Czech Republic, May 29-31, 2012. Proceedings}}
  \emph{(\bibinfo{series}{Lecture Notes in Computer Science})},
  \bibfield{editor}{\bibinfo{person}{Carlo~A. Furia} {and}
  \bibinfo{person}{Sebastian Nanz}} (Eds.), Vol.~\bibinfo{volume}{7304}.
  \bibinfo{publisher}{Springer}, \bibinfo{pages}{202--218}.
\newblock
\urldef\tempurl%
\url{https://doi.org/10.1007/978-3-642-30561-0\_15}
\showDOI{\tempurl}


\bibitem[\protect\citeauthoryear{Pierce}{Pierce}{2002}]%
        {pierce2002types}
\bibfield{author}{\bibinfo{person}{Benjamin~C Pierce}.}
  \bibinfo{year}{2002}\natexlab{}.
\newblock \bibinfo{booktitle}{\emph{Types and programming languages}}.
\newblock \bibinfo{publisher}{The MIT Press}.
\newblock


\bibitem[\protect\citeauthoryear{Portillo, Hammond, Loidl, and
  Vasconcelos}{Portillo et~al\mbox{.}}{2002}]%
        {DBLP:conf/ifl/PortilloHLV02}
\bibfield{author}{\bibinfo{person}{{\'{A}}lvaro J.~Reb{\'{o}}n Portillo},
  \bibinfo{person}{Kevin Hammond}, \bibinfo{person}{Hans{-}Wolfgang Loidl},
  {and} \bibinfo{person}{Pedro~B. Vasconcelos}.}
  \bibinfo{year}{2002}\natexlab{}.
\newblock \showarticletitle{Cost Analysis Using Automatic Size and Time
  Inference}. In \bibinfo{booktitle}{\emph{Implementation of Functional
  Languages, 14th International Workshop, {IFL} 2002, Madrid, Spain, September
  16-18, 2002, Revised Selected Papers}} \emph{(\bibinfo{series}{LNCS})},
  \bibfield{editor}{\bibinfo{person}{Ricardo Pena} {and}
  \bibinfo{person}{Thomas Arts}} (Eds.), Vol.~\bibinfo{volume}{2670}.
  \bibinfo{publisher}{Springer}, \bibinfo{pages}{232--248}.
\newblock
\showISBNx{3-540-40190-3}
\urldef\tempurl%
\url{https://doi.org/10.1007/3-540-44854-3\_15}
\showDOI{\tempurl}


\bibitem[\protect\citeauthoryear{Rauber and R{\"{u}}nger}{Rauber and
  R{\"{u}}nger}{2010}]%
        {Rauber10}
\bibfield{author}{\bibinfo{person}{Thomas Rauber} {and} \bibinfo{person}{Gudula
  R{\"{u}}nger}.} \bibinfo{year}{2010}\natexlab{}.
\newblock \bibinfo{booktitle}{\emph{Parallel Programming - for Multicore and
  Cluster Systems}}.
\newblock \bibinfo{publisher}{Springer}.
\newblock
\showISBNx{978-3-642-04817-3}
\urldef\tempurl%
\url{https://doi.org/10.1007/978-3-642-04818-0}
\showDOI{\tempurl}


\bibitem[\protect\citeauthoryear{scribble authors}{scribble authors}{2008}]%
        {scribble}
\bibfield{author}{\bibinfo{person}{The scribble authors}.}
  \bibinfo{year}{2008}\natexlab{}.
\newblock \bibinfo{title}{Scribble homepage}.
\newblock \bibinfo{howpublished}{\url{https://www.scribble.com}}.
\newblock


\bibitem[\protect\citeauthoryear{Taubenfeld}{Taubenfeld}{2006}]%
        {Taubenfeld06}
\bibfield{author}{\bibinfo{person}{Gadi Taubenfeld}.}
  \bibinfo{year}{2006}\natexlab{}.
\newblock \bibinfo{booktitle}{\emph{Synchronization algorithms and concurrent
  programming}}.
\newblock \bibinfo{publisher}{Pearson Education}.
\newblock


\bibitem[\protect\citeauthoryear{{USENIX}}{{USENIX}}{2020}]%
        {doublebuffer}
\bibfield{author}{\bibinfo{person}{{USENIX}}.} \bibinfo{year}{2020}\natexlab{}.
\newblock \bibinfo{title}{{Double-Buffering Algorithm} (web)}.
\newblock
  \bibinfo{howpublished}{\url{https://www.usenix.org/legacy/publications/library/proceedings/usenix02/full_papers/huang/huang_html/node8.html}}.
\newblock


\bibitem[\protect\citeauthoryear{Vasconcelos}{Vasconcelos}{2008}]%
        {DBLP:phd/ethos/Vasconcelos08}
\bibfield{author}{\bibinfo{person}{Pedro~B. Vasconcelos}.}
  \bibinfo{year}{2008}\natexlab{}.
\newblock \emph{\bibinfo{title}{Space cost analysis using sized types}}.
\newblock \bibinfo{thesistype}{Ph.D. Dissertation}. \bibinfo{school}{University
  of St Andrews, {UK}}.
\newblock
\urldef\tempurl%
\url{http://hdl.handle.net/10023/564}
\showURL{%
\tempurl}


\bibitem[\protect\citeauthoryear{Yoshida, Vasconcelos, Paulino, and
  Honda}{Yoshida et~al\mbox{.}}{2008}]%
        {YoshidaVPH08}
\bibfield{author}{\bibinfo{person}{Nobuko Yoshida},
  \bibinfo{person}{Vasco~Thudichum Vasconcelos}, \bibinfo{person}{Herv{\'{e}}
  Paulino}, {and} \bibinfo{person}{Kohei Honda}.}
  \bibinfo{year}{2008}\natexlab{}.
\newblock \showarticletitle{Session-Based Compilation Framework for Multicore
  Programming}. In \bibinfo{booktitle}{\emph{Formal Methods for Components and
  Objects, 7th International Symposium, {FMCO} 2008, Sophia Antipolis, France,
  October 21-23, 2008, Revised Lectures}} \emph{(\bibinfo{series}{Lecture Notes
  in Computer Science})}, \bibfield{editor}{\bibinfo{person}{Frank~S. de~Boer},
  \bibinfo{person}{Marcello~M. Bonsangue}, {and} \bibinfo{person}{Eric
  Madelaine}} (Eds.), Vol.~\bibinfo{volume}{5751}.
  \bibinfo{publisher}{Springer}, \bibinfo{pages}{226--246}.
\newblock
\urldef\tempurl%
\url{https://doi.org/10.1007/978-3-642-04167-9\_12}
\showDOI{\tempurl}


\bibitem[\protect\citeauthoryear{Zhou, Ferreira, Hu, Neykova, and Yoshida}{Zhou
  et~al\mbox{.}}{2020}]%
        {OOPSLA20FStar}
\bibfield{author}{\bibinfo{person}{Fangyi Zhou}, \bibinfo{person}{Francisco
  Ferreira}, \bibinfo{person}{Raymond Hu}, \bibinfo{person}{Rumyana Neykova},
  {and} \bibinfo{person}{Nobuko Yoshida}.} \bibinfo{year}{2020}\natexlab{}.
\newblock \bibinfo{title}{{Statically Verified Refinements for Multiparty
  Protocols}}.  (\bibinfo{year}{2020}).
\newblock
\newblock
\shownote{Conditionally Accepted by OOPSLA '20, Preprint on
  \url{https://www.doc.ic.ac.uk/~fz315/oopsla20-preprint.pdf}.}


\end{thebibliography}

\pagebreak
\appendix
\section{Proof of Lemmas \ref{lm:step_wf} and \ref{lm:cost_preservation}}
\label{app:global-cost-sound}

This proof relies on the following definitions and lemmas.

\queueWf*

\begin{proof}
  By induction on the structure of $G \xrightarrow{\ell} G'$.

  \noindent
  \underline{\textbf{Case} [GR1a]}:
  $p \gMsg  \Rq \gTy{\tau \hasCost \ccc} . G
  \xrightarrow{\Rp\Rq\tSend \tau}
  \Rp \gMsgt \Rq \gTy{\tau \hasCost \ccc} . G$.

  By $\WF(\Rp \gMsg \Rq \gTy{\tau \hasCost \ccc} . G, (\Time, \Queue))$, $\Queue =
  \Queue[\Rp\Rq \mapsto \epsilon]$, and $\WF(G, \Queue[\Rp\Rq \mapsto \epsilon])|$. By
  Def.\ \ref{def:wf_dep_queue}, we have $\WF(\Rp \gMsgt \Rq \gTy{\tau \hasCost \ccc} . G,
  \Queue[\Rp\Rq \mapsto \epsilon \cdot \Time(\Rp) + \csend(\tau)])$. This implies $\WF(\Rp
  \gMsgt \Rq \gTy{\tau \hasCost \ccc} . G, \Queue[\Rp\Rq \mapsto \Time(\Rp) + \csend(\tau)
  \cdot \epsilon])$. By Def. \ref{def:cost_trace}, $\cost(\Rp\Rq\tSend\tau) (\Time,
  \Queue[\Rp\Rq\mapsto\epsilon]) = (\Time[\Rp\addcost \csend(\tau)], \Queue[\Rp\Rq \mapsto
  \epsilon \cdot \Time(\Rp) + \csend(\tau)]) = (\Time[\Rp\addcost \csend(\tau)],
  \Queue[\Rp\Rq \mapsto \Time(\Rp) + \csend(\tau) \cdot \epsilon])$.
  Therefore, $\WF(\Rp \gMsgt \Rq \gTy{\tau \hasCost \ccc} . G, \cost(\Rp\Rq\tSend
  \tau)(\Time, \Queue))$.

  \noindent
  \underline{\textbf{Case} [GR1b]}:
  $\Rp \gMsg  \Rq \{\aInj{i}. G_i\}_{i \in I}
  \xrightarrow{\Rp\Rq\tSelect j}
  \Rp \gMsgt \Rq \; j \; \{\aInj{i}. G_i\}_{i \in I} \quad (j \in I)
  $

  By $\WF(\Rp \gMsg  \Rq \{\aInj{i}. G_i\}_{i \in I}, (\Time, \Queue))$,
  $\Queue=\Queue[\Rp\Rq\mapsto \epsilon]$,
  and
  $\WF(G_i, \Queue[\Rp\Rq\mapsto \epsilon])$.
  By Def.\ \ref{def:wf_dep_queue}
  and
  $\ccc \cdot \epsilon = \epsilon \cdot \ccc = \ccc$,
  we have
  $\WF(\Rp \gMsgt \Rq \; j \; \{\aInj{i}. G_i\}_{i \in I},
  \Queue[\Rp\Rq\mapsto \Time(\Rp) + \csend(1) \cdot \epsilon])$.
  By Def. \ref{def:cost_trace},
  $\cost(\Rp\Rq\tSelect l_j)(\Time, \Queue[\Rp\Rq \mapsto \epsilon])
  = (\Time[\Rp\addcost \csend(1)], \Queue[\Rp\Rq \mapsto \Time(\Rp) + \csend(1) \cdot \epsilon])$.
  Therefore,
  $\WF(\Rp \gMsgt \Rq \; j \; \{\aInj{i}. G_i\}_{i \in I},
  \cost(\Rp\Rq\tSelect l_j)(\Time, \Queue))$.

  \noindent
  \underline{\textbf{Case} [GR2a]}:
  $\Rp \gMsgt \Rq \gTy{\tau \hasCost \ccc} . G
  \xrightarrow{\Rp\Rq\tRecv \tau}
  \Rq \gEval (\tau \hasCost \ccc) . G$

  By $\WF(\Rp \gMsgt  \Rq \gTy{\tau \hasCost \ccc}. G, (\Time, \Queue))$,
  we have
  $Queue = \Queue[\Rp\Rq \mapsto \ccc_\Rp \cdot w]$, and
  $\WF(G, (\Time, \Queue[\Rp\Rq \mapsto w]))$.
  By Def.\ \ref{def:wf_dep_queue},
  $\WF(\Rq \gEval (\tau \hasCost \ccc) . G, (\Time[\ccc_\Rp| \Rq\addcost \crecv(\tau) + \ccc], \Queue[\Rp\Rq \mapsto w]))$.
  By Def.\ \ref{def:cost_trace},
  $\WF(\Rp \gMsgt \Rq \gTy{\tau \hasCost \ccc} . G, \cost(\Rp\Rq\tRecv \tau)(\Time, \Queue[\Rp\Rq \mapsto \ccc_\Rp \cdot w]))$.

  \noindent
  \underline{\textbf{Case} [GR2b]}:
  $\Rq \gEval (\tau \hasCost \ccc) . G
  \xrightarrow{q\tRun c}
  G$

  By $\WF(\Rq \gEval (\tau \hasCost \ccc) . G, (\Time, \Queue))$,
  we have
  $\WF(G, (\Time, \Queue))$.

  \noindent
  \underline{\textbf{Case} [GR2c]}:
  $\Rp \gMsgt \Rq \; : j \; \{\aInj{i}. G_i\}_{i \in I}
  \xrightarrow{\Rp\Rq\tBranch l_j}
  G_j$

  By $\WF(\Rp \gMsgt \Rq \; : j \; \{\aInj{i}. G_i\}_{i \in I}, (\Time, \Queue))$,
  we have
  $\Queue = \Queue[\Rp\Rq \mapsto \ccc_\Rp \cdot w]$,
  and
  $\WF(G_j, (\Time[\ccc_\Rp| \Rq \addcost \crecv(1)], \Queue[\Rp\Rq \mapsto w]))$.
  By Def.\ \ref{def:cost_trace},
  $\WF(G_j, \cost(\Rp\Rq \tBranch l_j)(\Time, \Queue[\Rp\Rq \mapsto \ccc_\Rp \cdot w]))$.

  \noindent
  \underline{\textbf{Case} [GR3]}:
  $G[\gFix X. G/X] \xrightarrow{\act} G'
  \Longrightarrow \gFix X. G \xrightarrow{\act} G'$

  Impossible, since $\neg \WF(\gFix X, \Queue)$.

  \noindent
  \underline{\textbf{Case} [GR4a]}:
  $ G \xrightarrow{\act} G' \wedge \Rp, \Rq \not\in \subj(\act)
  \Longrightarrow
  \Rp \gMsg \Rq \gTy{\tau \hasCost \ccc} . G
  \; \xrightarrow{\act} \;
  \Rp \gMsg \Rq \gTy{\tau \hasCost \ccc} . G'$

  By $\WF(\Rp \gMsg \Rq \gTy{\tau \hasCost \ccc} . G, (\Time, \Queue)$,
  we have
  $\Queue = \Queue[\Rp\Rq \mapsto \epsilon]$,
  and
  $\WF(G, (\Time, \Queue[\Rp\Rq \mapsto \epsilon]))$.
  By the IH,
  $\WF(G', \cost(\act)(\Time, \Queue[\Rp\Rq \mapsto \epsilon]))$.
  Therefore, by Def\ \ref{def:wf_dep_queue} and Def.\ \ref{def:cost_trace},
  $\WF(\Rp \gMsg \Rq \gTy{\tau \hasCost \ccc} . G', \cost(\act)(\Time, \Queue[\Rp\Rq \mapsto \epsilon]))$.

  \noindent
  \underline{\textbf{Case} [GR4b]}:
  $ \forall i \in I,
  G_i \xrightarrow{\act} G_i'
  \wedge
  (\Rp, \Rq \not\in \subj(\act))
  \Longrightarrow
  \Rp \gMsg \Rq \{\aInj{i}. \; G_i\}_{i \in I}
  \; \xrightarrow{\act} \;
  \Rp \gMsg \Rq \{\aInj{i}. \; G_i'\}_{i \in I} $

  By $\WF(\Rp \gMsg \Rq \{\aInj{i}. \; G_i\}_{i \in I}, (\Time, \Queue)$,
  we have
  $\Queue = \Queue[\Rp\Rq \mapsto \epsilon]$,
  and
  $\WF(G_i, (\Time, \Queue[\Rp\Rq \mapsto \epsilon]))$,
  for all
  $i \in I$.
  By the IH,
  $\WF(G_i', \cost(\act)(\Time, \Queue[\Rp\Rq \mapsto \epsilon]))$.
  Therefore, by Def\ \ref{def:wf_dep_queue} and Def.\ \ref{def:cost_trace},
  $\WF(\Rp \gMsg \Rq \{\aInj{i}. \; G_i\}_{i \in I}, \cost(\act) (\Time ,
  \\ \Queue[\Rp\Rq \mapsto \epsilon]))$.

  \noindent
  \underline{\textbf{Case} [GR5a]}:
  $ G \xrightarrow{\act} G'
  \wedge
  \Rq \not\in \subj(\act)
  \Longrightarrow
  \Rp \gMsgt \Rq \gTy{\tau \hasCost \ccc} . G
  \; \xrightarrow{\act} \;
  \Rp \gMsgt \Rq \gTy{\tau \hasCost \ccc} . G'$

  By $\WF(\Rp \gMsgt \Rq \gTy{\tau \hasCost \ccc} . G, (\Time, \Queue)$,
  we have
  $\Queue = \Queue[\Rp\Rq \mapsto \ccc_\Rp \cdot w]$,
  and
  $\WF(G, (\Time, \Queue[\Rp\Rq \mapsto w]))$.
  By the IH,
  $\WF(G', \cost(\act)(\Time, \Queue[\Rp\Rq \mapsto w]))$.
  Therefore, by Def\ \ref{def:wf_dep_queue} and Def.\ \ref{def:cost_trace},
  $\WF(\Rp \gMsgt \Rq \gTy{\tau \hasCost \ccc} . G', \cost(\act)(\Time, \Queue[\Rp\Rq \mapsto \ccc_\Rp \cdot w]))$.

  \noindent
  \underline{\textbf{Case} [GR5b]}:
  $ G \xrightarrow{\act} G'
  \wedge
  \Rp \not\in \subj(\act)
  \Longrightarrow
  \Rp \gEval (\tau \hasCost \ccc) . G
  \; \xrightarrow{\act} \;
  \Rp \gEval (\tau \hasCost \ccc) . G' $

  By $\WF(\Rp \gEval (\tau \hasCost \ccc) . G, (\Time, \Queue)$,
  we have
  $\WF(G, (\Time, \Queue))$.
  By the IH,
  $\WF(G', \cost(\tau)(\Time, \Queue))$.
  Therefore, by Def\ \ref{def:wf_dep_queue} and Def.\ \ref{def:cost_trace},
  By $\WF(\Rp \gEval (\tau \hasCost \ccc) . G', \cost(\tau)(\Time, \Queue))$.

  \noindent
  \underline{\textbf{Case} [GR5c]}:
  $\begin{array}[t]{@{}l@{}}
     G_j \xrightarrow{\act} G_j'
     \wedge
     (\Rq \not\in \subj(\act))
     \wedge
     (\forall i \in I \setminus j, G_i = G_i')
     \Longrightarrow \\
     \hspace{1cm}
     \Rp \gMsgt \Rq \; j \; \{\aInj{i}. \; G_i \}_{i \in I}
     \; \xrightarrow{\act} \;
     \Rp \gMsgt \Rq \; j \; \{\aInj{i}. \; G_i' \}_{i \in I}
   \end{array}$

   By $\WF(\Rp \gMsgt \Rq \; j \; \{\aInj{i}. \; G_i\}_{i \in I}, (\Time, \Queue)$,
   we have
   $\Queue = \Queue[\Rp\Rq \mapsto \ccc_\Rp \cdot w]$,
   and
   $\WF(G_j, (\Time, \Queue[\Rp\Rq \mapsto w]))$.
   By the IH,
   $\WF(G_j', \cost(\act)(\Time, \Queue[\Rp\Rq \mapsto w]))$.
   Therefore, by Def\ \ref{def:wf_dep_queue} and Def.\ \ref{def:cost_trace},
   $\WF(\Rp \gMsgt \Rq \; j \; \{\aInj{i}. \; G_i'\}_{i \in I}, \cost(\act)(\Time, \Queue[\Rp\Rq \mapsto \ccc_\Rp \cdot w]))$.
\end{proof}

\begin{lemma}\label{lm:action_queue}
  If
  $\Rq \neq \subj(\ell)$,
  and $\cost(\ell)(\Time, \Queue[\Rp\Rq \mapsto w]) = (\Time', \Queue')$,
  then $\exists w'$ s.t. $\Queue'(\Rp\Rq) = w \cdot w'$.
\end{lemma}

\begin{proof}
  By case analysis on $\ell$. We consider actions $\Rp\Rq$, and the queue elements:
  $[\Rr_1\Rr_2 \mapsto w]$.

  \noindent
  \underline{\textbf{Case} $\Rp\Rq\tRecv \tau$}.

  Since $\Rr_2 \neq \subj(\Rp\Rq\tRecv \tau)$, $q \neq \Rr_2$.

   $(\cost(\Rp\Rq\tRecv \tau) (\Time, \Queue[\Rp\Rq\mapsto \ccc_\Rp \cdot w_p][\Rr_1\Rr_2 \mapsto w]))
   = (\Time[\ccc_\Rp | \Rq \addcost \crecv(\tau)], \Queue[\Rp\Rq \mapsto w_p][\Rr_1\Rr_2 \mapsto w])
   $. Therefore, there exists $\epsilon$ s.t. $w = w \cdot \epsilon$.

  \noindent
  \underline{\textbf{Case} $\Rp\Rq\tBranch l_k$}.

  Since $\Rr_2 \neq \subj(\Rp\Rq\tBranch l_k)$, $\Rq \neq \Rr_2$.

   $(\cost(\Rp\Rq\tBranch l_k) (\Time, \Queue[\Rp\Rq\mapsto \ccc_\Rp \cdot
   w_\Rp][\Rr_1\Rr_2 \mapsto w])) = (\Time[\ccc_\Rp | \Rq \addcost
   \crecv(1)], \Queue[\Rp\Rq \mapsto w_\Rp][\Rr_1\Rr_2 \mapsto w]) $.
   Therefore, there exists $\epsilon$ s.t. $w = w \cdot \epsilon$.

  \noindent
  \underline{\textbf{Case} $\Rp\Rq\tSend \tau$}.

  Since $\Rr_2 \neq \subj(\Rp\Rq\tSend \tau)$, $\Rp \neq \Rr_2$. Consider two
  cases: if $\Rr_1 = \Rp$ and $\Rr_2 = \Rq$, then $\cost(\Rp\Rq\tSend \tau)
  (\Time, \Queue[\Rp\Rq \mapsto w]) = (\Time[\Rp\addcost \csend(\tau)],
  \Queue[\Rp\Rq \mapsto w \cdot \Time(\Rp) + \csend(\tau)]) $, and so there
  exists the singleton sequence $\Time(\Rp) + \csend(\tau)$ with
  $[\Rp\Rq\mapsto w \cdot \Time(\Rp) + \csend(\tau)]$. Otherwise,
  $[\Rr_1\Rr_2\mapsto w']$ cannot be $[\Rp\Rq\mapsto w]$, and we have
  $[\Rr_1\Rr_2 \mapsto w' \cdot \epsilon]$.

   \noindent
   \underline{\textbf{Case} $\Rp\Rq\tSelect l_k$}.

  Since $\Rr_2 \neq \subj(\Rp\Rq\tSelect l_k)$, $\Rp \neq \Rr_2$. Consider
  two cases: if $\Rr_1 = \Rp$ and $\Rr_2 = \Rq$, then $\cost(\Rp\Rq\tSelect
  l_k) (\Time, \Queue[\Rp\Rq \mapsto w]) = (\Time[\Rp\addcost \csend(1)],
  \Queue[\Rp\Rq \mapsto w \cdot \Time(\Rp) + \csend(1)]) $, and so there
  exists the singleton sequence $\Time(\Rp) + \csend(1)$ with $[\Rp\Rq\mapsto
  w \cdot \Time(\Rp) + \csend(1)]$. Otherwise, $[\Rr_1\Rr_2\mapsto w']$
  cannot be $[\Rp\Rq\mapsto w]$, and we have $[\Rr_1\Rr_2 \mapsto w' \cdot
  \epsilon]$.

   \noindent
   \underline{\textbf{Case} $\Rp \tRun \ccc$}.

   $\cost(\Rp \tRun \ccc) (\Time, \Queue[\Rr_1\Rr_2\mapsto w])
   = (\Time[\Rp\addcost \ccc], \Queue[\Rr_1\Rr_2\mapsto w])
   $, and so we have $\Queue[\Rr_1\Rr_2 \mapsto w \cdot \epsilon]$.

\end{proof}

For the next lemma, we define $(\Time, \Queue)[\Rp\addcost \ccc] = (\Time[\Rp\addcost
\ccc], \Queue)$ as notation to update the cost environment in a pair of cost
environment and dependency queue.

\begin{lemma}\label{lm:action_cost}
  If
  $\Rp, \Rq \notin \subj(\ell)$,
  then
  $(\cost(\ell)(\Time, \Queue))[\Rq|\Rp\addcost \ccc] =
  \cost(\ell)(\Time[\Rq|\Rp\addcost \ccc], \Queue)$.

\end{lemma}

\begin{proof}
  By case analysis on $\ell$. We assume actions with subject $\Rp \neq \Rr_1, \Rr_2$:

  \noindent
  \underline{\textbf{Case} $\Rp\Rq\tRecv \tau$}.

   $(\cost(\Rp\Rq\tRecv \tau) (\Time, \Queue[\Rp\Rq \mapsto \ccc_\Rp \cdot w]))[\Rr_1|\Rr_2\addcost \ccc]
   = (\Time[\ccc_\Rp | \Rq \addcost \crecv(\tau)][\Rr_1|\Rr_2\addcost \ccc], \Queue[\Rp\Rq \mapsto w])
   = (\Time[\Rr_1|\Rr_2 \addcost \ccc][\ccc_\Rp | \Rq \addcost \crecv(\tau)], \Queue[\Rp\Rq \mapsto w])
   = \cost(\Rp\Rq\tRecv \tau) (\Time[\Rr_1|\Rr_2\addcost \ccc], \Queue[\Rp\Rq \mapsto \ccc \cdot w]))
   $

  \noindent
  \underline{\textbf{Case} $\Rp\Rq\tBranch l_k$}.

   $(\cost(\Rp\Rq\tBranch l_k) (\Time, \Queue[\Rp\Rq \mapsto \ccc_\Rp \cdot w]))[\Rr_1|\Rr_2 \addcost \ccc]
   = (\Time[\ccc_\Rp | \Rq \addcost \crecv(1)][\Rr_1|\Rr_2 \addcost \ccc], \Queue[\Rp\Rq \mapsto w])
   = (\Time[\Rr_1|\Rr_2 \addcost \ccc][\ccc_\Rp | \Rq \addcost \crecv(1)], \Queue[\Rp\Rq \mapsto w])
   = \cost(\Rp\Rq\tBranch l_k) (\Time[\Rr_1| \Rr_2 \addcost \ccc], \Queue[\Rp\Rq \mapsto \ccc \cdot w]))
   $

  \noindent
  \underline{\textbf{Case} $\Rp\Rq\tSend \tau$}.

   $(\cost(\Rp\Rq\tSend \tau) (\Time, \Queue[\Rp\Rq \mapsto w]))[\Rr_1|\Rr_2 \addcost \ccc]
   = (\Time[\Rp\addcost \csend(\tau)][\Rr_1|\Rr_2\addcost \ccc], \Queue[\Rp\Rq \mapsto w \cdot \Time(\Rp) + \csend(\tau)])
   = (\Time[\Rr_1|\Rr_2 \addcost \ccc][\Rp\addcost \csend(\tau)], \Queue[\Rp\Rq \mapsto w \cdot \Time(\Rp) + \csend(\tau)])
   = \cost(\Rp\Rq\tSend \tau) (\Time[\Rr_1|\Rr_2 \addcost \ccc], \Queue[\Rp\Rq \mapsto w]))
   $

   \noindent
   \underline{\textbf{Case} $\Rp\Rq\tSelect l_k$}.

   $(\cost(\Rp\Rq\tSelect l_k) (\Time, \Queue[\Rp\Rq \mapsto w]))[\Rr_1|\Rr_2 \addcost \ccc]
   = (\Time[\Rp\addcost \csend(1)][\Rr_1|\Rr_2\addcost \ccc], \Queue[\Rp\Rq \mapsto w \cdot \Time(\Rp) + \csend(1)])
   = (\Time[\Rr_1|\Rr_2 \addcost \ccc][\Rp\addcost \csend(1)], \Queue[\Rp\Rq \mapsto w \cdot \Time(\Rp) + \csend(1)])
   = \cost(\Rp\Rq\tSelect l_k) (\Time[\Rr_1|\Rr_2 \addcost \ccc], \Queue[\Rp\Rq \mapsto w]))
   $

   \noindent
   \underline{\textbf{Case} $p \tRun c$}.

   $(\cost(\Rp \tRun \ccc) (\Time, \Queue))[\Rr_1|\Rr_2 \addcost \ccc]
   = (\Time[\Rp\addcost \ccc][\Rr_1|\Rr_2 \addcost \ccc], \Queue)
   = (\Time[\Rr_1|\Rr_2 \addcost \ccc][\Rp\addcost \ccc], \Queue)
    = \cost(\Rp \tRun c) (\Time[\Rr_1|\Rr_2 \addcost \ccc], \Queue)
   $
\end{proof}

\costPreserv*

\begin{proof}
  By induction on the structure of the derivation of $G \xrightarrow{\ell} G'$:

  \noindent
  \underline{\textbf{Case} [GR1a]}:
  $p \gMsg  \Rq \gTy{\tau \hasCost \ccc} . G
  \xrightarrow{\Rp\Rq\tSend \tau}
  \Rp \gMsgt \Rq \gTy{\tau \hasCost \ccc} . G$.

  By
  $\mathsf{Wf}(\Rp \gMsg  \Rq \gTy{\tau \hasCost \ccc} . G, \Queue)$,
  we have that
  $\Queue = \Queue[\Rp\Rq \mapsto \epsilon]$.

  $\begin{array}{@{}l@{}}
     \cost(\Rp \gMsgt \Rq \gTy{\tau \hasCost \ccc} . G)
             (\cost(\Rp\Rq\tSend \tau)(\Time, \Queue))
     \\ =
     \cost(\Rp \gMsgt \Rq \gTy{\tau \hasCost \ccc} . G)
             (\Time[\Rp\addcost \csend(\tau)],
              \Queue[\Rp\Rq \mapsto \Time(\Rp) + \csend(\tau)])
     \\ =
     \cost(G) (\Time[\Rp\addcost \csend(\tau)]
                    [\Time(\Rp) + \csend(\tau)|\Rq\addcost \crecv(\tau) + \ccc],
               \Queue)
     \\ =
     \cost(G) (\Time[\Rp\addcost \csend(\tau)][\Rp| \Rq \addcost \crecv(\tau) + \ccc],
               \Queue)
     \\ =
     \cost(\Rp\gMsg \Rq \gTy{\tau\hasCost \ccc}. G)(\Time, \Queue)
   \end{array}$.

   \noindent
   \underline{\textbf{Case} [GR1b]}:
   $\Rp\gMsg  \Rq \{\aInj{i}. G_i\}_{i \in I}
   \xrightarrow{\Rp\Rq\tSelect j}
   \Rp \gMsgt \Rq \; : j \; \{\aInj{i}. G_i\}_{i \in I} \quad (j \in I)$

   By
   $\mathsf{Wf(\Rp\gMsg  \Rq \{\aInj{i}. G_i\}_{i \in I}, \Queue)}$,
   we know that
   $\Queue = \Queue[\Rp\Rq \mapsto \epsilon]$.

   $\begin{array}{@{}l@{}}
      \cost(\Rp \gMsgt \Rq \; : j \; \{\aInj{i}. G_i\}_{i \in I})
                (\cost(\Rp\Rq\tSelect j) (\Time, \Queue))
      \\ =
      \cost(\Rp \gMsgt \Rq \; : j \; \{\aInj{i}. G_i\}_{i \in I})
                (\Time[\Rp\addcost \csend(1)],
                 \Queue[\Rp\Rq \mapsto \Time(\Rp) + \csend(1)])
      \\ =
      \cost(G_j) (\Time[\Rp\addcost \csend(1)]
                      [\Time(\Rp) + \csend(1)|\Rq\addcost \crecv(1)],
                 \Queue[\Rp\Rq \mapsto \epsilon])
      \\ =
      \cost(G_j) (\Time[\Rp\addcost \csend(1)][\Rp| \Rq\addcost \crecv(1)], \Queue)
      \\ \leq
      \max \{\cost(G_i) (\Time[\Rp\addcost \csend(1)]
                             [\Rp| \Rq\addcost \crecv(1)], \Queue)]\}_{i \in I}
      \\ =
      \cost(\Rp\gMsg  \Rq \{\aInj{i}. G_i\}_{i \in I})(\Time, \Queue)
    \end{array}$

    \noindent
    \underline{\textbf{Case} [GR2a]}:
    $\Rp \gMsgt \Rq \gTy{\tau \hasCost \ccc} . G
    \xrightarrow{\Rp\Rq\tRecv \tau}
    \Rq \gEval (\tau \hasCost \ccc) . G$

    By
    $\mathsf{Wf(\Rp \gMsgt \Rq \gTy{\tau \hasCost \ccc} . G, \Queue)}$,
    we know that
    $\Queue = \Queue[\Rp\Rq \mapsto \ccc_\Rp \cdot w]$.

    $\begin{array}{@{}l@{}}
       \cost(\Rq \gEval (\tau \hasCost \ccc) . G)
              (\cost(\Rp\Rq\tRecv \tau)(\Time, \Queue[\Rp\Rq \mapsto \ccc_\Rp \cdot w]))
       \\ =
       \cost(\Rq \gEval (\tau \hasCost \ccc) . G)
              (\Time[\ccc_\Rp| \Rq \addcost \crecv(\tau)], \Queue[\Rp\Rq \mapsto w])
       \\ =
       \cost(G)
              (\Time[\ccc_\Rp| \Rq \addcost \crecv(\tau)][\Rq \addcost \ccc],
               \Queue[\Rp\Rq \mapsto w])
       \\ =
       \cost(G)
              (\Time[\ccc_\Rp| \Rq \addcost \crecv(\tau) + \ccc],
               \Queue[\Rp\Rq \mapsto w])
       \\ =
       \cost(\Rp \gMsgt \Rq \gTy{\tau \hasCost \ccc} . G)
                (\Time, \Queue[\Rp\Rq \mapsto \ccc_\Rp \cdot w])
     \end{array}$

     \noindent
     \underline{\textbf{Case} [GR2b]}:
     $\Rq \gEval (\tau \hasCost \ccc) . G
     \xrightarrow{\Rq\tRun \ccc}
     G$

     $\begin{array}{@{}l@{}}
        \cost(G)(\cost(\Rq\tRun \ccc)(\Time, \Queue))
        \\ =
        \cost(G)(\Time[\Rq \addcost \ccc], \Queue)
        \\ =
        \cost(\Rq \gEval (\tau \hasCost \ccc) . G)(\Time, \Queue)
     \end{array}$

     \noindent
     \underline{\textbf{Case} [GR2c]}:
     $\Rp \gMsgt \Rq \; : j \; \{\aInj{i}. G_i\}_{i \in I}
     \xrightarrow{\Rp\Rq\tBranch}
     G_j$

     By
     $\mathsf{Wf(\Rp \gMsgt \Rq \; j \; \{\aInj{i}. G_i\}_{i \in I}, \Queue)}$,
     we know that
     $\Queue = \Queue[\Rp\Rq \mapsto \ccc_\Rp \cdot w]$.

     $\begin{array}{@{}l@{}}
        \cost(G_j) (\cost(\Rp\Rq\tBranch)(\Time, \Queue[\Rp\Rq\mapsto \ccc_\Rp \cdot w]))
         \\ =
        \cost(G_j)(\Time[\ccc_\Rp | \Rq\addcost \crecv(1)], \Queue[\Rp\Rq\mapsto w])
        \\ =
        \cost(\Rp \gMsgt \Rq \; j \; \{\aInj{i}. G_i\}_{i \in I})
              (\Time[\Rq\addcost \crecv(1)], \Queue[\Rp\Rq\mapsto \ccc_\Rp \cdot w])
     \end{array}$

     \noindent
     \underline{\textbf{Case} [GR3]}:
     $G[\gFix X. G/X] \xrightarrow{\act} G'
     \Longrightarrow \gFix X. G \xrightarrow{\act} G'$

     Impossible, since $\neg \WF(\gFix X. G, \Queue)$.

     \noindent
     \underline{\textbf{Case} [GR4a]}:
     $ G \xrightarrow{\act} G' \wedge \Rp, \Rq \not\in \subj(\act)
     \Longrightarrow
     \Rp \gMsg \Rq \gTy{\tau \hasCost \ccc} . G
     \; \xrightarrow{\act} \;
     \Rp \gMsg \Rq \gTy{\tau \hasCost \ccc} . G'$

     By $\WF(\Rp \gMsg \Rq \gTy{\tau \hasCost \ccc} . G, \Queue)$,
     we know that
     $\Queue=\Queue[\Rp\Rq\mapsto \epsilon]$,
     and that
     $\WF(G, \Queue)$.

     $\begin{array}{@{}l@{\hspace{1cm}}r}
        \cost(\Rp \gMsg \Rq \gTy{\tau \hasCost \ccc} . G') (\cost(\act)(\Time, \Queue))
        \\ =
        \cost(G', \cost(\act)(\Time, \Queue)[\Rp\addcost \csend(\tau)]
                                            [\Rp | \Rq \addcost \crecv(\tau) + \ccc])
        & \{ \Rp, \Rq \not\in \subj(\act) z\; \text{and Lemma \ref{lm:action_cost}}\}
        \\ =
        \cost(G')(\cost(\act)(\Time[\Rp\addcost \csend(\tau)]
                                   [\Rp | \Rq \addcost \crecv(\tau) + \ccc],
                              \Queue))
        & \text{\{by the IH and \WF(G, \Queue) \}}
        \\ \leq
        \cost(G) (\Time[\Rp\addcost \csend(\tau)][\Rp | \Rq \addcost \crecv(\tau) + \ccc], \Queue)
        \\ =
        \cost(\Rp \gMsg \Rq \gTy{\tau \hasCost \ccc} . G)(\Time, \Queue)
     \end{array}$

     \noindent
     \underline{\textbf{Case} [GR4b]}:
     $ \forall i \in I,
     G_i \xrightarrow{\act} G_i'
     \wedge
     (\Rp, \Rq \not\in \subj(\act))
     \Longrightarrow
     \Rp \gMsg \Rq \{\aInj{i}. \; G_i\}_{i \in I}
     \; \xrightarrow{\act} \;
     \Rp \gMsg \Rq \{\aInj{i}. \; G_i'\}_{i \in I} $

     By $\WF(\Rp \gMsg \Rq \{\aInj{i}. \; G_i\}_{i \in I}, \Queue)$,
     we know that
     $\Queue=\Queue[\Rp\Rq\mapsto \epsilon]$,
     and that
     $\forall i \in I, \WF(G_i, \Queue)$.

     $\begin{array}{@{}l@{\hspace{1cm}}r}
        \cost(\Rp \gMsg \Rq \{\aInj{i}. \; G_i'\}_{i \in I})(\cost(\act)(\Time, \Queue))
        \\ =
        \max\{\cost(G_i') ((\cost(\act)(\Time, \Queue))
                               [\Rp\addcost \csend(1)]
                               [\Rq|\Rq \addcost \crecv(1)])\}_{i\in I}
        & \{ \Rp, \Rq \not\in \subj(\act) \text{and Lemma \ref{lm:action_cost}}\}
        \\ =
        \max\{ \cost(G_i')
                  (\cost(\act)(\Time[\Rp\addcost \csend(1)]
                                    [\Rq|\Rq \addcost \crecv(1)], \Queue)) \}_{i\in I}
        & \{\text{by the IH and} \; \WF(G_i, \Queue)\}
        \\ \leq
        \max\{ \cost(G_i)
                  (\Time[\Rp\addcost \csend(1)]
                        [\Rq|\Rq \addcost \crecv(1)], \Queue) \}_{i\in I}
        \\ =
        \cost(\Rp \gMsg \Rq \{\aInj{i}. \; G_i\}_{i \in I})(\Time, \Queue)
     \end{array}$

     \noindent
     \underline{\textbf{Case} [GR5a]}:
     $ G \xrightarrow{\act} G'
     \wedge
     \Rq \not\in \subj(\act)
     \Longrightarrow
     \Rp \gMsgt \Rq \gTy{\tau \hasCost \ccc} . G
     \; \xrightarrow{\act} \;
     \Rp \gMsgt \Rq \gTy{\tau \hasCost \ccc} . G'$

     By $\WF(\Rp \gMsgt \Rq \gTy{\tau \hasCost \ccc} . G, \Queue)$,
     we know that
     $\Queue=\Queue[\Rp\Rq\mapsto \ccc_\Rp \cdot w]$,
     and that
     $\WF(G, \Queue[\Rp\Rq \mapsto w])$.

     $\begin{array}{@{}l@{\hspace{1cm}}r}
        \cost(\Rp \gMsgt \Rq \gTy{\tau \hasCost \ccc} . G')
             (\cost(\act)(\Time, \Queue[\Rp\Rq\mapsto \ccc_\Rp \cdot w]))
        & \{\text{By}\; \Rq \notin \subj(\act) \;
            \text{and Lemma \ref{lm:action_queue}}\}
        \\ =
        \cost(G')((\cost(\act)(\Time, \Queue[\Rp\Rq\mapsto w]))[\ccc_\Rp| \Rq \addcost \crecv(\tau) + \ccc])
        & \{ \Rq \not\in \subj(\act) \text{and Lemma \ref{lm:action_cost}} \}
        \\ =
        \cost(G')(\cost(\act)(\Time[\ccc_\Rp| \Rq \addcost \crecv(\tau) + \ccc], \Queue[\Rp\Rq\mapsto w]))
        & \{ \text{by the IH and} \; \Wf(G, \Queue[\Rp\Rq\mapsto w])\}
        \\ \leq
        \cost(G)(\Time[\ccc_\Rp| \Rq \addcost \crecv(\tau) + \ccc], \Queue[\Rp\Rq\mapsto w])
        \\ =
        \cost(\Rp \gMsgt \Rq \gTy{\tau \hasCost \ccc} . G)(\Time, \Queue[\Rp\Rq\mapsto \ccc_\Rp \cdot w]))
     \end{array}$

     \noindent
     \underline{\textbf{Case} [GR5b]}:
     $ G \xrightarrow{\act} G'
     \wedge
     \Rp \not\in \subj(\act)
     \Longrightarrow
     \Rp \gEval (\tau \hasCost \ccc) . G
     \; \xrightarrow{\act} \;
     \Rp \gEval (\tau \hasCost \ccc) . G' $

     By $\WF(\Rp\gEval (\tau \hasCost \ccc) . G, \Queue)$,
     we know that
     $\WF(G, \Queue)$.

     $\begin{array}{@{}l@{\hspace{1cm}}r}
        \cost(\Rp\gEval (\tau \hasCost \ccc) . G')(\cost(\act)(\Time, \Queue))
        \\ =
        \cost(G')((\cost(\act)(\Time, \Queue))[\Rp\addcost \ccc])
        & \{\Rp \not\in \subj(\act)\}
        \\ =
        \cost(G')(\cost(\act)(\Time[\Rp\addcost \ccc], \Queue))
        & \{\text{by the IH and} \; \WF(G, \Queue)\}
        \\ \leq
        \cost(G)(\Time[\Rp\addcost \ccc], \Queue)
        \\ =
        \cost(\Rp \gEval (\tau \hasCost \ccc) . G)(\Time, \Queue)
     \end{array}$

     \noindent
     \underline{\textbf{Case} [GR5c]}:
     $\begin{array}[t]{@{}l@{}}
        G_j \xrightarrow{\act} G_j'
        \wedge
        (\Rq \not\in \subj(\act))
        \wedge
        (\forall i \in I \setminus j, G_i = G_i')
        \Longrightarrow \\
        \hspace{1cm}
        \Rp \gMsgt \Rq \; j \; \{\aInj{i}. \; G_i \}_{i \in I}
        \; \xrightarrow{\act} \;
        \Rp \gMsgt \Rq \; j \; \{\aInj{i}. \; G_i' \}_{i \in I}
     \end{array}$

     By $\WF(\Rp \gMsgt \Rq \; j \; \{\aInj{i}. \; G_i \}_{i \in I}, \Queue)$,
     we know that
     $\Queue=\Queue[\Rp\Rq\mapsto \ccc_\Rp \cdot w]$,
     and that \\
     $\WF(G_j, \Queue[\Rp\Rq \mapsto w])$.

     $\begin{array}{@{}l@{\hspace{1cm}}r}
        \cost(\Rp \gMsgt \; j \; \Rq \{\aInj{i}. \; G_i'\}_{i \in I})
            (\cost(\act)(\Time, \Queue[\Rp\Rq\mapsto \ccc_\Rp \cdot w]))
        & \{\text{By} \; \Rq \notin \subj(\act) \; \text{and Lemma \ref{lm:action_queue}}\}
        \\ =
        \cost(G_j')((\cost(\act)(\Time, \Queue[\Rp\Rq\mapsto w]))[c_\Rp | \Rq \addcost \crecv(1)])
        & \{ \Rq \not\in \subj(\act) \; \text{and Lemma \ref{lm:action_cost}}\}
        \\ =
        \cost(G_j')(\cost(\act)(\Time[\ccc_\Rp| \Rq \addcost \crecv(1)], \Queue[\Rp\Rq\mapsto w]))
        & \{\text{by the IH and} \; \WF(G_j, \Queue[\Rp\Rq\mapsto w])\}
        \\ \leq
        \cost(G_j)(\Time[\ccc_\Rp| \Rq \addcost \crecv(1)], \Queue[\Rp\Rq\mapsto w])
        \\ =
        \cost(\Rp \gMsgt \; j \; \Rq \{\aInj{i}. \; G_i\}_{i \in I})
                (\Time, \Queue[\Rp\Rq\mapsto \ccc_\Rp \cdot w])
      \end{array}
      $

\end{proof}

\end{document}